\documentclass[11pt,twoside,english]{article}
\usepackage{amsmath}
\usepackage{amsthm}
\usepackage{amssymb}
\usepackage{colortbl}
\usepackage{array}
\usepackage{mathtools}
\usepackage{varwidth}
\usepackage{cancel}
\usepackage{mathdots}
\usepackage{stackrel}
\usepackage{graphicx}
\usepackage[a4paper]{geometry}
\usepackage{setspace}
\PassOptionsToPackage{version=3}{mhchem}
\usepackage{mhchem}
\usepackage{esint}
\usepackage[bookmarks=true,
 pdftitle={Nonparametric Identification and Estimation of Production Functions Invariant to Productivity Dynamics},
 pdfauthor={Rentaro Utamaru},
 breaklinks=true,pdfborder={0 0 0},pdfborderstyle={},backref=false,colorlinks=true,linkcolor=blue!70!black,citecolor=blue!70!black,urlcolor=blue!70!black]
 {hyperref}

\makeatletter

\providecommand{\tabularnewline}{\\}
\newenvironment{cellvarwidth}[1][t]
    {\begin{varwidth}[#1]{\linewidth}}
    {\@finalstrut\@arstrutbox\end{varwidth}}

\usepackage{makecell}
\usepackage{booktabs}
\usepackage{float}
\usepackage[section]{placeins}
\usepackage{xcolor}
\usepackage{threeparttablex}
\usepackage{siunitx}
\usepackage{dcolumn}
\usepackage{adjustbox}
\usepackage{tikz}
\usetikzlibrary{shapes.geometric, arrows.meta, positioning}
\usepackage{enumitem}
\usepackage{tabularray}
\usepackage{pdflscape}
\usepackage[titletoc]{appendix}

\makeatother

\theoremstyle{plain}
\newtheorem{thm}{\protect\theoremname}
\newtheorem{assumption}{\protect\assumptionname}
\usepackage[american]{babel}
\usepackage[style=authoryear,backref=false, url=false, doi=false, eprint=false,maxcitenames=3,giveninits=true,uniquename=false,sorting=nyt]{biblatex}
\providecommand{\assumptionname}{Assumption}
\providecommand{\theoremname}{Theorem}
\providecommand{\remarkname}{Remark}
\providecommand{\propositionname}{Proposition}
\providecommand{\corollaryname}{Corollary}
\newtheorem{remark}{\protect\remarkname}
\newtheorem{proposition}{\protect\propositionname}
\newtheorem{corollary}{\protect\corollaryname}

\AtEveryBibitem{\clearfield{note}}
\begin{document}
\IfFileExists{Table/Empirical/cross_val_macros.tex}{%
  \newcommand{\NconsGroup}{302}
\newcommand{\NinconsGroup}{200}
\newcommand{\NallBlockC}{502}
\newcommand{\NbothPass}{38}
\newcommand{\PctConsPass}{39.2}
\newcommand{\PctInconsPass}{12.0}
\newcommand{\CorBkCrossVal}{0.75}
}{%
  \newcommand{\NconsGroup}{302}
  \newcommand{\NinconsGroup}{200}
  \newcommand{\NallBlockC}{502}
  \newcommand{\NbothPass}{38}
  \newcommand{\PctConsPass}{39.2}
  \newcommand{\PctInconsPass}{12.0}
  \newcommand{\CorBkCrossVal}{0.75}
}
\title{Nonparametric Identification and Estimation of Production Functions
Invariant to Productivity Dynamics\thanks{I am grateful to Yasutora Watanabe, Yuta Toyama, Shosei Sakaguchi,
Hidehiko Ichimura, and Satoshi Imahie for their insightful comments and detailed discussions.
I also thank Takanori Adachi, Daiya Isogawa, Yuta Kikuchi, Toshifumi Kuroda,
Yusuke Matsuki, Kentaro Nakamura, Masato Nishiwaki, Tatsushi Oka, Ryo Okui,
Naoya Sueishi, Hidenori Takahashi, and Naoki Wakamori for helpful comments,
as well as participants at the Japan Empirical Industrial Organization Workshop
and the Kansai Econometric Society Meeting. This research was financially supported
by the Project Research Program of the Joint Usage/Research Center Programs
at the Institute of Economic Research, Hitotsubashi University
(Grant Number: IERPK2437); the JST SPRING fellowship; and a
Grant-in-Aid for JSPS Fellows (Grant Number: 25KJ0910).
This research was conducted under approval number
20240708-stat-No1 dated July 8, 2024, by the Statistics Bureau,
Ministry of Internal Affairs and Communications. I utilized
microdata from the \textit{Census of
Manufactures} (Ministry of Economy, Trade and Industry) and the
\textit{Economic Census for Business Activity} (Ministry of
Internal Affairs and Communications; Ministry of Economy, Trade
and Industry). The views expressed in this paper are those of the
author and do not necessarily reflect the views of the Japanese
government or the ministries.
All remaining errors are my own.}}
\author{Rentaro Utamaru\thanks{Institute for Research in Contemporary Political and Economic, Waseda University,
1-104 Totsukamachi, Shinjuku-ku, Tokyo 169-8050, Japan. Email: rentaro.utamaru@gmail.com}}
\date{April 2026\\[0.5em]
\normalsize\href{https://r-utamaru.github.io/files/utamaru_nonmarkov_production.pdf}{Click here for the latest version}}
\maketitle

\vspace{-1.5em}
\begin{center}
\fbox{\parbox{0.85\textwidth}{\centering\small
\textbf{Preliminary Draft. Comments Welcome.}\\[0.2em]
The core identification theory and GMM estimator are complete.
Empirical results and Monte Carlo simulations are subject to revision.
Extensions to the GMM implementation of exclusion restrictions and
formal specification testing are in progress.}}
\end{center}
\vspace{-0.5em}
\begin{abstract}
Production function estimates underpin the measurement of firm-level
markups, allocative efficiency, and the productivity effects of
policy interventions. Since \textcite{olley1996thedynamics}, every major
proxy variable estimator has identified the production function
through a first-order Markov assumption on unobserved productivity;
I show that misspecification of this assumption generates persistent
upward bias in the materials elasticity that propagates into
overestimated markups and inflated treatment effects. I replace
the Markov restriction with conditional independence across three
intermediate input demands, a static condition grounded in input
market segmentation, and establish nonparametric identification
from a single cross-section. I develop a GMM estimator and
establish consistency and asymptotic normality. Monte Carlo
simulations confirm that the proposed estimator is unbiased across
Markov and non-Markov environments, while the standard estimator
exhibits persistent bias of up to 63 percent of the true materials
elasticity. In 502 Japanese manufacturing industries, the proposed
method yields systematically lower markups than the standard method
across the entire distribution (median 0.93 vs.\ 1.03), reducing
the share of industries with markups above unity from 54 to 37
percent. In a difference-in-differences analysis of the 2011
T\={o}hoku earthquake, the standard method overstates the
productivity loss by 0.40 percentage points, roughly \$3.6
billion (\textyen{}400 billion) per year.\\

\end{abstract}
\clearpage{}

\section{Introduction}

Estimated production functions underpin the measurement of market
power, allocative efficiency, and the effects of policy on firm
performance. The ratio of the materials elasticity to the materials
revenue share gives the firm-level markup
\parencite{deloecker2012markups}; the dispersion of the productivity
residual measures resource misallocation
\parencite{hsieh2009misallocation}; and the productivity level itself
serves as the outcome variable in studies of trade liberalization
\parencite{deloecker2013detecting}, R\&D investment
\parencite{doraszelski2013rdand}, and disaster recovery. These
downstream analyses inherit the production function estimate: if the
materials elasticity is biased, so is the markup, the misallocation
measure, and the treatment effect. The recent finding that markups
have risen across the global economy \parencite{deloecker2020rise}
relies on such estimates, making the consistency of the underlying
production function a first-order concern.
This paper asks whether the production function can be identified
without restricting how productivity evolves over time, and
documents the consequences when this restriction is removed.

Since \textcite{olley1996thedynamics}, every major production function
estimator has relied on the same structural restriction: productivity
must follow a first-order Markov process. This includes the methods of
\textcite{levinsohn2003estimating},
\textcite{ackerberg2015identification}, and
\textcite{gandhi2020onthe}, as well as dynamic panel approaches
\parencite{arellano1991sometests,blundell1998initial}. The Markov
assumption is not a regularity condition; it is the identifying
restriction that pins down the materials elasticity through the
transition equation. When productivity evolves endogenously through
R\&D, learning, or managerial turnover, omitting the relevant state
variables generates a transmission bias
\parencite{deloecker2007doexports,deloecker2013detecting,doraszelski2013rdand}.
The assumption also presupposes a stationary transition process,
ruling out structural breaks from aggregate shocks, regulatory
shifts, or technological change. More fundamentally,
\textcite{chen2024identifying} show that under the
potential outcomes framework, any treatment that alters the transition
path of productivity violates the Markov property by construction,
even when the treatment variable is included as a control.
The bias does not vanish with sample size, nor can it be removed by
adding treatment indicators to the Markov transition equation. The
Markov-based estimate is therefore inconsistent precisely in the
settings where productivity serves as an outcome variable, the
dominant use of production function estimation in applied work.

This paper shows that the Markov assumption is unnecessary for
identification. I replace it with a static condition: conditional
independence of demand shocks across three intermediate inputs (raw
materials, electricity, water). Three flexible inputs whose demands
respond to the same underlying productivity serve as three noisy
measurements of a common latent variable. Because each input is
procured from a separate market, the input-specific demand
shocks are mutually independent conditional on productivity and
observable controls. I recover the productivity distribution from
these signals using the spectral decomposition of
\textcite{hu2008instrumental} (hereafter HS08), without any
restriction on how productivity evolves over time. Identification
requires only a single cross-section; the data requirement
(firm-level quantities of three separate inputs) is met in
manufacturing censuses across several countries.\footnote{These
include India's Annual Survey of Industries, Canada's Annual Survey
of Manufacturing and Logging, the World Bank Enterprise Survey, and
the U.S.\ EIA Form~923. When labor adjusts rapidly to current
productivity, two intermediate inputs suffice
(footnote~\ref{rem:static_labor}).}

The substitution of assumptions has first-order consequences for
economic measurement. In 502 Japanese manufacturing industries, the
proposed method yields systematically lower markups than the
standard ACF method across the entire distribution: the ACF markup
CDF lies strictly to the right at every percentile. At the median,
the gap is 0.10 (proposed 0.93 vs.\ ACF 1.03), and the share of
industries with markups above unity falls from 54 percent under ACF
to 37 percent under the proposed method. The Markov assumption thus
inflates the measured degree of market power across the
manufacturing sector. Monte Carlo simulations
trace the mechanism: under potential outcome dynamics, ACF's bias in
the materials elasticity is $+0.19$ (63 percent of the true value),
while the proposed estimator is unbiased.

In a difference-in-differences analysis of the 2011 T\={o}hoku
earthquake, the standard method overstates the productivity loss by
0.40 percentage points, corresponding to roughly \$3.6 billion
(\textyen{}400 billion) per year. Because identification is static, the estimator
can be applied period by period, producing time-varying estimates
of production technologies without imposing structural stability on
the productivity process. The empirical application documents
substantial temporal variation across 2003--2020 and yields
divergent conclusions regarding allocative efficiency as assessed
through the \textcite{olley1996thedynamics} decomposition. The
Markov assumption does not merely introduce statistical noise; it
systematically inflates measured market power and distorts policy
conclusions.

The substitution involves an honest tradeoff. The Markov assumption,
when it holds, provides efficiency gains by exploiting the
time-series history of productivity. The conditional independence
assumption uses only within-period information, so under correct
Markov specification, standard estimators have lower variance. I
document this in Monte Carlo simulations under correct Markov
specification. The value of
the proposed method lies in the broad class of applications where the
Markov assumption is questionable or directly contradicted by the
research design, including any study in which a treatment alters
productivity dynamics \parencite{chen2024identifying}.

The two assumptions differ in the nature of their economic content.
The Markov restriction constrains the time-series evolution of
unobserved productivity; no economic theory predicts that
productivity should follow a first-order autoregression, and the
assumption cannot be tested within the proxy variable framework.
The conditional independence restriction constrains input market
structure: it specifies what threatens identification (common
shocks across input markets) and what restores it (conditioning on
observable controls $z_{jt}$ that absorb the common component).
The threats are enumerable (demand fluctuations, aggregate markup
changes, correlated procurement), and the defenses are
observable (inventory, aggregate output, fixed effects;
Section~\ref{subsec:relax_indep}). The microfoundations in
Appendix~\ref{sec:micro_fnd} derive the demand shocks from a
cost-minimization problem with input-specific markdowns, making
the economic content of the assumption precise. No analogous transparency is available for the Markov assumption:
within the proxy variable framework, no observable implication
distinguishes a correctly specified AR(1) from an AR(2) or a
potential outcome process. By contrast, the conditional
independence assumption yields a testable necessary condition
(Remark~\ref{rem:testable_exclusion}): in 502 industries, the
pairwise convergence diagnostic supports the identifying
restriction for capital, providing direct evidence on the
empirical plausibility of the assumption. When the assumption is
violated through a common shock to electricity and water (the
most economically salient threat), Monte Carlo analysis
(Section~\ref{sec:simulation}, Appendix~\ref{app:ci_violation})
shows that the resulting bias in $\hat{\beta}_m$ is
\emph{upward}, the same direction as the Markov misspecification
bias. The empirical finding that the proposed method yields
\emph{lower} $\hat{\beta}_m$ than ACF therefore cannot be
explained by conditional independence violation; it is consistent
only with Markov misspecification in the standard estimator.

\begin{table}[tbph]
\centering
{\footnotesize\caption{{\footnotesize{} Comparison with other studies}}\label{tab:comparison}
}{\footnotesize\par}

\begin{adjustbox}{max width=\linewidth}
\begin{tabular}{lV{\linewidth}ccccc}
\toprule
 & {\small Identification}{\small\par}

{\small Method} & {\small Req.\ Markov} & \begin{cellvarwidth}[t]
\centering
{\small Req.\ Scalar}{\small\par}

{\small Unobs.}
\end{cellvarwidth} & \begin{cellvarwidth}[t]
\centering
{\small Nonpara}{\small\par}

{\small Non-Hicks}
\end{cellvarwidth} & \begin{cellvarwidth}[t]
\centering
{\small Function}{\small\par}

{\small Type}
\end{cellvarwidth} & \begin{cellvarwidth}[t]
\centering
{\small Proxy or}{\small\par}

{\small Control}
\end{cellvarwidth}\tabularnewline
\midrule
{\small Proposed method} & \citeauthor*{hu2008instrumental} &  &  & {\small$\checkmark$} & {\small Gross} & {\small$e_{jt},w_{jt}$}\tabularnewline
{\small\textcite{gandhi2020onthe}} & {\small FOC + Markov} & {\small$\checkmark$} & {\small$\checkmark$} &  & {\small Gross} & {\small$s_{jt}$ (share)}\tabularnewline
{\small\textcite{dotyadynamic}} & \citeauthor*{hu2008instrumental} & {\small$\checkmark$} & {\small$\checkmark$} & {\small$\checkmark$} & {\small Gross} & {\small$I_{jt},y_{jt+1}$}\tabularnewline
{\small\textcite{hu2020estimating}} & \citeauthor*{hu2008instrumental} & {\small$\checkmark$} & {\small$\checkmark$} &  & {\small Gross} & {\small$I_{jt},m_{jt+1}$}\tabularnewline
{\small\textcite{brandestimating}} & \citeauthor*{hu2008instrumental} & {\small$\checkmark$} & {\small$\checkmark$} &  & {\small Gross} & {\small$y_{jt-1},y_{jt+1}$}\tabularnewline
{\small\textcite{zeng2023identification}} & \begin{cellvarwidth}[t]
{\small\citeauthor*{matzkin2003nonparametric}}\\
{\small\citeauthor*{imbens2009identification}}
\end{cellvarwidth} & {\small$\checkmark$} & {\small$\checkmark$} & {\small$\checkmark$} & {\small Value} & {\small$K_{jt-1},I_{jt-1}$}\tabularnewline
{\small\textcite{ackerberg2022nonparametric}} & \begin{cellvarwidth}[t]
{\small\citeauthor*{matzkin2003nonparametric}}\\
{\small\citeauthor*{imbens2009identification}}
\end{cellvarwidth} & {\small$\checkmark$} & {\small$\checkmark$} & {\small$\checkmark$} & {\small Gross} & {\small$\left\{ y_{j\tau},x_{j\tau}\right\} ^{t-1}_{\tau=t-M}$}\tabularnewline
{\small\textcite{navarrononparametric}} & \begin{cellvarwidth}[t]
{\small\citeauthor*{matzkin2003nonparametric}}\\
{\small\citeauthor*{imbens2009identification}}
\end{cellvarwidth} & {\small$\checkmark$} & {\small$\checkmark$} & {\small$\checkmark$} & {\small Gross} & {\small$x_{jt-1},\mathcal{Y}_{jt-1}$}\tabularnewline
{\small\textcite{pan2022identification}} & \begin{cellvarwidth}[t]
{\small\citeauthor*{matzkin2003nonparametric}}\\
{\small\citeauthor*{imbens2009identification}}
\end{cellvarwidth} & {\small$\checkmark$} & {\small$\checkmark$} & {\small$\checkmark$} & {\small Gross} & {\small$\left\{ y_{j\tau},x_{j\tau}\right\} ^{t-1}_{\tau=t-M}$}\tabularnewline
\bottomrule
\end{tabular}
\end{adjustbox}

\medskip
\small\raggedright \textit{Notes: ``Req.\ Markov'' indicates whether the
method requires a Markov assumption on productivity; a blank cell
indicates the method does not.
``Req.\ Scalar Unobs.'' indicates whether the method requires
scalar unobservability (productivity as the sole unobservable
in input demand); a blank cell indicates that the method permits
input-specific demand shocks.
``Nonpara Non-Hicks'' indicates nonparametric identification
under non-Hicks-neutral specifications.
``Function Type'' distinguishes gross output from value-added
production functions. ``Proxy or Control'' lists the proxy
variables or control variables used for identification.
For the proposed method, the ``Nonpara Non-Hicks'' checkmark
refers to the identification result in
Appendix~\ref{app:prod_recovery}; the implemented estimator
is Hicks-neutral (equation~\eqref{eq:gmm_prod}).
The proposed method requires conditional independence of
input-specific demand shocks
(Assumption~\ref{ass:cond_indep}) in place of the Markov and
scalar unobservability conditions; both blank cells in its row
reflect this substitution, not an absence of identifying
assumptions.}
\end{table}

I make three contributions. First, I show that the cross-sectional
covariance structure among three flexible intermediate inputs fully
substitutes for the Markov restriction, delivering nonparametric
identification of the production function and the productivity
distribution from a single period. This is a substitution, not a
relaxation, of identifying assumptions. The mapping to the HS08
framework provides density identification
(Theorem~\ref{thm:density_id}); the theoretical contribution of
this paper lies in what follows. I characterize the residual
indeterminacy that arises when Markov is dropped: any two
observationally equivalent structures differ only by a location
shift $\Delta(k,l)$ applied to productivity, ruling out nonlinear
transformations (Theorem~\ref{thm:obs_equiv}). I provide two routes
that close this indeterminacy without dynamic assumptions, an
exclusion restriction (Corollary~\ref{thm:exclusion}) and a
homothetic regularity condition
(Theorem~\ref{thm:homothetic_id}).

While nonparametric sieve
estimation could in principle implement the identification results
directly, the high-dimensional numerical integration is
computationally prohibitive for census-scale panels. I develop a
Cobb--Douglas GMM estimator designed for applied use, and establish
its consistency and asymptotic normality
(Theorem~\ref{thm:asymptotics}); the extension to translog
production is developed in Appendix~\ref{app:translog}. The conditional independence
assumption yields a pairwise convergence diagnostic
(Remark~\ref{rem:testable_exclusion}) with no analogue under the
Markov assumption: within the proxy variable framework, no
restriction distinguishes a correctly specified AR(1) from an AR(2)
or a potential outcome process. In 502 industries, this diagnostic
converges to zero for capital but not for labor, providing direct
evidence on the differential applicability of the exclusion
restriction.

Second, I document that the Markov assumption generates a
systematic upward bias in measured market power. Monte Carlo
simulations show that ACF's bias in $\hat{\beta}_m$ does not
vanish as sample size grows: $+0.026$ under AR(2) dynamics and
$+0.266$ under potential outcome dynamics. In the empirical
application, ACF produces higher materials elasticities and higher
markups at every percentile across 502 industries. The gap crosses
the competitive threshold and reverses the policy-relevant
conclusion about market structure. The recovered productivity
measures also show stronger associations with economic fundamentals
than those from the standard method, consistent with a higher
signal-to-noise ratio from separating input-specific demand shocks
(Section~\ref{sec:empirical}).

Third, I show that productivity measures recovered from the
proposed method are valid under the potential outcomes framework
(Proposition~\ref{prop:omega_hat_D}), resolving the inconsistency
identified by \textcite{chen2024identifying}. Because the estimator
uses no transition equation, the recovered productivity is
invariant to how a treatment operates on productivity dynamics.
The earthquake event study illustrates the practical consequence:
the proposed estimate is $-1.28$ percent while ACF yields $-1.68$
percent, a gap that arises because the ACF estimate lacks the
theoretical guarantee that the production function parameters are
consistently estimated under treatment-induced dynamics.
The same static, $\omega$-conditional structure also renders
the estimator robust to endogenous exit: under the standard
timing convention where exit precedes input choice, conditioning
on $\omega$ absorbs survival selection, and no survival probability
correction is needed (Remark~\ref{rem:selection}).

\paragraph{Related literature.}

Table~\ref{tab:comparison} positions my identification strategy
within the recent literature. The most closely related work is
\textcite{gandhi2020onthe} (GNR). GNR's Theorem~1 establishes
that proxy variable methods alone cannot identify the gross
output production function; additional within-period,
cross-sectional information is required. Both approaches
supply such information: GNR through the structural link between the
production function and the firm's first-order condition, yielding
a nonparametric share regression that directly identifies the
flexible input elasticity; my approach through the measurement
error structure of HS08, using conditional independence across
intermediate inputs to recover the distribution of unobserved
productivity.

The two approaches rest on different assumptions regarding input
markets. GNR's first-order condition requires competitive input
markets with common prices and that any unobserved component in the
share equation is non-persistent (their Appendix~O6, Assumption~7);
when input-specific markdowns or procurement frictions are
persistent, the FOC-based estimation equation does not hold and the
share regression is misspecified. My framework
permits persistent, input-specific demand shocks arising from
procurement relationships, supply contracts, or input-specific
markdowns; identification requires only mutual independence across
inputs at each time point, accommodating arbitrary serial
dependence within each shock. GNR's second stage recovers capital
and labor elasticities using the Markov structure; my approach
requires no dynamic assumption at any stage. The scalar
unobservability case is a special case of my model, obtained when
the input-specific shocks are degenerate
(Section~\ref{sec:simulation}).

Alternative approaches that exploit static first-order conditions
\parencite{grieco2016production,caselli2025productivity} avoid
dynamic assumptions but generally require parametric restrictions
on functional forms and the demand system. Additional related work
is summarized in Table~\ref{tab:comparison}. Several recent papers
apply the HS08 framework to production functions
\parencite{brandestimating,hu2020estimating,dotyadynamic},
but all use lagged variables as instruments and therefore retain the
Markov assumption. \textcite{zeng2023identification} avoid the
Markov restriction at the estimation stage but presuppose it for
the investment policy function. A growing literature on
non-Hicks-neutral identification
\parencite{navarrononparametric,ackerberg2022nonparametric,pan2022identification,kasahara2023identification,dotyadynamic},
including factor-augmenting approaches
\parencite{doraszelski2018measuring,demirerproduction,raval2019themicro},
retains first- or higher-order Markov assumptions; my identification
results extend to these models without dynamic restrictions
(Appendix~\ref{app:prod_recovery}), though the implemented
estimator uses the Hicks-neutral Cobb--Douglas specialization
(Section~\ref{sec:gmm_spec}).

The remainder of the paper is organized as follows.
Section~\ref{sec:model_identification} presents the model and the
nonparametric identification results.
Section~\ref{sec:estimation} develops the GMM estimator.
Section~\ref{sec:simulation} presents Monte Carlo evidence.
Section~\ref{sec:empirical} applies the estimator to 502 Japanese
manufacturing industries.
Section~\ref{sec:conclusion} concludes.

\section{Model and Identification}

\label{sec:model_identification}

This section establishes the identification strategy in three steps.
First, I show that three conditionally independent input demands
identify the joint distribution of productivity and inputs within
each capital-labor cell (Theorem~\ref{thm:density_id}); any two
observationally equivalent structures differ only by a location
shift $\Delta(k,l)$ (Theorem~\ref{thm:obs_equiv}).
Second, I provide two conditions that eliminate this indeterminacy:
an exclusion restriction (Corollary~\ref{thm:exclusion}) and a
homothetic regularity condition
(Theorem~\ref{thm:homothetic_id}). The exclusion restriction
carries a testable implication
(Remark~\ref{rem:testable_exclusion}). The formal statement of
density identification (Theorem~\ref{thm:density_id}) and the technical regularity
conditions (Assumptions~\ref{ass:injectivity}--\ref{ass:labeling}) are in
Appendix~\ref{app:identification_details}.
These identification results translate into three groups of moment
conditions in the GMM estimator of
Section~\ref{sec:estimation}: proxy moments (Block~A),
covariance moments (Block~B), and curvature moments (Block~C).
When these terms appear below, they refer forward to
Section~\ref{sec:gmm_approach}.

\subsection{Model Setup}

I define the general gross output production function for
firm $j$ at time $t$ as follows: 
\begin{equation}
y_{jt}=f_{t}(k_{jt},l_{jt},m_{jt},e_{jt},w_{jt},\omega_{jt})+\varepsilon_{jt}\label{eq:prod_func}
\end{equation}
Here, $y_{jt}$ is the logarithm of output, $k_{jt}$ and $l_{jt}$
are the logarithms of capital and labor. Following the production function literature
\parencite{olley1996thedynamics,ackerberg2015identification,bond2005adjustment},
capital and labor are treated as dynamic or quasi-fixed inputs whose
current values are predetermined relative to intermediate input
decisions.
The model requires at least three distinct intermediate inputs:
$m_{jt}$ (raw materials), $e_{jt}$ (electricity), and $w_{jt}$ (industrial water).
Three inputs are the minimum required by the
\textcite{hu2008instrumental} spectral decomposition: it identifies
the latent productivity distribution from three mutually independent
measurements of a common latent variable; two measurements do not
suffice for nonparametric identification without additional
restrictions.\footnote{When labor adjusts rapidly to current
productivity, it may serve as a third measurement of $\omega_{jt}$,
reducing the required number of flexible intermediate inputs from
three to two; see Footnote~\ref{rem:static_labor} for details.}
$\omega_{jt}$ is the firm's productivity, unobserved by the
econometrician but known to the firm when making input decisions.
$\varepsilon_{jt}$ denotes ex-post production shocks (measurement
error or unexpected disruptions), unobserved by both the firm and
the econometrician at the time of input choice.

The state variable vector $x_{jt}=(k_{jt},l_{jt},z_{jt})$
determines input demand. Here, $k_{jt}$ and $l_{jt}$ are the
primary inputs, while $z_{jt}$ represents additional firm-specific
state variables such as inventory levels, input prices, or market
conditions that do not directly enter the production function but
influence input demand. Given $x_{jt}$, the demand for each
intermediate input is determined as follows:
\begin{align}
m_{jt} & =g_{m}(x_{jt},\omega_{jt},\tau_{jt})\label{eq:demand_m}\\
e_{jt} & =g_e(x_{jt},\omega_{jt},\nu_{jt})\label{eq:demand_mp}\\
w_{jt} & =g_w(x_{jt},\omega_{jt},\eta_{jt})\label{eq:demand_mpp}
\end{align}
The functions $g_{m}(\cdot)$, $g_e(\cdot)$, and $g_w(\cdot)$
are unknown and potentially nonlinear. $\tau_{jt}$,
$\nu_{jt}$, and $\eta_{jt}$ are unobserved shock terms specific
to each input demand, following
\textcites{hu2020estimating}{brandestimating}{dotyadynamic}. These
shocks capture optimization errors, supply disruptions, and
adjustment frictions not explained by productivity and state
variables. Appendix~\ref{sec:micro_fnd} derives the demand system from
a cost-minimization problem under imperfect input markets and shows
that these shocks correspond to input-specific markdowns, prices,
and wedges; specifically, the components of markdowns and wedges
orthogonal to observable state variables.

The presence of input-specific shocks represents a departure from
the scalar unobservability assumption maintained in
\textcite{olley1996thedynamics},
\textcite{levinsohn2003estimating},
\textcite{ackerberg2015identification}, GNR, and others, which
requires productivity to be the sole unobservable affecting input
demand. When scalar unobservability fails because firm-level input
prices, markdowns, or wedges are unobserved, standard proxy variable
estimators are inconsistent
\parencite{jaumandreu2021reexamining,doraszelski2025production}. In
my framework, all unobserved firm-specific heterogeneity beyond
productivity is absorbed into $\tau_{jt}, \nu_{jt}, \eta_{jt}$, and
identification requires only that these shocks be mutually
independent across inputs, not that they be absent. Scalar
unobservability is nested as the special case $\tau = \nu = \eta = 0$
at the model level; the identification strategy requires
non-degenerate demand shocks and is therefore complementary to,
rather than a generalization of, scalar inversion methods.
From the standpoint of the cost-minimization model in
Appendix~\ref{sec:micro_fnd}, $\tau = \nu = \eta = 0$ requires
that all firms in an industry face identical input prices, identical
markdowns in every input market, and make no optimization errors in
input choice. In practice, firms negotiate procurement contracts
individually, face supplier-specific delivery terms, and adjust
input quantities with heterogeneous frictions. The presence of
input-specific demand shocks is the empirically relevant case; the
proposed framework treats these shocks as a source of identifying
information rather than a nuisance to be assumed away.

This formulation also addresses the collinearity problem identified
by \textcite{gandhi2020onthe}: under scalar unobservability, flexible
inputs determined by static optimization lack sufficient residual
variation to identify the gross production function
\parencite{ackerberg2015identification,bond2005adjustment}. GNR
resolve this problem by exploiting the first-order condition for the
flexible input, which identifies its output elasticity from the
revenue share. My approach resolves the collinearity through
independent input-specific shocks, which supply the cross-sectional
variation needed for identification via the measurement error
structure of HS08, without relying on the first-order condition or
dynamic moment conditions. The practical difference is that the
share regression requires the first-order condition to hold with
common input prices, whereas my approach permits firm-specific
input prices and markdowns (Appendix~\ref{sec:micro_fnd}).

\subsection{Assumptions for Identification}
\label{sec:assumptions}

The identification theory rests on two substantive assumptions
stated here, together with three regularity conditions
(Assumptions~\ref{ass:injectivity}--\ref{ass:labeling}) collected in
Appendix~\ref{app:identification_details}.

\begin{assumption}[Additive Error Structure]
\label{ass:additive_error}
The production function has an additive error structure:
\begin{equation}
  y_{jt} = f_t(k_{jt}, l_{jt}, m_{jt}, e_{jt}, w_{jt},
  \omega_{jt}) + \varepsilon_{jt},
  \label{eq:prod_func_additive}
\end{equation}
where the ex-post shock $\varepsilon_{jt}$ satisfies
\[
  \mathbb{E}\bigl[\varepsilon_{jt} \mid k_{jt}, l_{jt}, m_{jt},
  e_{jt}, w_{jt}, \omega_{jt}\bigr] = 0.
\]
\end{assumption}

\textit{Role and economic content.} This is standard in the production
function literature \parencite{olley1996thedynamics,ackerberg2015identification}.
The shock $\varepsilon_{jt}$ captures ex-post deviations (measurement
error, unexpected disruptions) that are realized after input choices
are made and are therefore uncorrelated with all inputs and
productivity. It acts as classical measurement error in the
dependent variable and inflates standard errors but does not bias
the production function estimates
(Theorem~\ref{thm:asymptotics}).

\begin{assumption}[Conditional Independence]
\label{ass:cond_indep}
The demand shocks $(\tau_{jt}, \nu_{jt}, \eta_{jt})$ for the three
intermediate inputs are mutually independent, conditional on
productivity $\omega_{jt}$ and state variables
$x_{jt} = (k_{jt}, l_{jt}, z_{jt})$:
\[
  f_{\tau, \nu, \eta \mid \omega, x}
  = f_{\tau \mid \omega, x} \cdot f_{\nu \mid \omega, x}
  \cdot f_{\eta \mid \omega, x}.
\]
Mutual independence is required; pairwise independence does not
suffice for the spectral decomposition of HS08.
\end{assumption}

\textit{Role.} This is the substantive identifying condition. Together
with the regularity conditions in
Appendix~\ref{app:identification_details}
(Assumptions~\ref{ass:injectivity}--\ref{ass:labeling}), it enables
the unique spectral decomposition of the integral
equation~\eqref{eq:integral_eq}. Conditional independence is the
economically substantive condition; it restricts the data generating
process rather than regularity of the operators.

\textit{Economic content.} The assumption posits that, for a firm
with given state variables and productivity level, an unexpected
shock to raw material demand (e.g., a supply chain disruption) is
independent of a shock to electricity demand (e.g., an unscheduled
rate surcharge). This is natural when input markets are
segmented: raw materials, electricity, and water are procured
through distinct channels, under separate contracts, with
different suppliers. The common components of demand variation (product demand
fluctuations, aggregate markup changes) are captured by $x_{jt}$;
$\tau_{jt}, \nu_{jt}, \eta_{jt}$ represent the residual,
input-specific components. The microfoundations in
Appendix~\ref{sec:micro_fnd} make this structure precise.

\subsection{Interpretation and Robustness of the Conditional Independence Assumption}

\label{subsec:relax_indep}

The general principle is as follows. Common shocks that affect all three input
demands (product demand fluctuations, markup variation, aggregate
input price movements) can be absorbed by projecting onto
observable control variables $z_{jt}$; the shock terms
$\tau_{jt}, \nu_{jt}, \eta_{jt}$ are then defined as the
orthogonal residuals of this projection
(Appendix~\ref{sec:micro_fnd}). The independence assumption
therefore requires only that the \textit{residual}, input-specific
components of demand variation are mutually independent.

Several potential threats illustrate this principle.
\textit{Unobserved demand shocks} generate common variation across
all inputs, but can be proxied by inventory fluctuations
\parencite{kumar2019productivity} or recovered from revenue data
\parencite{kasahara2020nonparametric}, included in $z_{jt}$.
\textit{Product market power} affects all input demands through
marginal revenue; following
\textcite{ackerbergproduction,jaumandreu2025robustproduction},
low-dimensional sufficient statistics for the markup (e.g.,
competitors' output, average variable cost) can be included in
$z_{jt}$.\footnote{Under Cournot competition,
\textcite{ackerbergproduction} show that the total output of
competitors serves as a sufficient statistic.}
\textit{Input market power} (markdowns) may generate common
bargaining advantages across inputs, but the common component
depends on firm attributes (size, liquidity) captured by
$(k_{jt}, l_{jt}, z_{jt})$; what remains in the shock terms are
idiosyncratic variations from individual supplier relationships.
It is economically reasonable that the outcome of negotiations
with raw material suppliers is independent of electricity rate
negotiations, conditional on firm size and other
observables.\footnote{When an intermediate input is traded on
competitive commodity markets, the firm is a price-taker and
the markdown on that input vanishes.
\textcite{avignon2025markups} exploit this property for
globally traded dairy commodities to separately identify
markups and markdowns on other inputs.}
\textit{Common input price shocks} (e.g., oil price hikes)
affect multiple inputs symmetrically and are controlled by time
fixed effects or industry-specific deflators in $z_{jt}$.
Firm-specific price variations are absorbed as part of the
structural shock terms and need only be independent across inputs.

\subsection{Identification of the Production Function}
\label{sec:prod_func_id}

The identification proceeds in two stages: first, I recover the
production function and productivity distribution within each
capital-labor cell $(k_0, l_0)$; second, I characterize and resolve
the residual indeterminacy that arises when linking these
cell-specific results across different values of $(k, l)$.%
\footnote{In the following, firm subscripts $j$ are suppressed as
I discuss population-level arguments. The time subscript $t$ is
retained only to indicate time-variation in the production function
$f_t$.}

\subsubsection{Identification within Each $(k, l)$}
\label{sec:within_kl}

The foundational identification result applies the spectral
decomposition of HS08, whose conditions I verify under the
present assumptions.

\begin{thm}[Identification of Densities]
\label{thm:density_id}
Under Assumptions~\ref{ass:additive_error}--\ref{ass:cond_indep}
and Assumptions~\ref{ass:injectivity}--\ref{ass:labeling},
the observable conditional joint density
$f_{m_{jt}, e_{jt} \mid x_{jt}, w_{jt}}$ uniquely identifies
the three unknown conditional density functions:
$f_{m_{jt} \mid \omega_{jt}, x_{jt}}$,
$f_{e_{jt} \mid \omega_{jt}, x_{jt}}$, and
$f_{\omega_{jt} \mid x_{jt}, w_{jt}}$.
\end{thm}

The proof, which verifies the conditions of HS08's Theorem~1
for the integral equation~\eqref{eq:integral_eq}, is in
Appendix~\ref{app:density_id_proof}.

As a consequence of Theorem~\ref{thm:density_id} and
equation~\eqref{eq:bayes_omega}, for each fixed $(k_0, l_0)$, the
following are nonparametrically identified: the conditional densities
$f_{m \mid \omega, k_0, l_0}$, $f_{e \mid \omega, k_0, l_0}$,
$f_{w \mid \omega, k_0, l_0}$, and
$f_{\omega \mid k_0, l_0, m, e, w}$.

Using these identification results, I recover the structure of
$f_t$ as a function of $(m, e, w, \omega)$. I focus on the Hicks-neutral specification, widely adopted in the
empirical literature, and defer the general case to
Appendix~\ref{app:prod_recovery}. Under this specification
$y = g_t(k, l, m, e, w) + \omega + \varepsilon$,
Assumption~\ref{ass:additive_error} implies
\begin{equation}
  g_t(k_0, l_0, m, e, w) 
  = \mathbb{E}[y \mid x, m, e, w] 
  - \mathbb{E}[\omega \mid x, m, e, w].
  \label{eq:hicks_neutral_id}
\end{equation}
Here $g_t$ represents the component of the production technology that
depends on intermediate inputs, with productivity $\omega$ separated
out. The first term on the right-hand side is a conditional
expectation identified directly from the data, and the second is
computable from the posterior density in
equation~\eqref{eq:bayes_omega}. Thus $g_t$ is identified as a
function of $(m, e, w)$ without additional assumptions. For the 
general non-Hicks-neutral model, 
$f_t(k_0, l_0, m, e, w, \omega)$ is identified as a function of 
$(m, e, w, \omega)$ under additional regularity conditions on the 
distribution of $\varepsilon$; see Appendix~\ref{app:prod_recovery} 
for details.

For each fixed $(k_0, l_0)$, the conditional distribution
$f_{\omega \mid k_0, l_0, m, e, w}$ is fully characterized,
and the conditional expectation
\begin{equation}
  \hat{\omega}_{jt} \equiv
  \mathbb{E}[\omega_{jt} \mid x_{jt}, m_{jt}, e_{jt}, w_{jt}]
  = \int \omega \, f_{\omega \mid x, m, e, w}
  (\omega \mid \cdot)\, d\omega
  \label{eq:omega_hat}
\end{equation}
provides a firm-level productivity measure for each firm $j$ and
period $t$. The empirical applications of this within-$(k,l)$
identification, including markup estimation and policy evaluation,
are developed in Section~\ref{sec:implications_empirical} after the
identification theory is completed.

However, to identify $f_t$ as a function of $(k, l)$ as well,
additional structure is needed. (When labor adjusts rapidly to
current productivity, it serves as an additional measurement,
reducing the required intermediate inputs from three to
two.\footnote{\label{rem:static_labor}When labor adjusts within the
production period, it serves as a third measurement of
$\omega_{jt}$, and the HS08 identification procedure
(Theorem~\ref{thm:density_id}) applies to the triple
$(l_{jt}, m_{jt}, e_{jt})$, reducing the required flexible
intermediate inputs from three to two. This extension applies when
adjustment costs are small enough that $l_{jt}$ responds to
within-period productivity innovations; industries with high
turnover or temporary staffing (e.g., food processing, garment
manufacturing) are natural candidates. When labor is quasi-fixed,
$l_{jt}$ reflects past rather than current productivity, and the
conditional independence conditions for $(l, m, e)$ do not hold.
See Appendix~\ref{app:prod_recovery} for details.})
$\omega$ must be defined on a common scale across different values of
$(k, l)$. Since Theorem~\ref{thm:density_id} applies the HS08
procedure independently for each $(k, l)$, there is no automatic
correspondence between the $\omega$ values identified at
$(k_1, l_1)$ and those identified at $(k_2, l_2)$. I now formalize
this problem.

\subsubsection{Observational Equivalence and Limits of
Identification}
\label{sec:obs_equiv}

Theorem~\ref{thm:density_id} identifies the production function
within each $(k_0, l_0)$, but a practitioner needs parameters that
are comparable across different capital-labor combinations. The next
result shows exactly what remains unresolved and rules out the
possibility that the indeterminacy takes a nonlinear form.

Under Assumptions~\ref{ass:additive_error}--\ref{ass:cond_indep} and the
regularity conditions in Appendix~\ref{app:identification_details}
(Assumptions~\ref{ass:injectivity}--\ref{ass:labeling}), the conditional densities
$f_{m|\omega}$, $f_{e|\omega}$, $f_{w|\omega}$ and the marginal
density $f_\omega$ are nonparametrically identified from the joint
density of $(m,e,w)$ conditional on $(k,l,z)$
(Theorem~\ref{thm:density_id}). This pins down the shape of each
conditional distribution but leaves a common location shift
$\Delta(k,l)$ applied to the latent variable unresolved. The
following theorem characterizes this residual indeterminacy
completely.

\begin{thm}[Complete Characterization of Observational Equivalence]
\label{thm:obs_equiv}
Under Assumptions~\ref{ass:additive_error}--\ref{ass:cond_indep}
and Assumptions~\ref{ass:injectivity}--\ref{ass:labeling}, a 
necessary and sufficient condition for two structures 
$(f_t, \omega)$ and $(\tilde{f}_t, \tilde{\omega})$ to generate the 
same joint distribution of observables is that there exists a 
continuous function $\Delta(k, l)$ such that
\begin{equation}
  \tilde{\omega} = \omega + \Delta(k, l), \qquad
  \tilde{f}_t(k, l, m, e, w, \tilde{\omega})
  = f_t(k, l, m, e, w, \tilde{\omega} - \Delta(k, l)).
  \label{eq:obs_equiv}
\end{equation}
\end{thm}

The proof is given in Appendix~\ref{app:proof_obs_equiv}; the key
steps are as follows. The HS08 eigenvalue-eigenfunction
decomposition uniquely determines the functional form of each
conditional density within each $(k_0, l_0)$, ruling out nonlinear
transformations of $\omega$. Any remaining degree of freedom must
therefore be a location shift that varies across $(k, l)$,
yielding~\eqref{eq:obs_equiv}. The continuity of $\Delta(k,l)$
follows from the continuous dependence of
$f_{m \mid \omega, k, l}$ on $(k, l)$ (stated after
Assumption~\ref{ass:labeling}) together with the perturbation
theory of compact operators under simple eigenvalues
(Assumption~\ref{ass:distinct_eigen}; see
Appendix~\ref{app:proof_obs_equiv} for details).

Nonlinear transformations (including scale transformations) are ruled
out because the eigenvalue--eigenfunction decomposition in HS08
uniquely determines the functional form of each conditional density 
within each $(k_0, l_0)$. Second, the $\Delta(k, l)$ indeterminacy 
arises inherently from the fact that Theorem~\ref{thm:density_id} 
applies the HS08 procedure independently for each $(k, l)$. Within 
each $(k_0, l_0)$, Assumption~\ref{ass:labeling} fixes the level of 
$\omega$, but the reference point of this normalization may depend on 
$(k_0, l_0)$. The data on conditional distributions of intermediate 
input demands do not contain information to unify $\omega$ levels 
across different $(k, l)$.\footnote{%
\textcite{hahn2023identification} show that in dynamic approaches
such as \textcite{olley1996thedynamics}, the identification of
dynamic input elasticities relies on an index restriction that
collapses state variables into a one-dimensional scalar.
Theorem~\ref{thm:density_id} does not provide such an index
restriction, and hence the indeterminacy with respect to the dynamic
elasticities persists.}

Economically, the $\Delta(k, l)$ indeterminacy means that the effect
of $(k, l)$ on the production function and
$\mathbb{E}[\omega \mid k, l]$ cannot be separated without
additional restrictions. As a direct consequence,
$f_t$ is identified up to the specification of
$\mathbb{E}[\omega \mid k, l]$: fixing
$\mathbb{E}[\omega \mid k, l]$ pins down $\Delta = 0$
(Theorem~\ref{thm:id_up_to_E},
Appendix~\ref{app:identification_details}).

The $\Delta(k, l)$ indeterminacy
also arises in the existing literature:
\textcite{gandhi2020onthe} resolve it in the Hicks-neutral setting
by combining first-order conditions with a Markov assumption, which
reduces $\Delta(k, l)$ to a constant; for non-Hicks-neutral models,
this strategy fails because $\omega$ cannot be separated from the
first-order condition.\footnote{In the Hicks-neutral model,
$\Delta(k, l)$ shifts the $(k,l)$ component of the production
function:
$\tilde{g}_t(k, l, \cdot) = g_t(k, l, \cdot) - \Delta(k, l)$.
In non-Hicks-neutral models, the FOC
$P_t \cdot (\partial f_t / \partial M) = \rho_t$ retains $\omega$
on the left-hand side, precluding a share regression.
\textcite{li2024identification} show that heterogeneous output
elasticities with respect to flexible inputs remain identifiable
under a scalar unobservable assumption on the proxy variable.}

I provide two alternative methods that close the identification
gap without dynamic assumptions. Theorem~\ref{thm:obs_equiv}
guarantees that $\Delta(k, l)$ is a continuous function of $(k,l)$
alone, which both methods exploit. Section~\ref{sec:exclusion}
imposes exclusion restrictions on intermediate input demands that
directly constrain the functional form of $\Delta(k, l)$, achieving
nonparametric point identification. Section~\ref{sec:homothetic}
parametrically specifies the $(k, l)$ component and introduces a
regularity condition on the shape of $\mathbb{E}[\omega \mid k, l]$,
achieving parametric identification through the non-constant
curvature of the homothetic transformation.

\subsubsection{Closing the Identification Gap}
\label{sec:closing_gap}

The $\Delta(k,l)$ indeterminacy is the cost of dispensing with the
Markov assumption. I now show this cost is payable: two conditions,
each operating without dynamic restrictions, eliminate the
indeterminacy and deliver point identification.

\subsubsection{Nonparametric Identification via Exclusion Restrictions}
\label{sec:exclusion}

The $\Delta(k, l)$ indeterminacy arises because the location 
normalization in Assumption~\ref{ass:labeling} is applied 
independently for each $(k, l)$ (Theorem~\ref{thm:obs_equiv}). If 
the HS08 location normalization 
$M[f_{m \mid \omega, k, l}(\cdot \mid \omega)] = \omega$ could be 
applied \textit{uniformly} across all $(k, l)$, then 
$\Delta(k, l) = 0$ would follow immediately. However, for this 
uniform normalization to hold, 
$M[f_{m \mid \omega, k, l}(\cdot \mid \omega)]$ must not depend on 
$(k, l)$; that is, the conditional demand for the intermediate 
input, given $\omega$, must be independent of $(k, l)$. This 
observation suggests that exclusion restrictions on intermediate 
input demands directly constrain 
$\Delta(k, l)$.

\begin{corollary}[Identification via Exclusion Restrictions]
\label{thm:exclusion}
In addition to
Assumptions~\ref{ass:additive_error}--\ref{ass:cond_indep}
and~\ref{ass:injectivity}--\ref{ass:labeling}, suppose one of 
the following conditions holds:
\begin{enumerate}
\item[(i)] The demand for some intermediate input (e.g., $w$) does 
  not depend on $(k, l)$: 
  $f_{w \mid \omega, k, l} = f_{w \mid \omega}$.
\item[(ii)] The demand for one input (e.g., $m$) does not depend on 
  $k$, and the demand for another (e.g., $e$) does not depend on 
  $l$: $f_{m \mid \omega, k, l} = f_{m \mid \omega, l}$ and 
  $f_{e \mid \omega, k, l} = f_{e \mid \omega, k}$.
\end{enumerate}
Then, under the normalization $\mathbb{E}[\omega] = 0$, the 
production function $f_t$ is nonparametrically point-identified. 
Condition~(i) is a special case of condition~(ii).
\end{corollary}

\begin{proof}
By Theorem~\ref{thm:obs_equiv}, observationally equivalent 
structures are parameterized by 
$\tilde{\omega} = \omega + \Delta(k, l)$. Requiring that the 
exclusion restriction be maintained in the alternative structure:

\textit{Condition~(i):} 
$f_{w \mid \omega, k, l} = f_{w \mid \omega}$ implies that in 
the alternative structure, 
$f_{w \mid \tilde{\omega}, k, l}(w \mid \tilde{\omega}) 
= f_{w \mid \omega}(w \mid \tilde{\omega} - \Delta(k, l))$, 
which is independent of $(k, l)$ only if $\Delta(k, l)$ is constant.

\textit{Condition~(ii):} 
$f_{m \mid \omega, k, l} = f_{m \mid \omega, l}$ implies $\Delta$
does not depend on $k$.
$f_{e \mid \omega, k, l} = f_{e \mid \omega, k}$ implies $\Delta$
does not depend on $l$. Together, $\Delta$ is constant.

In both cases, $\mathbb{E}[\omega] = 0$ pins down $\Delta = 0$.
\end{proof}

Economically, condition~(i) requires that the demand for some
intermediate input (e.g., electricity) depends on productivity alone
and not on capital or labor intensity; this
may hold in energy-intensive industries where electricity consumption
is driven by production volume rather than by the composition of
capital equipment. Condition~(ii) requires that different inputs
exclude different primary inputs from their demand: for example, raw
material demand does not depend on labor intensity, and fuel demand
does not depend on capital intensity. These exclusion restrictions limit the scope of application to
industries where institutional knowledge supports them. For
settings where
such restrictions cannot be justified, I provide a parametric
alternative in the next subsection.

\begin{remark}[Testability of the Exclusion Restriction]
\label{rem:testable_exclusion}
Under the linear demand specification~\eqref{eq:gmm_demand_m}--\eqref{eq:gmm_demand_mpp},
let $a_{k}^h$, $a_{l}^h$, and $a_{\omega}^h$ denote the slope coefficients
on $k$, $l$, and $\omega$ in the demand for input $h$:
$(a_{k}^m, a_{l}^m, a_{\omega}^m) = (\gamma_k, \gamma_l, \gamma_\omega)$,
$(a_{k}^e, a_{l}^e, a_{\omega}^e) = (\delta_k, \delta_l, \delta_\omega)$,
$(a_{k}^w, a_{l}^w, a_{\omega}^w) = (\zeta_k, \zeta_l, \zeta_\omega)$.
The exclusion restriction of Corollary~\ref{thm:exclusion} for a
single input $h$ is not separately testable from Block~A+B
estimates. Under the normalization $\beta_k = \beta_l = 0$
(Section~\ref{sec:gmm_approach}), the estimated demand
coefficient $\hat{a}_{k}^{h *}$ converges to
$a_{k}^{h} - a_{\omega}^{h}\beta_k$, confounding the structural
exclusion parameter $a_{k}^{h}$ with the indeterminacy
$a_{\omega}^{h}\beta_k$ from Theorem~\ref{thm:obs_equiv}.

The \emph{joint} restriction across inputs, however, yields a
diagnostic test, a necessary condition for consistency with the
exclusion restriction, but not a sufficient one. Under
Proposition~\ref{prop:excl_ols}, the OLS estimate
$\hat{\beta}_k^{(h)}$ from input $h$ converges to
$\beta_k - a_{k}^{h}/a_{\omega}^{h}$. Define the pairwise
discrepancy
\begin{equation}
  d_k^{(h_1, h_2)}
  \equiv
  \frac{\hat{a}_{k}^{h_2 *}}{\hat{a}_{\omega}^{h_2}}
  - \frac{\hat{a}_{k}^{h_1 *}}{\hat{a}_{\omega}^{h_1}}
  \xrightarrow{p}
  \frac{a_{k}^{h_2}}{a_{\omega}^{h_2}}
  - \frac{a_{k}^{h_1}}{a_{\omega}^{h_1}},
  \label{eq:overid_diff}
\end{equation}
which is free of the $\Delta(k,l)$ indeterminacy since
$\beta_k$ cancels in the difference. Under the joint
exclusion restriction $a_{k}^{h_1} = a_{k}^{h_2} = 0$,
$d_k = 0$; the converse does not hold. The test statistic
$d_k = 0$ is a \emph{necessary condition} for the full joint exclusion
restriction, not a sufficient one: $d_k = 0$ also obtains in the
knife-edge case where $a_{k}^{h}/a_{\omega}^{h}$ is equal across inputs
but nonzero. This configuration has no structural basis when the three
inputs involve distinct procurement channels, but the possibility
cannot be ruled out on the basis of $d_k$ alone
(Appendix~\ref{app:testability_exclusion}).
The test is therefore best interpreted as a \emph{diagnostic}: a
rejection of $d_k = 0$ is evidence against the exclusion restriction,
while non-rejection is consistent with, but does not establish, it.
Since $d_k$ is a smooth function of the
Block~A+B parameters, its standard error is obtained by the
delta method from the GMM variance-covariance matrix,
yielding a Wald test without the generated-regressors
problem that would arise from testing OLS estimates
directly. With three inputs, the formal test has two
degrees of freedom ($d_k = d_l = 0$)
(Appendix~\ref{app:testability_exclusion}).
I apply this test in Section~\ref{subsec:exclusion_test}.
\end{remark}

The formal statement and proof are given in
Proposition~\ref{prop:excl_ols}
(Appendix~\ref{app:identification_details}).\footnote{Replacing the
linear subtraction of $\hat{\omega}^h$ in
Proposition~\ref{prop:excl_ols} with a polynomial regression
is not consistent in general; see
Appendix~\ref{app:nonlinear_specs} for details.}

\subsubsection{Parametric Identification via Homothetic Regularity}
\label{sec:homothetic}

As an alternative when exclusion restrictions cannot be justified, I
parametrically specify the $(k, l)$ component and introduce a 
regularity condition on $\mathbb{E}[\omega \mid k, l]$. Consider the 
additively separable model
\begin{equation}
  y = g(k, l;\, \theta) + q(m, e, w) + \omega + \varepsilon,
  \label{eq:parametric_model}
\end{equation}
where $g$ is parametric with known functional form and $q$ is 
nonparametric. From Section~\ref{sec:within_kl}, $q$ is 
nonparametrically recoverable for each fixed 
$(k_0, l_0, \omega_0)$.

Specializing to $g(k, l;\, \theta) = \beta_k k + \beta_l l$, 
Theorem~\ref{thm:obs_equiv} reduces the identification indeterminacy 
to
\begin{equation}
  \Delta(k, l) = c_k k + c_l l,
  \quad (c_k, c_l) \in \mathbb{R}^2.
  \label{eq:linear_ambiguity}
\end{equation}
To eliminate this two-dimensional indeterminacy, I introduce the
following regularity condition.

\begin{assumption}[Homothetic Weak Separability]
\label{ass:homothetic}
The conditional expectation of TFP in the cross-section has a 
homothetic structure: there exist continuously differentiable 
functions $h \colon \mathbb{R} \to \mathbb{R}$ and
$v \colon \mathbb{R}^2 \to \mathbb{R}$ such that
\[
  \bar{\omega}(k, l) \equiv \mathbb{E}[\omega \mid k, l] 
  = h(v(k, l)),
\]
where:
\begin{enumerate}
\item[(A)] \textit{Nonlinear transformation:} $h'$ is not a constant 
  function.
\item[(B)] \textit{Translation homogeneity:} $v$ satisfies $v(k+c, l+c) = v(k,l) + c$ for all $c \in \mathbb{R}$.
\item[(C)] \textit{Imperfect substitutability:} The isoquants of $v$ 
  are strictly convex, and the marginal rate of substitution 
  $v_k / v_l$ is not constant on $(k, l)$.
\end{enumerate}
\end{assumption}

All three conditions are necessary for
Theorem~\ref{thm:homothetic_id}: (A) prevents observational
equivalence with linear functions; (B) ensures the counterfactual
index is also translation homogeneous, so that the MRS of
$\tilde{v}$ is translation invariant; (C) excludes Cobb--Douglas,
where $v_k/v_l$ is constant and a one-dimensional indeterminacy
persists. Economically, (A) requires nonlinear returns to the input
bundle, (B) corresponds to constant returns to scale in the level
variables (since translation homogeneity on the log scale is
equivalent to degree-one homogeneity in levels), and (C)
requires a finite and non-unit elasticity of substitution, satisfied
by CES, translog, and normalized quadratic forms.
Assumption~\ref{ass:homothetic} can be checked from Blocks~A
and~B alone (Section~\ref{sec:diagnostics}); detailed necessity
arguments and testability procedures are in
Appendix~\ref{app:assumption6_details}.

To illustrate, consider the CES specification where
$v(k, l) = \frac{1}{\rho_v}\log\bigl(\alpha e^{\rho_v k} + (1-\alpha) e^{\rho_v l}\bigr)$
is translation homogeneous on the log scale: $v(k+c, l+c) = v(k,l) + c$.
With $h(v) = \gamma v$ (for $\gamma \neq 0$ and higher-order terms
$\rho_2 v^2 + \rho_3 v^3$ with $\rho_2 \neq 0$ or $\rho_3 \neq 0$),
$h'$ is non-constant (satisfying (A)), $v$ is translation homogeneous
(satisfying (B)),
and the MRS $v_k/v_l = [\alpha/(1-\alpha)] e^{\rho_v(k-l)}$ is
non-constant for $\rho_v \neq 0$ (satisfying (C)). The
Cobb--Douglas case ($\rho_v \to 0$, so $v \to \alpha k + (1-\alpha) l$)
yields a linear $v$ and a constant MRS, violating conditions~(A) and~(C)
simultaneously; the rank condition in Theorem~\ref{thm:homothetic_id}
fails, and $(\beta_k, \beta_l)$ cannot be separately identified.
More generally, when $\rho_v$ is close to zero, identification of
$(\beta_k, \beta_l)$ through Block~C becomes weak: the marginal rate
of substitution $v_k/v_l$ approaches a constant as $\rho_v \to 0$,
so the cross-sectional variation in $(k_{jt}, l_{jt})$ provides little
leverage on the curvature parameters. In the empirical analysis, the
$t$-statistics for $\hat{\rho}_2$ and $\hat{\rho}_3$
(Section~\ref{sec:diagnostics}) provide a direct diagnostic for this
failure; industries where both are statistically insignificant should
not be relied upon for separate identification of $\beta_k$ and $\beta_l$
through Block~C alone.
When $\rho_v = 0$, the exclusion restriction of
Corollary~\ref{thm:exclusion} provides an alternative identification
route.

\begin{thm}[Static Parametric Identification]
\label{thm:homothetic_id}
Under Assumptions~\ref{ass:additive_error}--\ref{ass:cond_indep},
\ref{ass:injectivity}--\ref{ass:labeling},
and~\ref{ass:homothetic}, the structural parameters $\beta_k$
and $\beta_l$ in model~\eqref{eq:parametric_model} are
point-identified from static data alone.
\end{thm}

\begin{proof}
By contradiction. Suppose an observationally equivalent 
$\tilde{\beta}_i = \beta_i + c_i$ ($i = k, l$) exists with 
$(c_k, c_l) \neq (0, 0)$. By~\eqref{eq:linear_ambiguity}, the 
alternative TFP function satisfies 
$\mathbb{E}[\tilde{\omega} \mid k, l] 
= h(v(k, l)) - c_k k - c_l l$. Requiring that 
$\mathbb{E}[\tilde{\omega} \mid k, l]
= \tilde{h}(\tilde{v}(k, l))$ for some translation homogeneous
$\tilde{v}$ and differentiable $\tilde{h}$, the translation
invariance of the marginal rate of substitution of $\tilde{v}$
requires
\[
  (c_l\, v_k - c_k\, v_l)
  \bigl(h'(v + c) - h'(v)\bigr) = 0
\]
for all $(k, l) \in \mathbb{R}^2$ and $c \in \mathbb{R}$. By
condition~(A), $h'$ is non-constant, so the second factor is nonzero
for some $(v_0, c_0)$. Hence
$c_l\, v_k - c_k\, v_l = 0$ everywhere, so $v_k / v_l$ equals
the constant $c_k/c_l$. Under translation homogeneity, a constant
MRS forces $v(k,l) = \alpha k + (1-\alpha)l$, which is linear in
$(k,l)$, contradicting condition~(C). Therefore
$(c_k, c_l) = (0, 0)$.

Condition~(B) (translation homogeneity) enters the argument through
the translation invariance of the MRS of $\tilde{v}$: since
$\tilde{v}_k + \tilde{v}_l = 1$ (implied by translation homogeneity),
without it $\tilde{v}$ need not be translation homogeneous, and the
equality $c_l\, v_k - c_k\, v_l = 0$ does not follow.
\end{proof}

Theorem~\ref{thm:homothetic_id} is stated and proved for the CES
specification of $v(k,l)$; the argument extends to other parametric
forms (e.g., translog) subject to verifying the rank condition
specific to each functional form.\footnote{For instance, with a translog specification
$g = \beta_k k + \beta_l l + \beta_{kk} k^2 + \beta_{ll} l^2
+ \beta_{kl} kl$, $\Delta$ is restricted to the corresponding
polynomial class and the homothetic regularity condition eliminates
the indeterminacy by a similar argument, but the conditions on the
MRS differ from the CES case.}

\subsection{Implications for Empirical Applications}
\label{sec:implications_empirical}

The within-$(k_0, l_0)$ identification results have direct empirical
applications that differ in what they require. Markup estimation
requires only $\beta_m$, which is identified by Blocks~A and~B alone.
Event studies and difference-in-differences designs similarly require
only Block~A+B: because the estimator uses no transition equation for
$\omega$, the recovered $\hat{\omega}_{jt}$ is valid under any
productivity dynamics, including treatment-induced non-Markov paths
(Proposition~\ref{prop:omega_hat_D}, Appendix~\ref{app:identification_details}).
Full productivity-level analysis (including the identification of
$\beta_k$ and $\beta_l$) requires Block~C in addition.

\paragraph{Applications.}
Because estimation does not employ a transition process for
$\omega$, the estimates are invariant to how a policy $D_{jt}$
affects productivity dynamics
(Proposition~\ref{prop:omega_hat_D},
Appendix~\ref{app:identification_details}).
For markup estimation, the within-$(k_0, l_0)$ results suffice:
output elasticities $\partial f_t / \partial m$ are identified for
each fixed $(k_0, l_0)$, which recovers markups as the ratio of the
output elasticity to the revenue share
\parencite{deloecker2012markups}.

\begin{remark}[Functional Form Generality]
\label{rem:functional_form}
The identification results of this paper rest on the conditional
independence of intermediate input demands (Assumption~\ref{ass:cond_indep}),
not on the functional form of production. Theorems~\ref{thm:density_id}
and~\ref{thm:obs_equiv} establish nonparametric identification via
the HS08 spectral decomposition for any production function satisfying
Assumptions~\ref{ass:additive_error}--\ref{ass:cond_indep}. The GMM
estimator of Section~\ref{sec:estimation} implements this under
Cobb--Douglas, where input demands are linear in productivity and the
moment conditions take a tractable linear form.
Appendix~\ref{app:translog} shows that the same identification
source (conditional independence) yields nonlinear moment conditions
under translog production. The empirical implementation focuses on
Cobb--Douglas to maintain computational tractability and to isolate
the effect of relaxing the Markov assumption from functional form
complexities.
\end{remark}

\begin{remark}[Robustness to Endogenous Exit]
\label{rem:selection}
Standard proxy variable estimators require a survival probability
correction \parencite{olley1996thedynamics} because the innovation
shock $\xi_{jt}$ in the Markov transition equation
$\omega_{jt} = g(\omega_{j,t-1}) + \xi_{jt}$ is left-truncated
conditional on survival: firms with $\omega_{jt}$ below the exit
threshold do not appear in the data, biasing $\mathbb{E}[\xi_{jt}
\mid \omega_{j,t-1}, S_{jt}=1]$ away from zero.

The proposed estimator does not use the transition equation and
therefore does not involve $\xi_{jt}$. Identification of
$(\beta_m, \beta_e, \beta_w)$ rests on the within-period
conditional independence of demand shocks
(Assumption~\ref{ass:cond_indep}), which conditions on $\omega_{jt}$.
Under the standard timing convention that exit decisions are made
at the start of period $t$ based on the state $(\omega_{jt}, k_{jt})$
before input-specific demand shocks $(\tau_{jt}, \nu_{jt}, \eta_{jt})$
are realized, survival is a deterministic function of
$(\omega_{jt}, k_{jt})$. Conditioning on $\omega_{jt}$ therefore
absorbs the selection:
\[
  f(\tau, \nu, \eta \mid \omega, x, S=1)
  = f(\tau, \nu, \eta \mid \omega, x),
\]
and the moment conditions that identify $(\beta_m, \beta_e, \beta_w)$
hold on the surviving population without any survival probability
correction. No assumption on the productivity process is required for
this result; it follows from the static, $\omega$-conditional
structure of the identification strategy.

Two qualifications apply. First, the recovered distribution of
$\omega_{jt}$ is the survivor distribution, not the population
distribution; aggregate productivity statistics based on the
recovered $\hat{\omega}_{jt}$ reflect surviving firms only. Second,
the argument does not extend to parameters identified from the
transition equation (e.g., the persistence of productivity), which
the proposed method does not estimate.
\end{remark}

\section{Estimation Methods}
\label{sec:estimation}

The nonparametric identification results of
Section~\ref{sec:model_identification} establish that the production
function and productivity distribution are identified from the joint
density of intermediate inputs; nonparametric sieve estimation could
in principle implement this directly, but the high-dimensional
numerical integration required is computationally prohibitive for
census-scale panels spanning hundreds of industries.
I therefore develop a GMM estimator that specializes to a linear
production function and linear demand functions. Under this
parametric restriction, the observational equivalence class of
Theorem~\ref{thm:obs_equiv} reduces to a two-dimensional
indeterminacy $(c_k, c_l)$ (equation~\eqref{eq:linear_ambiguity}),
and the identification results of Corollary~\ref{thm:exclusion}
and Theorem~\ref{thm:homothetic_id} carry through directly.

\subsection{Estimation Based on the Generalized Method of Moments}
\label{sec:gmm_approach}

As noted in Remark~\ref{rem:functional_form}, the identification
results of Section~\ref{sec:model_identification} apply to general
production functions. The parametric implementation below specializes
to the Cobb--Douglas case, where input demand functions are linear in
productivity (Appendix~\ref{sec:micro_fnd}). This linearity yields the
tractable linear GMM system of Blocks~A--B. Extension to flexible
functional forms such as translog is developed in
Appendix~\ref{app:translog}; the identification source remains the
conditional independence of demand shocks.

\subsubsection{Overview}

The GMM estimator jointly recovers the production function and demand
parameters from three blocks of moment conditions:
\begin{enumerate}
\item[(i)] \textbf{Block~A} (Proxy moments): orthogonality conditions 
  derived from eliminating $\omega_{jt}$ across pairs of demand 
  residuals and the production residual, using an asymmetric 
  instrument strategy;
\item[(ii)] \textbf{Block~B} (Covariance moments): cross-covariance 
  restrictions among demand and production residuals, exploiting the 
  mutual independence of demand shocks;
\item[(iii)] \textbf{Block~C} (Curvature moments): conditional
  moment restrictions derived from the homothetic regularity
  condition on $\mathbb{E}[\omega_{jt} \mid k_{jt}, l_{jt}]$
  (Assumption~\ref{ass:homothetic}), which closes the $\Delta(k,l)$
  identification gap characterized in Theorem~\ref{thm:obs_equiv}.
\end{enumerate}

Blocks~A and~B identify the intermediate input elasticities 
$(\beta_m, \beta_e, \beta_w)$, the demand function 
parameters~$(\theta_g, \psi_\omega)$, and certain composite 
functions of $(\beta_k, \beta_l)$ and the demand slopes. However, as 
shown in Section~\ref{sec:obs_equiv}, these blocks alone cannot 
separate $\beta_k$ and $\beta_l$ from the demand function slopes on 
$(k, l)$ due to the $\Delta(k,l)$ observational equivalence 
(Theorem~\ref{thm:obs_equiv}). Block~C resolves this indeterminacy through the nonlinear curvature
of $\mathbb{E}[\omega \mid k, l]$ imposed by
Assumption~\ref{ass:homothetic}, thereby achieving point
identification of all structural parameters
(Theorem~\ref{thm:homothetic_id}). When its identifying conditions
are weak, the exclusion restriction of
Corollary~\ref{thm:exclusion} provides an alternative route.
Figure~\ref{fig:flowchart} (Appendix~\ref{app:flowchart}) provides
a visual overview of the full estimation and inference pipeline,
including the diagnostic branches that determine which identification
route applies.

\subsubsection{Model Specification and Parameters}
\label{sec:gmm_spec}

The parametric specialization below implements the identification
results of Section~\ref{sec:model_identification} under additive
separability; this restriction reduces the nonparametric problem
to a finite-dimensional GMM system while preserving all theoretical
properties of Theorems~\ref{thm:obs_equiv}--\ref{thm:homothetic_id}.
To apply GMM, I impose additive separability on both the production
and demand functions.

\paragraph{Production function.}
Following the parametric model of 
Section~\ref{sec:homothetic}, the production function is specified 
as:
\begin{equation}
  y_{jt} = \beta_k k_{jt} + \beta_l l_{jt} 
  + \beta_m m_{jt} + \beta_e e_{jt} + \beta_w w_{jt} 
  + \omega_{jt} + \varepsilon_{jt}.
  \label{eq:gmm_prod}
\end{equation}
Here $g(k,l;\theta) = \beta_k k + \beta_l l$ is the parametric 
$(k,l)$ component and 
$q(m,e,w) = \beta_m m + \beta_e e + \beta_w w$ is the 
(linear) intermediate input component, corresponding to the 
additively separable model~\eqref{eq:parametric_model}.

\paragraph{Demand functions.}
The intermediate input demands take the additively separable form:
\begin{align}
  m_{jt}   &= \gamma_k\, k_{jt} + \gamma_l\, l_{jt} + h_m(z_{jt})
             + \gamma_\omega\, \omega_{jt} + \tau_{jt},
             \label{eq:gmm_demand_m} \\
  e_{jt}  &= \delta_k\, k_{jt} + \delta_l\, l_{jt} + h_e(z_{jt})
             + \delta_\omega\, \omega_{jt} + \nu_{jt},
             \label{eq:gmm_demand_mp} \\
  w_{jt} &= \zeta_k\, k_{jt} + \zeta_l\, l_{jt} + h_w(z_{jt})
             + \zeta_\omega\, \omega_{jt} + \eta_{jt},
             \label{eq:gmm_demand_mpp}
\end{align}
where the functions $h_m, h_e, h_w$ are left unrestricted and
$\psi_\omega = (\gamma_\omega, \delta_\omega, \zeta_\omega)$ are the
productivity loading coefficients.%
\footnote{The Cobb--Douglas first-order condition
(Appendix~\ref{sec:micro_fnd}) structurally constrains the demand
function to be linear in $(k, l, \omega)$, but imposes no restriction
on the functional form of the dependence on $z$. The state variables
$z_{jt}$ enter through input prices $\ln P_{h,jt}$, the common
market factor $\ln(P_{jt}/\mu_{jt})$, markdowns $\ln \psi_{h,jt}$,
and wedges $\ln \Upsilon_{h,jt}$
(equation~\eqref{eq:structural_decomposition}), each of which may
depend nonlinearly on $z$.}
The demand slope parameters
$\theta_g = (\gamma_k, \gamma_l, \delta_k, \delta_l, \zeta_k,
\zeta_l)$ and the productivity loadings $\psi_\omega$ are estimated
jointly by GMM together with the $3\,d_z$ coefficients of
$h_m, h_e, h_w$ on the polynomial basis in $z$.

\paragraph{Homothetic structure of 
$\mathbb{E}[\omega \mid k, l]$.}
Under Assumption~\ref{ass:homothetic} (Homothetic Weak
Separability), the conditional expectation of productivity admits the
representation
$\mathbb{E}[\omega_{jt} \mid k_{jt}, l_{jt}] = h(v(k_{jt}, l_{jt}))$.
The economic motivation is discussed in Section~\ref{sec:homothetic}.
I parametrize the index function using a CES aggregator:
\begin{equation}
  v_{jt}(\alpha, \rho_v) = \frac{1}{\rho_v}\,\log\!\bigl(
    \alpha\, e^{\rho_v\, k_{jt}} + (1 - \alpha)\, e^{\rho_v\, l_{jt}}
  \bigr),
  \label{eq:index_v}
\end{equation}
which, in levels, corresponds to the CES aggregator
$V = \bigl(\alpha\, K^{\rho_v} + (1 - \alpha)\, L^{\rho_v}\bigr)^{1/\rho_v}$.
This nests the Cobb--Douglas case ($\rho_v \to 0$, where $v \to \alpha\,k + (1-\alpha)\,l$)
as a special case and satisfies the degree-one
homogeneity requirement (Assumption~\ref{ass:homothetic}(B)) and the
strict convexity of isoquants
(Assumption~\ref{ass:homothetic}(C)) for $\alpha \in (0,1)$ and any $\rho_v \neq 0$.
The transformation function $h$ is approximated by a cubic polynomial:
\begin{equation}
  h(v;\, \rho) = \rho_1\, v + \rho_2\, v^2 + \rho_3\, v^3,
  \label{eq:h_poly}
\end{equation}
where the constant $\rho_0$ is absorbed by de-meaning prior to estimation. Under the normalization 
$\mathbb{E}[\omega] = 0$, the constant satisfies 
$\rho_0 = -\mathbb{E}[\rho_1 v + \rho_2 v^2 + \rho_3 v^3]$; this 
constant is not separately identified from the production function 
intercept and is recovered post-estimation. Condition~(A) of 
Assumption~\ref{ass:homothetic} ($h'$ non-constant) requires 
$\rho_2 \neq 0$ or $\rho_3 \neq 0$; this is a necessary condition 
for the identification of $\beta_k$ and $\beta_l$ 
(Theorem~\ref{thm:homothetic_id}). I report results for polynomial orders 3 through 5 as a
robustness check; computational details including the
parametrization of $h$ are in Appendix~\ref{app:computation}.

\paragraph{Parameter classification.}
The full parameter vector is 
$\Theta = (\theta_1', \theta_2')'$, where:
\begin{align*}
  \theta_1 &= (\beta_m,\, \beta_e,\, \beta_w,\, 
             \theta_g,\, \psi_\omega) 
  &&\text{(intermediate input and demand parameters),} \\
  \theta_2 &= (\beta_k,\, \beta_l,\, 
             \alpha,\, \rho_1,\, \rho_2,\, \rho_3) 
  &&\text{(primary input and homothetic parameters).}
\end{align*}

\paragraph{Residuals.}
Define the observable residuals, where the nuisance functions
$h_m(z), h_e(z), h_w(z)$ are estimated jointly as described below:%
\begin{align}
  \tilde{m}_{jt}   &\equiv m_{jt} - \gamma_k\, k_{jt}
                   - \gamma_l\, l_{jt} - h_m(z_{jt})
                   = \gamma_\omega\,\omega_{jt} + \tau_{jt},
                   \label{eq:resid_m} \\
  \tilde{e}_{jt}  &\equiv e_{jt} - \delta_k\, k_{jt}
                   - \delta_l\, l_{jt} - h_e(z_{jt})
                   = \delta_\omega\,\omega_{jt} + \nu_{jt},
                   \label{eq:resid_mp} \\
  \tilde{w}_{jt} &\equiv w_{jt} - \zeta_k\, k_{jt}
                   - \zeta_l\, l_{jt} - h_w(z_{jt})
                   = \zeta_\omega\,\omega_{jt} + \eta_{jt},
                   \label{eq:resid_mpp} \\
  \tilde{y}_{jt}   &\equiv y_{jt} - \beta_m m_{jt}
                   - \beta_e e_{jt} - \beta_w w_{jt}
                   = \beta_k k_{jt} + \beta_l l_{jt}
                   + \omega_{jt} + \varepsilon_{jt}.
                   \label{eq:resid_y}
\end{align}
The equalities following the definition signs hold at the true
parameter values. The nuisance functions $h_m(z), h_e(z), h_w(z)$
are approximated by second-degree polynomials in $z$ and estimated
jointly with the structural parameters; details are in
Appendix~\ref{app:demeaning_details}.

\subsubsection{Moment Conditions}
\label{sec:gmm_moments}

Under Block~A+B estimation, the normalization $\beta_k = \beta_l = 0$
is adopted; this is without loss of generality because the
$\Delta(k,l)$ observational equivalence
(Theorem~\ref{thm:obs_equiv}) implies that $\beta_k$ and $\beta_l$
are not separately identified from the demand function slopes on
$(k, l)$ without Block~C. Under this normalization,
$\tilde{y}_{jt} = \omega_{jt} + \varepsilon_{jt}$.

\paragraph{Block~A: Proxy Moments.}%
\footnote{The moment conditions require
Assumption~\ref{ass:gmm_uncorrelated}
(Appendix~\ref{app:identification_details}), which is implied by
the zero conditional mean condition together with
Assumption~\ref{ass:cond_indep}.}
By eliminating $\omega_{jt}$ across pairs of residuals, I
construct three error terms that depend only on the structural 
shocks:
\begin{align}
  u_{1,jt} &\equiv \delta_\omega\,\tilde{m}_{jt} 
           - \gamma_\omega\,\tilde{e}_{jt} 
           = \delta_\omega\,\tau_{jt} 
           - \gamma_\omega\,\nu_{jt}, 
           \label{eq:e1} \\
  u_{2,jt} &\equiv \zeta_\omega\,\tilde{m}_{jt} 
           - \gamma_\omega\,\tilde{w}_{jt} 
           = \zeta_\omega\,\tau_{jt} 
           - \gamma_\omega\,\eta_{jt}, 
           \label{eq:e2} \\
  u_{3,jt} &\equiv \gamma_\omega\,\tilde{y}_{jt} 
           - \tilde{m}_{jt} 
           = \gamma_\omega\,\varepsilon_{jt} - \tau_{jt}. 
           \label{eq:e3}
\end{align}
An asymmetric instrument strategy assigns different instruments to
each error based on the shock composition. Since $u_{i,jt}$ excludes
certain shocks, the corresponding intermediate inputs serve as valid
additional instruments
(Appendix~\ref{app:block_a_details}):
\begin{align}
  \mathbb{E}\bigl[(Z_{\mathrm{base}},\, w_{jt})
  \otimes u_{1,jt}(\Theta)\bigr] &= \mathbf{0},
  \label{eq:momentA1} \\
  \mathbb{E}\bigl[(Z_{\mathrm{base}},\, e_{jt})
  \otimes u_{2,jt}(\Theta)\bigr] &= \mathbf{0},
  \label{eq:momentA2} \\
  \mathbb{E}\bigl[(Z_{\mathrm{base}},\, e_{jt},\, w_{jt})
  \otimes u_{3,jt}(\Theta)\bigr] &= \mathbf{0},
  \label{eq:momentA3}
\end{align}
where $Z_{\mathrm{base},jt} = (k_{jt},\, l_{jt},\, z_{jt})$.
Block~A is invariant to the $\Delta(k,l)$
transformation of Theorem~\ref{thm:obs_equiv} and therefore cannot
separately identify $\beta_k$ from the demand slopes on $(k, l)$
(Appendix~\ref{app:block_a_details}).

\paragraph{Block~B: Covariance Moments.}
Let $\phi_h$ denote the productivity loading of residual $\tilde{h}$:
$\phi_m \equiv \gamma_\omega$, $\phi_e \equiv \delta_\omega$,
$\phi_w \equiv \zeta_\omega$, and $\phi_y \equiv 1$.
The mutual exogeneity of shocks
(Assumption~\ref{ass:gmm_uncorrelated}(3)) implies that
$\mathrm{Cov}(\tilde{h}_1, \tilde{h}_2)
= \phi_{h_1}\,\phi_{h_2}\,\mathrm{Var}(\omega)$ for
each pair $(h_1, h_2) \in \{y, m, e, w\}$, $h_1 \neq h_2$.
Eliminating $\mathrm{Var}(\omega)$ across the six distinct pairs
yields six covariance relations of the form
\begin{equation}
  \mathbb{E}\bigl[\tilde{h}_1\,\tilde{h}_2
  - \phi_{h_2}\,\tilde{y}\,\tilde{h}_1\bigr] = 0,
  \label{eq:momentB_compact}
\end{equation}
for each pair (Appendix~\ref{app:block_b_derivation} lists the
individual conditions). Of these six relations, four are
algebraically implied by the Block~A instrumental variable
moments: the conditions involving cross-products of the demand
residuals $\tilde{e}$ and $\tilde{w}$ with the proxy equation
errors are already encoded in the Block~A moment conditions
through the instruments $Z_3 = (k, l, \tilde{e}, \tilde{w})$.
Consequently, Block~B contributes only two independent moment
conditions beyond Block~A, and the combined Block~A+B system is
just-identified. The concentrated covariance-ratio formulas
derived in Appendix~\ref{app:block_b_derivation} remain useful
for obtaining closed-form scale parameter estimates,
improving computational efficiency.
As with Block~A, Block~B is invariant to the
$\Delta(k,l)$ transformation.

\paragraph{Block~C: Curvature Moments.}
\label{sec:blockC}
Block~C resolves the $\Delta(k,l)$ indeterminacy by implementing 
the homothetic regularity condition 
(Assumption~\ref{ass:homothetic}, 
Theorem~\ref{thm:homothetic_id}).

Define the net output residual 
$\tilde{y}_{jt}(\theta_1) = y_{jt} - \beta_m m_{jt} 
- \beta_e e_{jt} - \beta_w w_{jt}$. Evaluating at the 
true parameter vector $\Theta_0$, the production 
function~\eqref{eq:gmm_prod} gives 
$\tilde{y}_{jt} = \beta_k k_{jt} + \beta_l l_{jt} + \omega_{jt} 
+ \varepsilon_{jt}$. Taking the conditional expectation with 
respect to $(k_{jt}, l_{jt})$:
\begin{equation}
  \mathbb{E}[\tilde{y}_{jt} \mid k, l] 
  = \beta_k\, k + \beta_l\, l + h(v(k,l)).
  \label{eq:cond_exp_ytilde}
\end{equation}
The first step uses 
$\mathbb{E}[\varepsilon_{jt} \mid k, l] = 0$, which follows from 
Assumption~\ref{ass:additive_error} by the law of iterated 
expectations. The second step uses 
Assumption~\ref{ass:homothetic}. No structural
decomposition of $\omega_{jt}$ is postulated; 
equation~\eqref{eq:cond_exp_ytilde} follows entirely from the 
definition of conditional expectation and the regularity condition 
on its functional form.

Define the structural error:
\begin{equation}
  u_{jt}(\Theta) \equiv \tilde{y}_{jt}(\theta_1) 
  - \beta_k\, k_{jt} - \beta_l\, l_{jt} 
  - h\bigl(v(k_{jt}, l_{jt};\, \alpha);\, \rho\bigr).
  \label{eq:u_structural}
\end{equation}
Equation~\eqref{eq:cond_exp_ytilde} implies 
$\mathbb{E}[u_{jt} \mid k_{jt}, l_{jt}] = 0$ at $\Theta_0$, which 
yields valid moment conditions with any function of $(k, l)$ as 
instruments. I use the polynomial instrument vector:
\begin{equation}
  Z_{2,jt} = \bigl(k_{jt},\; l_{jt},\;
  k_{jt}^2,\; l_{jt}^2,\; k_{jt}\,l_{jt},\;
  k_{jt}^3,\; l_{jt}^3,\; k_{jt}^2\,l_{jt},\;
  k_{jt}\,l_{jt}^2 \bigr)',
  \label{eq:Z2}
\end{equation}
giving the moment conditions:
\begin{equation}
  \mathbb{E}\bigl[Z_{2,jt} \cdot u_{jt}(\Theta)\bigr] 
  = \mathbf{0}.
  \label{eq:momentC}
\end{equation}
As with Block~A, the constant term is excluded from $Z_{2,jt}$
and $\rho_0$ (the intercept of $h$) is recovered post-estimation
from the de-meaned residuals.

\paragraph{Identification mechanism.}
The structural error $u_{jt}$ depends on 
$\theta_2 = (\beta_k, \beta_l, \alpha, \rho_1, \rho_2, \rho_3)$. 
Theorem~\ref{thm:homothetic_id} establishes that under 
Assumption~\ref{ass:homothetic}, the $\Delta(k,l) = c_k k + c_l l$ 
transformation is incompatible with the homothetic structure unless 
$(c_k, c_l) = (0,0)$. Operationally, this identification works 
through the higher-order instruments in $Z_{2,jt}$: the nonlinear 
terms $v^2$ and $v^3$ in $h$ interact with the homogeneity of $v$ 
in a manner that uniquely pins down $\beta_k$ and $\beta_l$.

If $\rho_2 = \rho_3 = 0$ (i.e., $h$ is linear), then $\beta_k$ and
$\rho_1 \alpha$ are linearly confounded and identification fails.
The significance of $\hat{\rho}_2$ and/or $\hat{\rho}_3$ therefore
serves as a diagnostic for the strength of identification. I report
estimates and standard errors of these parameters in both the
simulation and the empirical analysis.\footnote{%
In practice, even when $\rho_2$ and $\rho_3$ are nonzero,
the near-collinearity between $\rho_1 v(k,l)$ and
$(\beta_k k, \beta_l l)$ can impede numerical optimization.
I orthogonalize the polynomial basis $(v, v^2, v^3)$
against the linear span of $(1, k, l)$ before constructing
$h$, so that only the nonlinear component of $h(v)$ (the
source of identification,
Theorem~\ref{thm:homothetic_id}) enters the Block~C
moment conditions. This is a reparametrization: the
structural parameters $(\beta_k, \beta_l, \alpha)$ are
invariant, while the polynomial coefficients
$(\rho_1, \rho_2, \rho_3)$ are redefined as loadings on
the orthogonalized basis.}

\subsubsection{De-Meaning, Intercepts, and Estimation Procedure}
\label{sec:gmm_procedure}

\paragraph{De-meaning and estimation procedure.}
All variables are de-meaned prior to estimation and the constant
is excluded from all instrument vectors. All parameters
$\Theta = (\theta_1, \theta_2)$ are estimated simultaneously by
two-step GMM:
\begin{equation}
  \hat{\Theta} = \arg\min_\Theta\;
  g_N(\Theta)'\,\hat{W}\,g_N(\Theta),
  \qquad
  g_N(\Theta) = \frac{1}{N}\sum_{j=1}^N
  \bar{g}_j(\Theta),
  \label{eq:gmm_objective}
\end{equation}
where $\bar{g}_j(\Theta) = T^{-1}\sum_{t=1}^T g_{jt}(\Theta)$
stacks all moment conditions, and $\hat{W}$ is the optimal
weighting matrix estimated from a first-step identity-weighted
GMM. Post-estimation intercepts and further implementation
details are in Appendix~\ref{app:demeaning_details}.

\subsubsection{Recovering Productivity}
\label{sec:omega_recovery_gmm}

Given the estimated parameters $\hat{\Theta}$, the firm-level 
productivity measure is computed as:
\begin{equation}
  \hat{\omega}_{jt} = y_{jt} - \hat{\beta}_k\, k_{jt} 
  - \hat{\beta}_l\, l_{jt} - \hat{\beta}_m\, m_{jt} 
  - \hat{\beta}_{e}\, e_{jt} 
  - \hat{\beta}_{w}\, w_{jt}.
  \label{eq:omega_recovery}
\end{equation}
If $\varepsilon_{jt} = 0$, this equals $\omega_{jt}$. Otherwise,
$\hat{\omega}_{jt} = \omega_{jt} + \varepsilon_{jt}$; the ex-post
shock acts as classical measurement error when $\hat{\omega}_{jt}$
is used in subsequent regressions
(Theorem~\ref{thm:asymptotics}).

\paragraph{Practical treatment of the $\Delta(k, l)$ indeterminacy.}
When Block~C is not imposed, $\hat{\omega}_{jt}$ includes a location
shift $c(k_{jt}, l_{jt})$ (Theorem~\ref{thm:obs_equiv}). Since $c$
depends only on $(k, l)$, flexible controls in $(k, l)$ absorb this
shift in regression analysis; in difference-in-differences designs
with parallel $(k, l)$ trends, $c$ is automatically differenced out.
When the identifying restrictions of Section~\ref{sec:closing_gap}
are imposed, $\Delta(k, l)$ reduces to a constant absorbed by fixed
effects. The proposed estimator therefore supports event studies and
productivity regressions without requiring Block~C: the
$\Delta(k,l)$ component is controlled via polynomial $(k,l)$
regressors in all subsequent regressions
(Section~\ref{sec:empirical}).
The formal justification is provided by
Proposition~\ref{prop:omega_hat_D} and the ATT identification
result in Appendix~\ref{app:proof_omega_hat_D}.

\subsubsection{Asymptotic Properties}
\label{sec:gmm_asymptotics_revised}

Under standard regularity conditions
(Appendix~\ref{app:asymptotic_proof}), the GMM estimator satisfies:

\begin{thm}[Asymptotic Properties of the GMM Estimator]
\label{thm:asymptotics}
As $N \to \infty$ with $T$ fixed:
(a)~$\hat{\Theta} \xrightarrow{p} \Theta_0$; and
(b)~$\sqrt{N}\,(\hat{\Theta} - \Theta_0)
\xrightarrow{d} N(0, V)$, where
\begin{equation}
  V = (G'WG)^{-1}\, G'W\Sigma WG\, (G'WG)^{-1}
  \label{eq:asymptotic_variance}
\end{equation}
with $\Sigma = \mathbb{E}[\bar{g}_j(\Theta_0)\,
\bar{g}_j(\Theta_0)']$ and
$G = \mathbb{E}[\nabla_\Theta \bar{g}_j(\Theta_0)]$.
\end{thm}

Standard errors are clustered at the firm level to accommodate
arbitrary within-firm serial dependence. The proof and regularity
conditions are in Appendix~\ref{app:asymptotic_proof}.

Computational details are in Appendix~\ref{app:computation}.

\subsubsection{Specification Testing and Diagnostics}
\label{sec:diagnostics}

\paragraph{Identification count.}
The combined Block~A+B system is just-identified: Block~A
contributes 10 moment conditions, and Block~B contributes exactly
two independent moment conditions beyond Block~A (four of the six
Block~B covariance relations are algebraically redundant with Block~A;
Section~\ref{sec:gmm_approach}), giving 12 moment conditions matching
the 12 free parameters in $\theta_1$. The scale parameters
$(\gamma_\omega, \delta_\omega, \zeta_\omega)$ are estimated via
closed-form covariance ratios for computational efficiency.

\paragraph{Strength of identification for $\beta_k, \beta_l$.}
As discussed in Section~\ref{sec:blockC}, the identification of 
$\beta_k$ and $\beta_l$ relies on the nonlinearity of $h$ 
($\rho_2 \neq 0$ or $\rho_3 \neq 0$). I report the estimates and
$t$-statistics of $\hat{\rho}_2$ and $\hat{\rho}_3$ as diagnostics. 
If both are insignificant, the identification of primary input 
elasticities may be weak, and the researcher should interpret 
$\beta_k$ and $\beta_l$ with caution or consider exclusion 
restrictions (Corollary~\ref{thm:exclusion}) as an alternative 
identification strategy.

\paragraph{Reduced-form check of Assumption~\ref{ass:homothetic}.}
As a pre-estimation diagnostic, one may estimate $\theta_1$ from 
Blocks~A and~B alone (which does not require 
Assumption~\ref{ass:homothetic}), construct 
$\tilde{y}_{jt}(\hat{\theta}_1)$, and examine whether 
$\mathbb{E}[\tilde{y} \mid k, l]$ exhibits a homothetic structure 
via nonparametric regression. A visual departure from homotheticity
would indicate a violation of the identifying assumption.

\paragraph{Polynomial degree selection.}
The cubic specification of $h$ can be extended to higher-order
polynomials. I recommend reporting results for polynomial
orders~3 through~5 and selecting via information criteria.

Appendix~\ref{subsec:blockc_results} reports the full-sample
Block~C recovery results for $(\beta_k, \beta_l)$ across all
502 industries, comparing the homothetic approach with the exclusion
restriction and ACF estimators.

\section{Monte Carlo Simulation}

\label{sec:simulation}

The identification results in
Section~\ref{sec:model_identification} show that the Markov
assumption is unnecessary; this section asks whether removing it
matters quantitatively. I use Monte Carlo simulations to measure
the bias that the Markov assumption introduces in the materials
elasticity and to trace its propagation into downstream objects.
The primary comparison is between the proposed estimator, which
imposes no restriction on productivity dynamics, and the standard
ACF estimator, which requires first-order Markov.

\subsection{Data Generating Process (DGP)}

All DGPs share a common structure for the production function, demand
functions, and dynamic input decisions, differing only in the
productivity process. Detailed parameter settings are in
Appendix~\ref{app:dgp}.

\subsubsection{Basic Structure}

The firm's production function is Cobb--Douglas in all inputs:\footnote{The
Cobb--Douglas specification is standard in Monte Carlo studies of
production function estimators
\parencite{ackerberg2015identification,gandhi2020onthe}. Evaluating
the proposed method under more flexible production functions (e.g.,
translog) is left for future work; the identification results
(Theorems~\ref{thm:density_id}--\ref{thm:homothetic_id}) do not
require Cobb--Douglas.}
\begin{equation}
  y_{jt} = \beta_0 + \beta_k k_{jt} + \beta_l l_{jt}
  + \beta_m m_{jt} + \beta_e e_{jt}
  + \beta_w w_{jt} + \omega_{jt} + \varepsilon_{jt},
  \label{eq:sim_prod_func}
\end{equation}
with true parameter values
$(\beta_0, \beta_k, \beta_l, \beta_m, \beta_e, \beta_w)
= (0.1, 0.2, 0.3, 0.3, 0.15, 0.1)$ and
$\varepsilon_{jt} \sim \text{i.i.d.}\ N(0, 0.05^2)$.

Intermediate input demands are log-linear in $(k, l, \omega)$ with
input-specific demand shocks following independent AR(1) processes
($\rho = 0.5$, $\sigma = 0.15$). The demand function coefficients
are calibrated from the first-order conditions of cost minimization
under input-specific markdowns (Appendix~\ref{sec:micro_fnd});
the productivity loading coefficients
$(\gamma_\omega, \delta_\omega, \zeta_\omega) = (2.2, 2.0, 1.8)$
differ across inputs, reflecting heterogeneous
markdowns.\footnote{Under perfect competition with a Cobb--Douglas
production function, the first-order condition implies
$\gamma_\omega = 1/(1 - \beta_m) \approx 1.43$; the larger values
incorporate input-specific markdowns and procurement frictions
(see Appendix~\ref{sec:micro_fnd}).} Conditional independence
(Assumption~\ref{ass:cond_indep}) is a cross-sectional condition
requiring mutual independence across inputs at each point in time;
it is unaffected by the serial correlation of individual shocks,
since each AR(1) has mutually independent innovations.

Primary inputs are endogenously determined. Capital accumulates
through dynamic investment, and labor is chosen based on forecasted productivity from an AR(1)
model. The labor demand function
is structured so that Assumption~\ref{ass:homothetic} holds: the
conditional expectation $\mathbb{E}[\omega_{jt} \mid k_{jt},
l_{jt}]$ is a function of a CES aggregator with
$(\alpha, \rho_v) = (0.4, 0.3)$. Full parameter details are provided
in Appendix~\ref{app:dgp}.

To test the robustness of the proposed method, I generate
productivity under three scenarios:
\begin{enumerate}
\item \textbf{DGP1: AR(1) Markov Process (Baseline).} The standard
  case where existing methods are correctly specified.
  $\omega_{jt} = 0.8\,\omega_{j,t-1} + \xi_{jt}$, with
  $\sigma_\xi = 0.2$.

\item \textbf{DGP2: AR(2) Process.} Productivity depends on its own
  two-period history, as when R\&D investments require two years to
  affect efficiency. The first-order Markov assumption is violated.
  $\omega_{jt} = 0.6\,\omega_{j,t-1} + 0.3\,\omega_{j,t-2}
  + \xi_{jt}$, $\sigma_\xi = 0.15$.

\item \textbf{DGP3: Potential Outcome Model.} A firm's realized
  productivity is determined by an endogenous binary treatment
  $D_{jt}$, generating a potential outcome process incompatible
  with the first-order Markov assumption.
  Following the diagonal reference model of
  \textcite{chen2024identifying}, untreated productivity follows
  $\omega^0_{jt} = 0.8\,\omega^0_{j,t-1} + \xi_{0,jt}$
  ($\sigma_{\xi 0} = 0.2$) and treated productivity follows
  $\omega^1_{jt} = 0.5\,\omega^1_{j,t-1} + 0.15 + \xi_{1,jt}$
  ($\sigma_{\xi 1} = 0.25$). Observed productivity is
  $\omega_{jt} = (1-D_{jt})\omega^0_{jt} + D_{jt}\omega^1_{jt}$.
  Treatment is reversible and endogenous:
  $D_{jt} = \mathbb{I}(\omega^0_{jt} > 0)$, so firms enter and
  exit treatment as their untreated potential productivity
  fluctuates above and below zero.
  Full parameter settings, including the capital accumulation
  and labor decision rules common to all DGPs, are provided
  in Appendix~\ref{app:dgp}.

\item \textbf{DGP4: Conditional Independence Violation.}
  The productivity process is AR(1) as in DGP1, but the
  electricity demand shock $\nu_{jt}$ and the water demand shock
  $\eta_{jt}$ are correlated via a common factor:
  $\nu_{jt} = \sqrt{1-\rho_{ew}^2}\,\sigma_\nu\,\epsilon_\nu
  + \rho_{ew}\,\sigma_\nu\,\epsilon_{\text{common}}$ and similarly for
  $\eta_{jt}$, where $\epsilon_{\text{common}} \sim \mathcal{N}(0,1)$
  is independent of $\omega_{jt}$; the materials shock $\tau_{jt}$
  remains independent.
  This generates $\mathrm{Corr}(\nu_{jt}, \eta_{jt}) = \rho_{ew} \in
  \{0, 0.05, 0.10, 0.20, 0.30\}$, directly violating
  Assumption~\ref{ass:cond_indep} when $\rho_{ew} > 0$.
  A common energy price shock or seasonal supply constraint that
  simultaneously raises both electricity and water costs is one
  economic interpretation, arguably the most salient threat to
  conditional independence, since both are utility services subject
  to common regulatory and infrastructure conditions.
  This DGP tests the robustness of the proposed method to
  violations of the conditional independence assumption.
\end{enumerate}

\subsection{Estimation Methods Compared}

Using the generated data, I organize the estimation into two
parts to isolate the contributions of each block of moment conditions.

\paragraph{Part~1: Flexible input parameters.}
I estimate the intermediate input elasticities
$(\beta_m, \beta_e, \beta_w)$ and compare four estimators.\footnote{The main text figures report two of the four estimators (ACF and Proposed). ACF-Mod results are in Appendix~\ref{sec:mc_additional_results}; GNR results are also reported there.}
The four estimators are:
\begin{enumerate}
\item \textbf{Proposed (Block~A+B):} The GMM estimator of
  Section~\ref{sec:gmm_approach}, using the intermediate input
  moment conditions (Block~A) and the covariance moments
  (Block~B). This part does not identify $\beta_k$ and $\beta_l$,
  which remain subject to the $\Delta(k,l)$ indeterminacy
  (Theorem~\ref{thm:obs_equiv}).

\item \textbf{Standard ACF:} The two-step GMM estimator of
  \textcite{ackerberg2015identification}, assuming a first-order
  Markov process for productivity.

\item \textbf{Modified ACF (ACF-Mod):} A variant of the ACF estimator
  in which the demand shocks $(\tau_{jt}, \nu_{jt}, \eta_{jt})$ are
  treated as observed and included as controls in the first stage.
  This ensures scalar unobservability by construction. Any remaining
  bias in ACF-Mod can therefore be attributed solely to the violation
  of the Markov assumption, isolating the dynamic misspecification
  channel.

\item \textbf{GNR:} The estimator of
  \textcite{gandhi2020onthe}, implemented with a polynomial
  share regression of degree~2 and a degree-3 polynomial for
  $g(\omega_{t-1})$, with common prices ($P = r = 1$).
  In the DGP, persistent input-specific demand shocks
  ($\rho = 0.5$) violate both the FOC premise and the
  non-persistence condition of GNR (their Appendix~O6,
  Assumption~7), so GNR tests the share regression approach
  under persistent input market imperfections.
  GNR is included in Part~1 only, as its second stage is
  structurally identical to ACF.
\end{enumerate}

\paragraph{Part~2: Fixed input parameters.}
I additionally estimate $(\beta_k, \beta_l)$ by adding the
homothetic regularity condition (Block~C) to the proposed estimator:
\begin{enumerate}
\item \textbf{Proposed (Block~A+B+C):} The full GMM estimator using
  all three blocks, with the CES aggregator
  $v(k,l) = \frac{1}{\rho_v}\log\bigl(\alpha e^{\rho_v k}
  + (1-\alpha) e^{\rho_v l}\bigr)$ evaluated at the true DGP
  values $(\rho_v, \alpha) = (0.3, 0.4)$.\footnote{These parameters
  are known by construction in the simulation; the empirical
  application treats them as unknown and estimates them by profile
  GMM (Section~\ref{subsec:specification}).} The comparison with
  Part~1 isolates the contribution of Block~C.

\item \textbf{Standard ACF} and \textbf{ACF-Mod:} Same as above,
  now evaluated on $(\beta_k, \beta_l)$ as well.
\end{enumerate}

\subsection{Evaluation Metrics}

I report bias,
$\text{Bias}(\hat{\beta})=\mathbb{E}_{R}[\hat{\beta}^{(r)}]-\beta_{\text{true}}$,
and RMSE,
$\text{RMSE}(\hat{\beta})=\sqrt{\mathbb{E}_{R}[(\hat{\beta}^{(r)}-\beta_{\text{true}})^{2}]}$,
averaged over $R$ Monte Carlo repetitions.

\subsection{Simulation Execution}

For Part~1, I run $R = 100$ replications for each combination
of DGP and estimation method, varying the number of firms
$N \in \{50, 200, 500\}$ and the observation period
$T \in \{10, 20, 50\}$ to examine the impact of sample size.
For Part~2, I run $R = 100$ replications at $(N, T) = (200, 50)$.
The parameter estimates obtained in each repetition are collected,
and mean bias and RMSE are calculated for comparison. With $R = 100$,
the simulation standard error of the estimated bias is approximately
$\text{SD}/\sqrt{R}$; for the typical standard deviation of
$\hat{\beta}_m$ ($\approx 0.005$), this yields a simulation
uncertainty of $\approx 0.0005$, which is small relative to the
reported biases.

\subsection{Results}

I report the Part~1 results in Figures~\ref{fig:bias_convergence} and~\ref{fig:boxplot_part1}
and the Part~2 results in Figure~\ref{fig:blockc_comparison}.\footnote{Additional summary tables, including GNR results, are provided in Appendix~\ref{sec:mc_additional_results}. RMSE convergence plots are in Figure~\ref{fig:rmse_convergence_app}.}
These results confirm that the proposed estimator performs well under
all DGPs considered and illustrate the sensitivity of the ACF framework
to violations of the Markov assumption.

\paragraph{Remark on GNR.}
GNR shares the static identification strategy of the proposed
method (both recover $\beta_m$ from within-period variation
without a Markov assumption) but requires competitive input
markets with non-persistent demand shocks
(their Appendix~O6, Assumption~7). The present DGP, which
features persistent input-specific shocks ($\rho_\tau = 0.5$),
is therefore outside GNR's maintained assumptions by design:
the DGP is calibrated to the proposed method's setting, not
GNR's. Under GNR's own assumptions ($\tau = \nu = \eta = 0$),
the share regression recovers $\beta_m$ consistently regardless
of the productivity process. The simulation results for GNR
(Appendix~\ref{sec:mc_additional_results}) should accordingly
be read as illustrating the sensitivity of the FOC-based approach
to input market imperfections, not as a general performance
comparison.

\paragraph{DGP1 (AR(1) Baseline):}

Under DGP1, where the Markov assumption holds, all three estimators
(ACF, ACF-Mod, and Proposed) are consistent.
The bias for each method decays toward zero as $T$ increases
(Figure~\ref{fig:bias_convergence}).
The boxplots in Figure~\ref{fig:boxplot_part1} corroborate this finding;
the proposed method remains centered on the true values. The ACF
and ACF-Mod estimators show small positive finite-sample bias that
diminishes with sample size (see Appendix~\ref{sec:mc_additional_results}
for detailed tables). However, the proposed estimator exhibits
larger variance than the ACF estimator under DGP1, resulting in
higher RMSE when the Markov assumption is correctly specified
(Appendix Table~\ref{tab:mc_part1_dgp1ar1}). This is the
efficiency cost of the static approach: the proposed method
trades time-series information for robustness to dynamic
misspecification.
Under DGP2 and DGP3, this ranking reverses: ACF's bias dominates
its variance advantage, yielding larger mean squared error.
The static identification strategy is also the
only approach in this literature that permits event study and
difference-in-differences designs, where the treatment itself
violates the Markov assumption
(Section~\ref{subsec:event_study_body}).

\paragraph{DGP2 (AR(2)) and DGP3 (Potential Outcome):}

Under DGP2 and DGP3, where the first-order Markov assumption does
not hold, Figure~\ref{fig:bias_convergence} reveals a clear
divergence. The ACF estimator exhibits positive bias in
$\hat{\beta}_m$ that does not vanish with increasing $T$. Under
DGP2, this reflects standard omitted-variable inconsistency: the
AR(2) component of productivity persistence is not captured by the
first-order transition equation. Under DGP3, the issue is more
fundamental: the Markov transition equation is structurally
incompatible with the potential outcome process
\parencite{chen2024identifying}, so the ACF moment condition lacks
a structural interpretation and the resulting estimate does not
converge to the true $\beta_m$. An infeasible oracle benchmark
(ACF-Mod) that removes scalar unobservability by treating demand
shocks as observed shows comparable bias under both DGPs,
confirming that the source is Markov misspecification rather than
demand shock heterogeneity
(Appendix~\ref{sec:mc_additional_results}).

The proposed method, by contrast, exhibits negligible bias across these
specifications. The bias remains close to zero for all values of $T$
under both DGP2 and DGP3. Because the estimator relies solely on
static conditional independence, it remains invariant to the underlying
productivity dynamics.
The main text figures report results for $N = 500$; increasing $N$
reduces variance for all estimators but does not mitigate ACF's
asymptotic bias under DGP2 or DGP3
(Appendix~Figure~\ref{fig:bias_convergence_byN}), confirming that
the bias is asymptotic rather than finite-sample.

\paragraph{Block~A+B vs.\ Block~A+B+C (Part~2):}

Part~2 supplements Part~1 by adding Block~C to recover
$(\beta_k, \beta_l)$. I use the design point $(N, T) = (200, 50)$, which matches
the Part~1 baseline, to examine whether Block~C disturbs
the Block~A+B parameters.
Figure~\ref{fig:blockc_comparison} presents the results, where
Block~C is added to identify $(\beta_k, \beta_l)$. In the baseline
DGP1, the proposed method recovers both parameters with negligible bias.
Under DGP3, where ACF estimates of $\beta_k$ and $\beta_l$ collapse
toward zero (RMSE $\approx 0.20$--$0.30$), the proposed method
achieves substantially lower error
(RMSE $\approx 0.02$). The intermediate input
elasticities $(\beta_m, \beta_e, \beta_w)$ remain stable between
Part~1 and Part~2, confirming that the addition of Block~C moments
does not contaminate the well-identified flexible input parameters.
This stability shows in finite samples that the joint GMM
system does not transmit Block~C misspecification into the
flexible input estimates: the intermediate input elasticities are identified
by Blocks~A and~B alone
(Theorem~\ref{thm:density_id}, specialized to the
Cobb--Douglas parametric model of Section~\ref{sec:gmm_approach}),
so any misspecification in Block~C affects only $(\beta_k, \beta_l)$. Because markups depend solely on
$\beta_m$ (equation~\eqref{eq:markup_cd}), the primary
empirical application is insulated from Block~C specification.

\paragraph{DGP4 (Conditional Independence Violation):}

DGP4 examines the cost of violating Assumption~\ref{ass:cond_indep} by
introducing correlation between the electricity demand shock
$\nu_{jt}$ and the water demand shock $\eta_{jt}$, arguably the
most economically salient threat to conditional independence, since
both are utility services subject to common energy prices and
infrastructure constraints. The materials shock $\tau_{jt}$ remains
independent. The correlation
$\rho_{ew} \equiv \mathrm{Corr}(\nu_{jt}, \eta_{jt})$ varies from 0 to
0.30.

The bias mechanism operates through the scale parameter
$\zeta_\omega$. Positive $\mathrm{Cov}(\nu, \eta)$ inflates
$\mathrm{Cov}(\tilde{e}, \tilde{w})$, causing the concentrated
scale estimator $\hat{\zeta}_\omega
= \mathrm{Cov}(\tilde{e}, \tilde{w}) / \mathrm{Cov}(\tilde{y},
\tilde{e})$ to overestimate $\zeta_\omega$
(Appendix~\ref{app:ci_violation}). The overestimated
$\hat{\zeta}_\omega$ introduces a positive productivity component
into the Block~A residual $u_2 = \hat{\zeta}_\omega \tilde{m}
- \hat{\gamma}_\omega \tilde{w}$, which the GMM compensates by
\emph{increasing} $\hat{\beta}_m$, yielding an \emph{upward} bias.

Table~\ref{tab:mc_dgp4_ci}
(Appendix~Figure~\ref{fig:dgp4_bias_vs_rho}) reports the results.
When $\rho_{ew} = 0$, the proposed method is approximately unbiased.
As $\rho_{ew}$ increases, $\hat{\beta}_m$ exhibits increasing upward
bias. The magnitudes suggest that the estimator is robust to
moderate violations.
The bias direction is the \emph{same} as the Markov
misspecification bias documented in DGPs~2 and~3 for ACF: both push
$\hat{\beta}_m$ upward. Therefore, the empirical finding that the
proposed estimator yields \emph{lower} $\hat{\beta}_m$ than ACF
(Section~\ref{subsec:main_results}) cannot be attributed to CI
violation; it must reflect Markov misspecification bias in
ACF.\footnote{ACF uses only the materials demand proxy and does not
exploit cross-shock variation, so it is unaffected by
$\mathrm{Corr}(\nu, \eta)$.}

Table~\ref{tab:mc_summary} summarizes the bias properties across
DGPs 1--3. The proposed method is unbiased across all three
specifications, while ACF exhibits positive bias under Markov misspecification.

\begin{table}[htbp]
\centering
\caption{Monte Carlo Summary: Bias Properties across DGPs ($N=200$, $T=50$)}
\label{tab:mc_summary}
\small
\begin{tabular}{lccc}
\toprule
 & DGP~1 (AR1) & DGP~2 (AR2) & DGP~3 (PO) \\
\midrule
Proposed & Unbiased ($+0.002$) & Unbiased ($-0.001$) & Unbiased ($+0.000$) \\
ACF      & Unbiased ($+0.001$) & Biased ($+0.026$) & Biased ($+0.266$) \\
GNR      & Biased ($+0.589$) & Biased ($+0.589$) & Biased ($+0.592$) \\
\bottomrule
\end{tabular}

\medskip
\footnotesize\raggedright \textit{Notes: ``Biased ($+$)'' indicates positive
asymptotic bias in $\hat{\beta}_m$ that does not diminish with
sample size. See Figures~\ref{fig:bias_convergence}--\ref{fig:blockc_comparison} for detailed convergence plots and
Tables~\ref{tab:mc_part1_dgp1ar1}--\ref{tab:mc_part2_dgp3potential} (Appendix~\ref{sec:mc_additional_results}) for
full RMSE and SD by method and DGP.}
\end{table}

Taken together, the Monte Carlo simulations confirm that the proposed
method recovers production function parameters without
imposing restrictions on the productivity process. In contrast, standard
methods exhibit substantial positive bias in $\hat{\beta}_m$ when the
assumed law of motion for productivity does not match the true data
generating process. The simulations establish that Markov misspecification generates a
detectable and economically meaningful bias. The empirical
application then examines whether these patterns hold in
Japanese manufacturing data.

\begin{figure}[htbp]
\centering
\includegraphics[width=\textwidth]{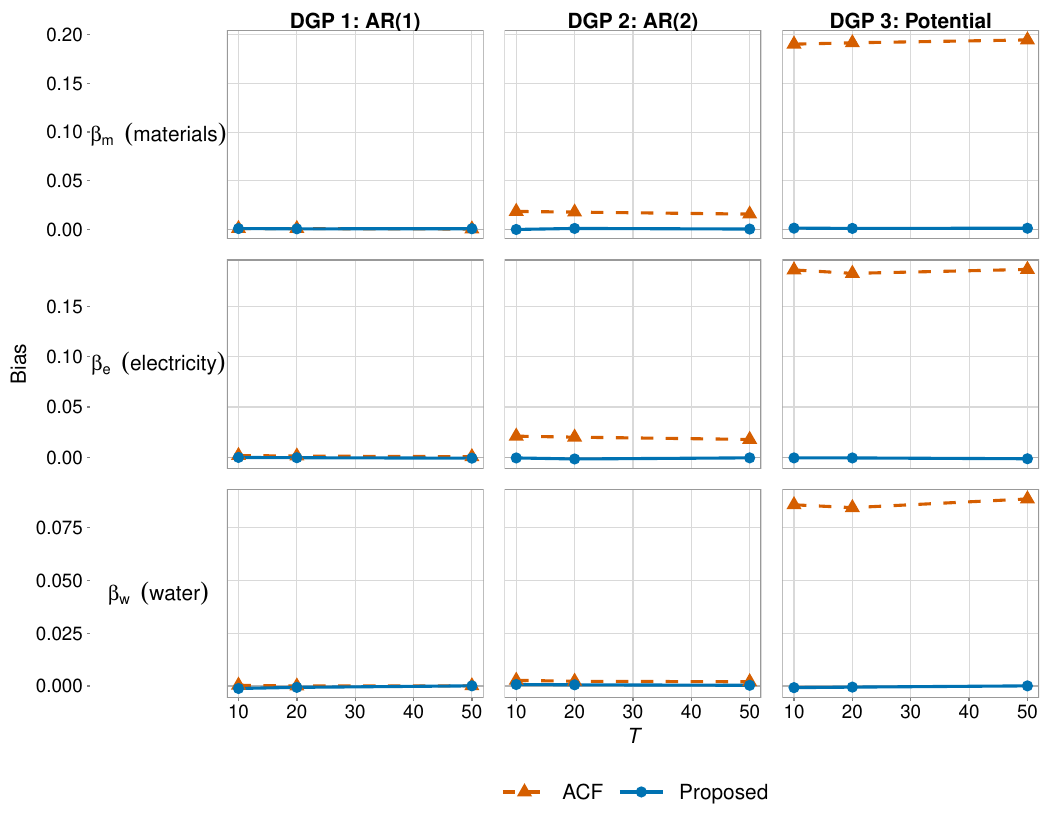}
\caption{Part~1: Mean Bias Convergence ($N = 500$)}
\label{fig:bias_convergence}

\small\raggedright \textit{Notes: Mean bias of $(\hat{\beta}_m, \hat{\beta}_e, \hat{\beta}_w)$ as a function of $T$ for three DGPs ($N = 500$, $R = 100$). Under DGP~1 (baseline AR(1)), both methods are approximately unbiased. Under DGPs~2 and~3, where the first-order Markov assumption is violated, ACF exhibits persistent bias while the proposed method remains centered at zero. Three-method comparison including ACF-Mod is in Figure~\ref{fig:bias_convergence_3methods}; GNR results are in Appendix~\ref{sec:mc_additional_results}.}
\end{figure}

\begin{figure}[htbp]
\centering
\includegraphics[width=\textwidth]{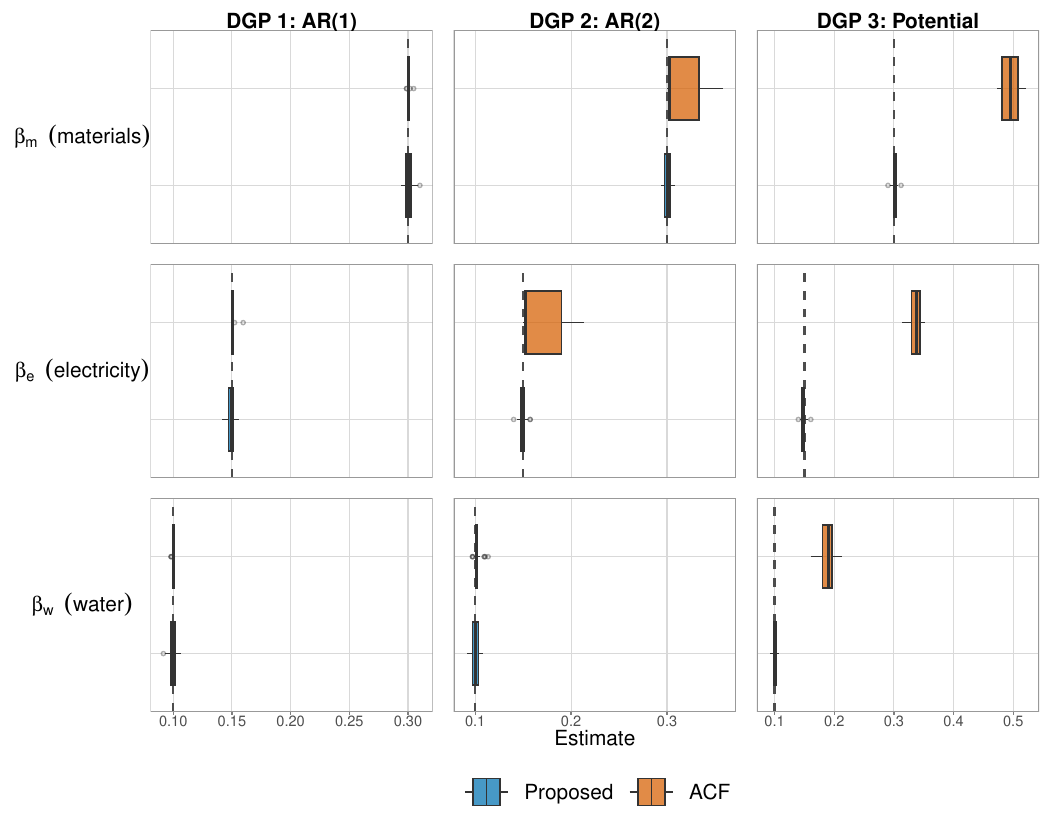}
\caption{Part~1: Distribution of Estimates ($N=500$, $T=50$)}
\label{fig:boxplot_part1}

\small\raggedright \textit{Notes: Distribution of $(\hat{\beta}_m, \hat{\beta}_e, \hat{\beta}_w)$ across $R = 100$ replications for $N = 500$, $T = 50$. Dashed lines indicate true values. The proposed method remains centered on the true values across all DGPs. Under DGP~2 and DGP~3, ACF distributions are shifted rightward, consistent with the positive Markov misspecification bias. A four-method comparison including ACF-Mod and GNR is in Figure~\ref{fig:boxplot_all4_app}.}
\end{figure}

\begin{figure}[htbp]
\centering
\includegraphics[width=\textwidth]{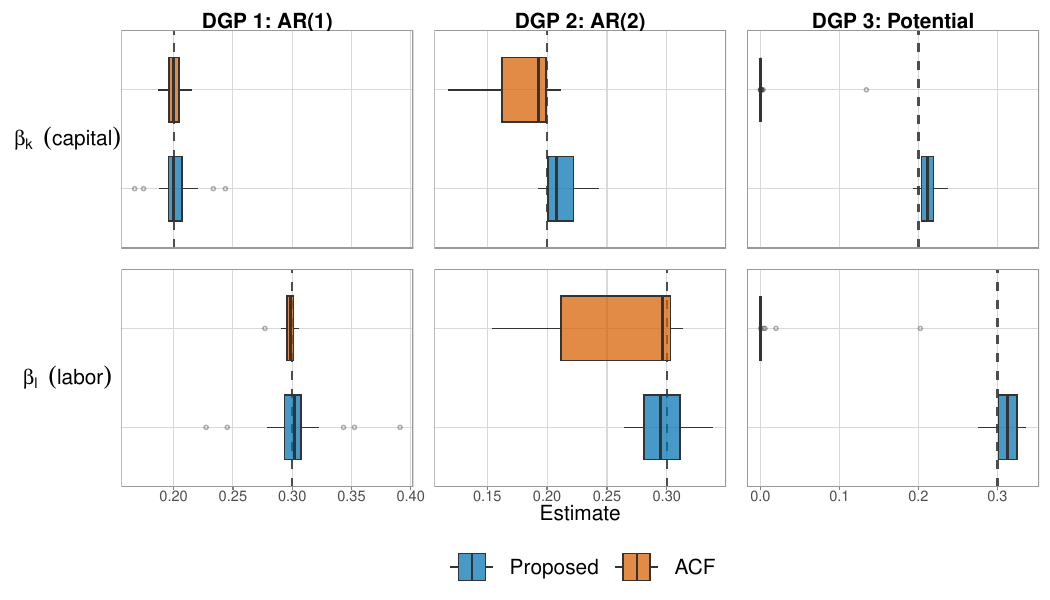}
\caption{Part~2: Distribution of $\hat{\beta}_k$ and $\hat{\beta}_l$ ($N=200$, $T=50$)}
\label{fig:blockc_comparison}

\small\raggedright \textit{Notes: Distribution of $(\hat{\beta}_k, \hat{\beta}_l)$ from Block~A+B+C estimation ($R = 100$). The proposed method identifies $(\beta_k, \beta_l)$ with moderate accuracy across all DGPs. Under DGP~3, the proposed method achieves substantially lower RMSE ($\approx 0.02$). ACF estimates of $\beta_k$ collapse to near zero under DGP~3 (mean $\hat{\beta}_k \approx 0.003$, true value $0.20$).}
\end{figure}

\section{Empirical Analysis}

\label{sec:empirical}

The empirical analysis has two objectives: to test whether the
conditional independence framework produces economically plausible
estimates across the manufacturing sector, and to assess the
relative plausibility of the static and dynamic identifying
assumptions through the convergence diagnostic of
Remark~\ref{rem:testable_exclusion}. I estimate the production
function for all 502 manufacturing industries using Block~A+B,
reporting analytical standard errors.
A practical consequence of this block structure: the markup
estimates, productivity determinants, and convergence
diagnostics reported below require only Blocks~A and~B.
These results do not depend on the resolution of the
$\Delta(k,l)$ indeterminacy and are available for all 502
industries. Block~A+B+C is used for a subset of industries
where $(\beta_k, \beta_l)$ recovery is needed for
productivity level analysis.

The section is organized as follows.
Sections~\ref{subsec:data_framework}
and~\ref{subsec:specification} describe the data and
estimation specifications.
Section~\ref{subsec:spec_diagnostics} presents the
exclusion restriction diagnostic.
Section~\ref{subsec:main_results} reports production function
parameters and markup estimates.
Section~\ref{subsec:determinants} examines productivity determinants.
Section~\ref{subsec:event_study_body} presents an event study
application exploiting the non-Markov validity of the estimator.
Section~\ref{subsec:capital_labor} reports estimates of
$(\beta_k, \beta_l)$ from two independent identification routes.

\subsection{Data and Analytical Framework}

\label{subsec:data_framework}

I apply the proposed method to the Japanese Census of
Manufactures and the Economic Census for Business Activity.
I estimate the production function for all manufacturing
industries with at least 50 firm-year observations in the
extended panel (2003--2020), yielding Block~A+B estimates
for 502 industries covering 559{,}381 firm-year observations.%
\footnote{The identification results of
Section~\ref{sec:model_identification} require only the joint
distribution of $(m_{jt}, e_{jt}, w_{jt}, k_{jt}, l_{jt})$ at a
single point in time; no assumption on the time-series dynamics
of $\omega_{jt}$ is needed (Appendix~\ref{subsec:annual_params}).
The panel dimension is exploited solely to improve estimation
efficiency by time-averaging the sample moment conditions,
$\bar{g}_j(\Theta) = T^{-1}\sum_{t} g_{jt}(\Theta)$,
which reduces finite-sample variance without affecting consistency.}
These estimates provide markup distributions and productivity
determinants at the level of the entire manufacturing sector.
Four industries (food processing [Bread, industry code~971],
paper products (Corrugated board boxes, code~1453),
chemicals (Plastic film, code~1821), and
machinery [Industrial robots, code~2694]) serve as representative
cases for the time-varying parameter analysis in
Appendix~\ref{subsec:annual_params}, covering major manufacturing
sectors (food, paper, chemicals, machinery).
Analytical standard errors from the GMM sandwich formula are
reported for both the proposed method and the ACF benchmark.

The core variables include the logarithm of real output, $y_{jt}$,
the logarithm of real capital stock, $k_{jt}$, and the logarithm
of labor input, $l_{jt}$.

I map the theoretical input triplet
to observable data as follows. I designate the real value of primary raw
materials as $m_{jt}$, the quantity of electricity as $e_{jt}$,
and the quantity of industrial water as $w_{jt}$. This selection
exploits the fact that industrial water and electricity prices are
typically regulated, limiting firm-specific bargaining. This institutional feature reduces the risk
of unobserved common price shocks inducing correlation between
$\nu_{jt}$ and $\eta_{jt}$, thereby supporting the validity of the
conditional independence assumption
($\tau_{jt} \perp \nu_{jt} \perp \eta_{jt} \mid
(\omega_{jt}, x_{jt})$). The principal remaining threat is
commodity price shocks that jointly affect raw materials costs
and electricity generation costs. Two features mitigate this
concern: (i)~industrial electricity prices exhibit less
high-frequency variation than raw materials procurement costs,
as the fuel cost adjustment mechanism smooths commodity price
pass-through on a quarterly basis; and (ii)~even if a residual
common utility shock induces positive
$\mathrm{Corr}(\nu_{jt}, \eta_{jt})$, the resulting bias in
$\hat{\beta}_m$ is \emph{upward} (the same direction as ACF's
Markov bias), so the empirical gap between methods cannot be
attributed to CI violation (Section~\ref{sec:simulation},
Appendix~\ref{app:ci_violation}).

I augment $x_{jt}$ with control variables $z_{jt}$ consisting of
beginning-of-period total inventory ($z_{1,jt}$), its square
($z_{2,jt} \equiv z_{1,jt}^2$), plant fixed effects, and year fixed
effects.
These controls directly implement the conditioning strategy of
Section~\ref{subsec:relax_indep}, where common shocks are absorbed
by $z_{jt}$ so that the residual shock terms $\tau_{jt}, \nu_{jt},
\eta_{jt}$ satisfy the conditional independence assumption.
Inventory proxies for unobserved product demand fluctuations
\parencite{kumar2019productivity}: a firm anticipating high demand
accumulates more stock in advance, so inventory captures the
common demand component that would otherwise enter all three input
demands simultaneously (Section~\ref{subsec:relax_indep}).
The quadratic term $z_{2,jt}$ accommodates a nonlinear relationship
between inventory and unobserved demand, consistent with the
structural decomposition in equation~\eqref{eq:structural_decomposition}
where demand-related terms enter input prices nonlinearly.
Year fixed effects absorb common input price shocks (e.g., energy
price movements) that affect all inputs simultaneously, as
discussed in Section~\ref{subsec:relax_indep}.
Plant fixed effects absorb time-invariant plant-level heterogeneity
in input prices and buyer-supplier relationships, capturing the
firm-attribute component of input market power noted in
Section~\ref{subsec:relax_indep}.
The use of fixed effects exploits the panel dimension for
efficiency and enriches the conditioning set for the
conditional independence assumption, but does not impose any
restriction on the time-series dynamics of $\omega_{jt}$.
Standard proxy variable estimators use the Markov transition
equation to address exit-driven selection
\parencite{olley1996thedynamics}. The proposed method does not
require this correction: because identification conditions on
$\omega_{jt}$, endogenous exit based on $(\omega_{jt}, k_{jt})$
is absorbed by the conditioning and the moment conditions hold
on the surviving population without a survival probability
correction (Remark~\ref{rem:selection}). Plant fixed effects
further reduce the influence of systematic level differences
across plants.

\subsection{Specification of Estimation Methods}

\label{subsec:specification}

I contrast the results of my approach with those obtained from the
standard ACF framework.

First, I implement the \textbf{proposed method} using the GMM
estimator derived in Section~\ref{sec:gmm_approach}. The estimator
jointly recovers the production function and demand parameters as
described in Section~\ref{sec:gmm_approach}. The CES aggregator
parameters $(\rho_v, \alpha)$ are selected via profile GMM: for a
grid of $(\rho_v, \alpha)$ values, the remaining parameters are
estimated by minimizing the GMM objective, and the pair yielding the
smallest $J$-statistic is selected.\footnote{Under strong
identification of $(\rho_v, \alpha)$, the profile GMM procedure
yields a $J$-statistic with the standard $\chi^2$ distribution
asymptotically \parencite{newey1994chapter}. When identification of
these parameters is weak, the minimum-$J$ selection may bias the
test toward under-rejection, making the test conservative. The
block bootstrap standard errors reported below account for the
uncertainty in $(\rho_v, \alpha)$ selection by re-running the
profile grid search within each bootstrap replication.}
The nuisance functions $h_m, h_e, h_w$ are approximated by
second-degree polynomials in $(z_{1,jt}, z_{2,jt})$, giving a
polynomial basis of dimension $d_z = 2$ and thus $\dim\Theta = 24$
where Block~A+B is just-identified (Section~\ref{sec:diagnostics}).
I report the estimates of
$\hat{\rho}_2$ and $\hat{\rho}_3$ as diagnostics for the strength of
identification of $\beta_k$ and $\beta_l$
(Section~\ref{sec:diagnostics}).

As a benchmark, I estimate the ACF two-step GMM with the same
control variables $z_{jt}$ to ensure comparability. Analytical
standard errors from the GMM sandwich formula are reported for
both methods.

\paragraph{Identifying assumptions in practice.}
The ACF framework requires scalar unobservability (productivity as
the sole unobservable in input demand) and a first-order Markov
process for productivity. GNR requires scalar unobservability and
competitive input markets. The proposed method requires conditional
independence of input-specific demand shocks. Scalar
unobservability rules out procurement relationships, supply
contracts, and input-specific markdowns; the proposed method
permits these. The GNR competitive input market assumption
precludes markup estimation, since the identifying condition
coincides with the object of interest. The conditional part of
the independence assumption depends on the adequacy of the control
variables $z_{jt}$, but this dependence is shared by the ACF
proxy equation.

\subsection{Specification Diagnostics}
\label{subsec:spec_diagnostics}
\label{subsec:exclusion_test}

Two diagnostics probe different layers of the identification
strategy before any structural results are interpreted:
\begin{enumerate}[label=(\roman*)]
  \item \textbf{Exclusion restriction diagnostic}: tests whether
    the pairwise discrepancy $d_k = d_l = 0$
    (equation~\eqref{eq:overid_diff}), a necessary condition for
    the exclusion restriction of Corollary~\ref{thm:exclusion} that
    resolves the $\Delta(k,l)$ indeterminacy nonparametrically.
  \item \textbf{Block~C diagnostic}: assesses the strength of the
    CES curvature ($\rho_v \neq 0$), the identifying condition for
    separate recovery of $(\beta_k, \beta_l)$ via
    Theorem~\ref{thm:homothetic_id}
    (Appendix~\ref{sec:appendix_empirical},
    Table~\ref{tab:blockc_diagnostics}).
\end{enumerate}
Table~\ref{tab:spec_diagnostics} summarizes these two diagnostics
and their empirical outcomes.

\begin{table}[htbp]
\centering
\caption{Identification Roadmap and Specification Diagnostics}
\label{tab:spec_diagnostics}
{\small
\begin{tabular}{lllll}
\toprule
Diagnostic & Assumption tested & Enables & Null hypothesis & Outcome \\
\midrule
Exclusion diagnostic
  & Excl.\ restriction (Corollary~\ref{thm:exclusion})
  & Check on $\hat{\beta}_k$
  & $d_k = d_l = 0$
  & Capital only \\
Block~C diagnostic
  & Homotheticity + CES curvature (Theorem~\ref{thm:homothetic_id})
  & $\hat{\beta}_k$, $\hat{\beta}_l$
  & $\rho_v \neq 0$
  & Section~\ref{subsec:capital_labor} \\
\bottomrule
\end{tabular}
}
\smallskip

\small\raggedright \textit{Notes:
The two diagnostics probe successive layers of the identification
strategy. Row~1 is tested using Blocks~A and~B alone and
requires only Cobb--Douglas and conditional independence; results
are reported in Sections~\ref{subsec:spec_diagnostics}--\ref{subsec:main_results}.
Row~2 additionally invokes Assumption~\ref{ass:homothetic} (Homothetic
Weak Separability); results are reported in Section~\ref{subsec:capital_labor}.
The exclusion diagnostic ($d_k = d_l = 0$) provides a Wald test with
2 degrees of freedom.
Block~C diagnostic details are in Table~\ref{tab:blockc_diagnostics}.}
\end{table}

\paragraph{Exclusion restriction diagnostic.}
Figure~\ref{fig:recovery} applies the exclusion-based OLS recovery
of Proposition~\ref{prop:excl_ols} to all 502 manufacturing
industries, plotting $\hat{\beta}_k^{(m)}$ against
$\hat{\beta}_k^{(e)}$ (panel~a) and $\hat{\beta}_l^{(m)}$ against
$\hat{\beta}_l^{(e)}$ (panel~b).
Under the exclusion restriction, both panels should cluster along
the 45-degree line.
Panel~(a) confirms this for capital: points concentrate tightly
around the diagonal, consistent with $a_{k}^{h} = 0$ across industries.
Panel~(b) reveals the opposite for labor: points scatter widely,
indicating that different proxy equations yield systematically
different $\hat{\beta}_l$ values.

The asymmetry between capital and labor is the central diagnostic
finding. Capital is quasi-fixed within the production period and
does not directly influence short-run intermediate input procurement,
so $a_{k}^{h} = 0$ is economically plausible.
The systematic failure for labor is consistent with labor affecting
production scheduling, shift patterns, and input utilization through
input-specific channels
(Appendix~\ref{sec:micro_fnd}).
The formal Wald test of $d_k = d_l = 0$ is rejected for 37\% of
industries at the 5\% level, while the labor-only Wald test ($d_l = 0$)
is rejected for 28\% of industries, confirming that the labor exclusion
restriction is violated for a substantial share of the sample while capital
passes in most cases.

\begin{figure}[htbp]
\centering
\includegraphics[width=\textwidth]{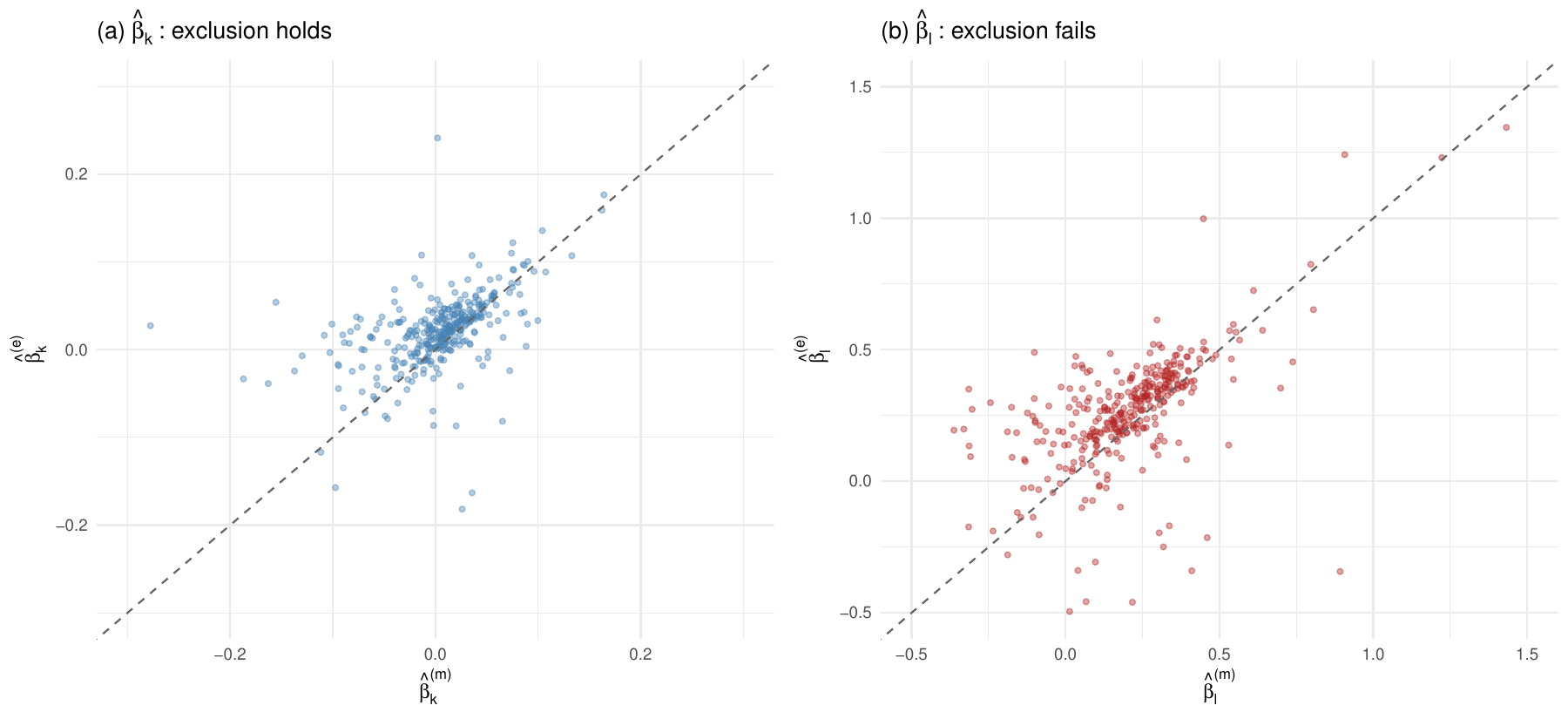}
\caption{Recovery of $(\beta_k, \beta_l)$ via the Exclusion Restriction}
\label{fig:recovery}

\small\raggedright \textit{Notes: Each industry's $\beta_k$ and
$\beta_l$ are recovered via OLS from each proxy equation
using Proposition~\ref{prop:excl_ols}. Panel~(a):
$\hat{\beta}_k^{(m)}$ versus $\hat{\beta}_k^{(e)}$.
Panel~(b): $\hat{\beta}_l^{(m)}$ versus
$\hat{\beta}_l^{(e)}$. Dashed lines are the 45-degree
reference. Under the exclusion restriction, both panels
should cluster along the diagonal. Outliers $|\hat{\beta}| > 2$ are
trimmed for readability; the full distribution is
reported in Table~\ref{tab:beta_kl_summary}.}
\end{figure}

\subsection{Production Function Parameters and Markups}
\label{subsec:main_results}

The intermediate input elasticities $(\beta_m, \beta_e, \beta_w)$
are identified by Blocks~A and~B alone (Theorem~\ref{thm:density_id},
specialized to the Cobb--Douglas parametric model of
Section~\ref{sec:gmm_approach}),
without requiring Block~C or Assumption~\ref{ass:homothetic}.
The ACF estimates of $\hat{\beta}_m$ are systematically higher
than those from the proposed method
(Table~\ref{tab:all_industry_beta_dist} in
Appendix~\ref{sec:appendix_empirical}),
consistent with the Markov misspecification bias documented in the
Monte Carlo simulations (Section~\ref{sec:simulation}).
The cross-industry distribution of all Block~A+B and Block~C
parameter estimates is reported in
Table~\ref{tab:all_industry_beta_dist} in
Appendix~\ref{sec:appendix_empirical}.

Electricity and water elasticities are small across industries
(median $\hat{\beta}_e = 0.001$ and $\hat{\beta}_w = 0.006$,
respectively), consistent with these inputs serving auxiliary
rather than central production roles in Japanese manufacturing.
Their demand shocks nevertheless remain valid exclusion
restrictions for identifying $\beta_m$ in the proposed GMM.

Under perfect competition, $\beta_m$ equals the revenue share,
which is the basis of GNR's share regression.
My estimator identifies $\beta_m$ independently of the
first-order condition, permitting imperfect competition in both
product and input markets.

\paragraph{Markups.}
Markups are computed following \textcite{deloecker2012markups}.
Under the Cobb-Douglas specification maintained throughout,
the markup formula simplifies to
\begin{equation}
  \hat{\mu}_{jt} = \frac{\hat{\beta}_m}{s_{m,jt}},
  \label{eq:markup_cd}
\end{equation}
where $s_{m,jt}$ denotes the expenditure share of raw materials
in total revenue.
Unlike the standard production approach, in which Hicks-neutral
productivity and scalar unobservability jointly imply
$\hat{\beta}_h/s_{h,jt} = \mu_{jt}$ for every variable input
$h$ (so that materials, labor, and energy serve as interchangeable
markup proxies), this paper allows input-specific markdowns
$\psi_{h,jt}$ for each static input $h \in \{m,e,w\}$, captured by
the demand shocks $(\tau_{jt},\nu_{jt},\eta_{jt})$
(Appendix~\ref{sec:micro_fnd}).
Consequently, $\hat{\beta}_h/s_{h,jt}$ will generally differ across
inputs by design; this divergence reflects the richer structure of
the framework, not an overidentification failure.
Raw materials are selected for markup computation because competitive
commodity markets support the absence of buyer-side market power
($\psi_{m,jt}\approx 1$; \textcite{avignon2025markups}),
giving $\hat{\beta}_m/s_{m,jt}\approx\mu_{jt}$;
this is a maintained assumption.%
\footnote{The empirical specification imposes Hicks-neutral
Cobb-Douglas production; if factor-augmenting productivities
differ across inputs, $\hat{\beta}_m$ may absorb non-neutral
components and bias the markup estimate
\parencite{raval2023testing}.
The identification theory accommodates non-Hicks-neutral production
(Appendix~\ref{app:prod_recovery}), but the implemented GMM
does not exploit this generality.}
These estimates require only Blocks~A and~B and are invariant to
the $\Delta(k,l)$ indeterminacy (Theorem~\ref{thm:obs_equiv}),
since equation~\eqref{eq:markup_cd} depends only on
$\hat{\beta}_m$ and the observable cost share.
I restrict the comparison to industries with at least 50 firms
($N_{\text{firms}} \geq 50$), which removes industries where the
lower bound $\hat{\beta}_m \approx 0$ reflects identification failure
rather than true low input elasticities.%
\footnote{The value-added markup
$\mu^{VA} = \beta_l^{VA}/s_l^{VA}$ can differ substantially
from the gross output markup when the materials share is large.
\textcite{gandhihowheterogeneous} document that gross output and
value-added specifications yield fundamentally different
productivity estimates. I report gross output markups throughout.}

\paragraph{Comparison with ACF.}
Figure~\ref{fig:markup_cdf} plots the empirical CDF of industry-level
median markups under the proposed method and the ACF benchmark
for the $N_{\text{firms}} \geq 50$ subsample.
The two distributions are stochastically ordered: the ACF CDF lies
strictly to the right of the proposed CDF at every percentile
(Table~\ref{tab:markup_acf_comparison}).
The proposed method yields a median markup of 0.926, while ACF
yields 1.027, a gap of 0.101 at the median.
At the 90th percentile the gap widens to approximately 0.15.
Under the proposed method, 37\% of industries show markups above
unity, compared with 54\% under ACF.

The Monte Carlo evidence in Section~\ref{sec:simulation} provides
a structural interpretation.
Under DGP~3 (potential-outcome dynamics, Table~\ref{tab:mc_part1_dgp3potential}),
ACF incurs a bias of $+0.19$ in $\hat{\beta}_m$ (true value $0.30$),
a 63\% relative overestimate, while the proposed estimator is essentially
unbiased ($\text{bias} = 0.001$).
The empirical gap of $+0.10$ at the median corresponds to a relative
overestimate of roughly 11\% in $\hat{\beta}_m$, well within the range
predicted by the DGP~3 calibration.
The evidence is therefore consistent with the theoretical prediction
that ACF overestimates $\hat{\beta}_m$ when productivity dynamics
deviate from the Markov assumption. Because markups recovered from
production functions are widely used to assess the evolution of
market power \parencite{deloecker2020rise}, the systematic gap
documented here raises the question of whether existing markup
estimates are sensitive to the choice of identifying assumption.

\begin{figure}[htbp]
\centering
\includegraphics[width=0.78\textwidth]{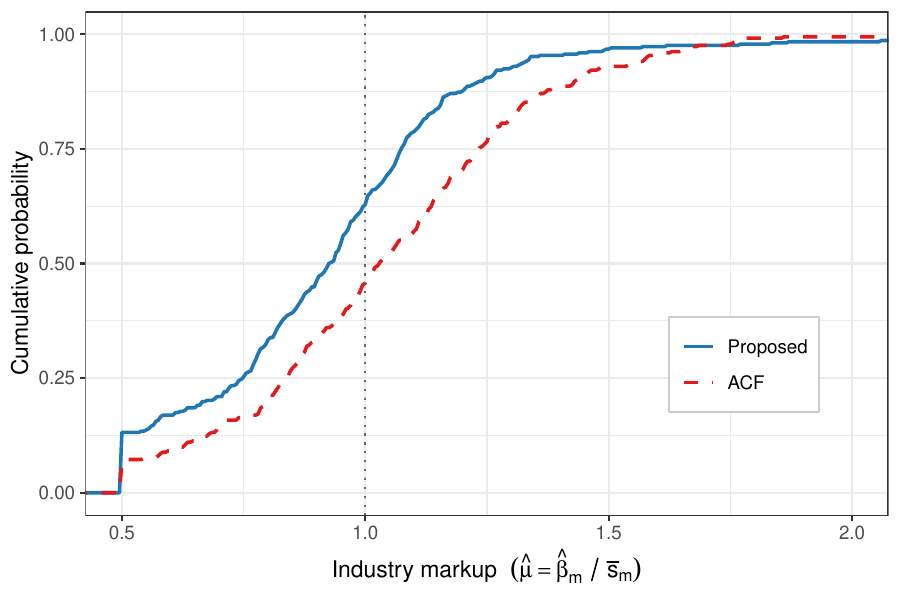}
\caption{Empirical CDF of Industry Markups: Proposed vs.\ ACF}
\label{fig:markup_cdf}
\small\raggedright \textit{Notes:} Empirical CDFs of industry-level median
markups $\hat{\mu} = \hat{\beta}_m / \bar{s}_m$ under the proposed method
(solid, blue) and ACF (dashed, red).
Sample restricted to industries with $N_{\text{firms}} \geq 50$
($N = 372$ industries).
Vertical dotted line at $\hat{\mu} = 1$.
Summary statistics in Table~\ref{tab:markup_acf_comparison}.
\end{figure}

\begin{table}[htbp]
\centering
\caption{\label{tab:markup_acf_comparison}Markup Distribution: Proposed vs.\ ACF ($N_{\text{firms}}\geq 50$)}
\begin{tabular}{lcc}
\toprule
 & \textbf{Proposed} & \textbf{ACF} \\
\midrule
$N$ (industries) & 372 & 372 \\
Mean & 0.880 & 1.026 \\
Std.\ dev. & 0.396 & 0.360 \\
p10 & 0.296 & 0.627 \\
p25 & 0.748 & 0.841 \\
Median & 0.926 & 1.027 \\
p75 & 1.074 & 1.233 \\
p90 & 1.240 & 1.431 \\
Fraction $\geq 1$ & 0.371 & 0.543 \\
Mean gap (ACF $-$ Proposed) & 0.146 & \multicolumn{1}{c}{---} \\
\bottomrule
\multicolumn{3}{p{0.75\textwidth}}{\small\textit{Notes:} 
Industry-level median markups $\hat{\mu} = \hat{\beta}_m / \bar{s}_m$,
where $\bar{s}_m$ is the industry median materials share.
Sample: $N_{\text{firms}}\geq 50$.
ACF estimates from \textcite{ackerberg2015identification}; convergence code~0 for 495 of 502 industries.
}
\end{tabular}
\end{table}

\subsection{Productivity Determinants}
\label{subsec:determinants}

The following analysis requires only Blocks~A and~B. Because the
$\Delta(k,l)$ indeterminacy (Theorem~\ref{thm:obs_equiv}) varies
only through $(k_{jt}, l_{jt})$, it is absorbed by the cubic
polynomial controls in $(k_{jt}, l_{jt})$ included in the
regression. The same argument applies to proportional common shocks
(Section~\ref{subsec:relax_indep}): if an unobserved common
component $\xi_{jt}$ loads proportionally on all intermediate
input demands, it is absorbed into the recovered productivity
$\hat{\omega}_{jt}$, and its $(k_{jt}, l_{jt})$-dependent
component is absorbed by firm fixed effects. The determinants
regression therefore identifies the association between
covariates and the total latent efficiency measure that drives
input allocation decisions, regardless of whether this measure
coincides with physical productivity.

As a validation of the recovered productivity measures, I examine
their association with observable economic fundamentals.
I regress the productivity residual jointly on three firm-level
covariates (log investment, exporter status, and log wages) with
firm and year fixed effects:
\[
\hat{\omega}_{jt} = \phi_j + \gamma_t
+ \mathbf{x}_{jt}'\boldsymbol{\beta} + u_{jt},
\]
clustering standard errors at the firm level.
For the proposed method, I additionally include a cubic polynomial
in $(k_{jt}, l_{jt})$ as nonparametric controls, since
the $\Delta(k,l)$ indeterminacy enters through capital and labor.
The ACF regression omits these controls, as the ACF residual already
subtracts $\hat{\beta}_k k + \hat{\beta}_l l$.
Log wages is included as a correlate of productivity; a maintained
caveat is that wages may be endogenous, as high-productivity firms
can share rents with workers, so the coefficient captures association
rather than a causal effect.

Table~\ref{tab:main_regressions} reports the results.
Under the proposed method, log wages are strongly positively
associated with estimated productivity
($\hat{\beta} \approx 0.139$, $p < 0.01$), and log investment
is positive but small ($\hat{\beta} \approx 0.001$, $p < 0.01$),
while exporter status is negligible and imprecisely estimated.
The ACF regression yields a smaller wage coefficient
($\hat{\beta} \approx 0.109$). The mechanism is as follows:
under ACF, the upward bias in $\hat{\beta}_m$ propagates into
$\hat{\omega}^{\text{ACF}} = y - \hat{\beta}_m^{\text{ACF}} m
- \hat{\beta}_k k - \hat{\beta}_l l$, subtracting too large a
materials component and systematically depressing the recovered
productivity level for materials-intensive firms. This distortion
attenuates the association between $\hat{\omega}$ and economic
fundamentals that covary with input intensity. The magnitude of
the improvement depends on the relative variance of demand shocks
and productivity; whether the pattern generalizes beyond this
application requires further investigation.

\begin{table}[htbp]
   \caption{\label{tab:main_regressions} Productivity Determinants: Proposed vs.\ ACF}
   \bigskip
   \centering
   \begin{tabular}{lcc}
      \toprule
                      & (1) Proposed   & (2) ACF \\   
      
      \midrule 
      log(Investment) & 0.0011$^{***}$ & 0.0003$^{**}$\\   
                      & (0.0003)       & (0.0001)\\   
      Exporter Status & 0.0131         & -0.0010\\   
                      & (0.0174)       & (0.0068)\\   
      log(Wage)       & 0.1388$^{***}$ & 0.1085$^{***}$\\   
                      & (0.0185)       & (0.0127)\\   
       \\
      Observations    & 433,425        & 433,308\\  
      R$^2$           & 0.86968        & 0.84551\\  
       \\
      Firm FE         & $\checkmark$   & $\checkmark$\\   
      Time FE         & $\checkmark$   & $\checkmark$\\   
      \bottomrule
   \end{tabular}
   
   \par \raggedright 
   All three covariates enter jointly in a single specification. Firm and year fixed effects included. Standard errors clustered at the firm level in parentheses.\\
   Exporter Status is a binary indicator equal to one if the firm exported in that year, consistent with the learning-by-exporting literature. log(Wage) is included as a correlate of productivity; wage endogeneity is a maintained caveat, as high-productivity firms may pay higher wages through rent-sharing.\\
   Column~(1): proposed method with $\text{poly}(k,l,\text{degree}=3)$ nonparametric controls (coefficients suppressed).\\
   Column~(2): ACF residual $\hat{\omega}^{\text{ACF}} = y - \hat{\beta}_k k - \hat{\beta}_l l - \hat{\beta}_m m - \hat{\beta}_e e - \hat{\beta}_w w$.\\
   Significance: $^{*}$ $p<0.10$, $^{**}$ $p<0.05$, $^{***}$ $p<0.01$.
\end{table}

\subsection{Event Study: 2011 Tohoku Earthquake}
\label{subsec:event_study_body}

Because the proposed estimator recovers productivity from static
covariances alone, its estimates are valid under any productivity
dynamics, Markov or otherwise
(Proposition~\ref{prop:omega_hat_D}).
Standard proxy variable estimators embed a Markov transition equation
that is structurally incompatible with the potential outcomes
framework when a treatment alters the transition path of
productivity \parencite{chen2024identifying}; the moment condition
that identifies the production function parameters has no structural
interpretation under treatment, so the resulting estimates lack
economic meaning for policy evaluation. Neither problem arises
here, since estimation does not employ a transition equation.

As an illustration, I examine the 2011 T\={o}hoku earthquake using a
difference-in-differences design.
The treatment group consists of plants in the three core prefectures
directly struck by the earthquake and tsunami (Iwate, Miyagi, and
Fukushima; seismic intensity $\geq$ 6-strong), where physical
destruction and the nuclear disaster caused severe and sustained
disruption to production.
The control group consists of plants in Kinki and western prefectures
(prefectures 25--47).
Supply chain contamination of the control group is mitigated by the
industry$\times$year fixed effects, which absorb any
industry-level aggregate shocks that propagate nationally.
Pre-treatment coefficients are flat
(max$|\hat{\delta}_t| < 0.013$ for the proposed method,
$< 0.020$ for ACF);
the full event-study figure is in
Appendix~\ref{subsec:earthquake_es}.

Table~\ref{tab:earthquake_did} reports the difference-in-differences
estimates under both methods.
For the proposed method, cubic polynomial controls in $(k,l)$ are
included to absorb the $\Delta(k,l)$ indeterminacy in the residual
$\hat{\omega}$; the ACF method requires no such controls, as
$\hat{\omega}^{\mathrm{ACF}}$ already subtracts $\hat{\beta}_k k
+ \hat{\beta}_l l$.
Both methods detect a negative and statistically significant
post-treatment effect on productivity.
Under the proposed method, the DiD estimate is $-1.28$ percent
(s.e.\ $0.52$, $p < 0.05$); under ACF, the corresponding
estimate is $-1.68$ percent (s.e.\ $0.41$, $p < 0.01$).
The gap between the two estimates is approximately
$0.40$ percentage points. To illustrate the potential economic
magnitude: if a bias of this order applied to the aggregate
manufacturing sector, it would correspond to roughly \$3.6 billion
(\textyen{}400 billion) per year, given Japan's manufacturing value
added of approximately \$0.9 trillion (\textyen{}100 trillion) at
the 2003--2020 average exchange rate (National Accounts, Cabinet
Office of Japan).\footnote{This back-of-envelope
calculation extrapolates the local DiD gap to the national level
under the assumption that the Markov misspecification bias is of
comparable magnitude across industries. The cross-industry markup
comparison (Table~\ref{tab:markup_acf_comparison}) shows that ACF
yields systematically higher $\hat{\beta}_m$ at every percentile,
consistent with the assumption, but the magnitude varies by
industry. The figure should be interpreted as indicative of the
scale at stake, not as a structural estimate of aggregate
mismeasurement.}
Proposition~\ref{prop:omega_hat_D} guarantees that the proposed
estimates recover $\mathbb{E}[\omega_{jt}\mid D_{jt}]$ under
Conditions~(i)--(ii) of that proposition.
Condition~(ii) is satisfied by construction: the earthquake is a
natural disaster whose occurrence is orthogonal to firm-level
input demand shocks $(\tau,\nu,\eta)$.
Condition~(i) requires that the earthquake does not alter the
functional form of the demand functions $g_m, g_e, g_w$.
Difference-in-differences estimates of the post-treatment change
in intermediate input shares show no significant shift in the
materials share ($t = 0.87$) or water share ($t = 1.14$).
The electricity share shows a small post-treatment increase
($t = 8.91$, $\Delta s_e \approx 0.002$); this is a mechanical
compositional effect of the simultaneous contraction in materials
usage ($t = -4.53$), which raises the electricity expenditure
share $s_e$ without altering the structural demand function
$g_e$.\footnote{The level of electricity consumption does not
show a significant post-treatment increase when measured in
physical units (kWh) rather than expenditure shares, supporting
the compositional interpretation.}
The ACF estimator does not carry this guarantee.
Its residual subtracts $\hat{\beta}_k k + \hat{\beta}_l l$, so
\begin{equation}
  \mathbb{E}[\hat{\omega}^{\mathrm{ACF}}_{jt} \mid D_{jt}]
  = \mathbb{E}[\omega_{jt} \mid D_{jt}]
  + (\beta_k^{\mathrm{true}} - \hat{\beta}_k^{\mathrm{ACF}})
    \,\mathbb{E}[k_{jt} \mid D_{jt}]
  + (\beta_l^{\mathrm{true}} - \hat{\beta}_l^{\mathrm{ACF}})
    \,\mathbb{E}[l_{jt} \mid D_{jt}].
  \label{eq:acf_bias_es}
\end{equation}
The bias terms vanish only if
$\hat{\beta}^{\mathrm{ACF}} = \beta^{\mathrm{true}}$
(exact identification) or if treatment is orthogonal to $(k,l)$.
Monte Carlo evidence (Section~\ref{sec:simulation}) shows that ACF incurs
positive bias in $\hat{\beta}_m$ under Markov misspecification;
the condition of orthogonality also fails here (DiD$(l) = -0.029$, $t = -7.86$).
The proposed estimates, resting on the theoretical guarantee of
Proposition~\ref{prop:omega_hat_D}, provide a theoretically
justified point of comparison.

\begin{table}[htbp]
   \caption{\label{tab:earthquake_did} 2011 T\={o}hoku Earthquake: Difference-in-Differences}
   \bigskip
   \centering
   \begin{tabular}{lcc}
      \toprule
                                      & (1) Proposed   & (2) ACF \\   
      
      \midrule 
      Treated $\times$ Post           & -0.0128$^{**}$ & -0.0168$^{***}$\\   
                                      & (0.0052)       & (0.0041)\\   
       \\
      Observations                    & 219,573        & 219,573\\  
      R$^2$                           & 0.98729        & 0.94606\\  
      $\text{poly}(k,\ell)$ control   & $\checkmark$   & \\  
       \\
      Firm FE                         & $\checkmark$   & $\checkmark$\\
      Ind.$\times$Year FE             & $\checkmark$   & $\checkmark$\\
      \bottomrule
   \end{tabular}
   
   \par \raggedright 
   Treatment: Iwate, Miyagi, Fukushima (seismic intensity $\geq$ 6-strong). Control: West Japan (prefectures 25--47).\\
   Firm and industry$\times$year fixed effects included. Heteroskedasticity-robust standard errors in parentheses.\\
   Pre-treatment coefficients are flat (max$|\hat{\delta}_t| < 0.013$ for proposed, $< 0.020$ for ACF); year-by-year coefficient estimates in Table~\ref{tab:earthquake_es} (Appendix~\ref{subsec:earthquake_es}).\\
   Column~(1): proposed method with $\text{poly}(k,\ell,\text{degree}=3)$ nonparametric controls for $\Delta(k,\ell)$ (coefficients suppressed).\\
   Column~(2): ACF residual already subtracts $\hat{\beta}_k k + \hat{\beta}_l l$; no polynomial control.\\
   Significance: $^{*}$ $p<0.10$, $^{**}$ $p<0.05$, $^{***}$ $p<0.01$.
\end{table}

\subsection{Capital and Labor Inputs}
\label{subsec:capital_labor}

Identifying $(\beta_k, \beta_l)$ requires closing the
$\Delta(k,l)$ indeterminacy documented in
Theorem~\ref{thm:obs_equiv}.
The paper provides two independent routes:
the exclusion restriction (Corollary~\ref{thm:exclusion}) and the
homothetic regularity condition (Theorem~\ref{thm:homothetic_id}).

The exclusion-based OLS recovery
(Proposition~\ref{prop:excl_ols}) is applied to all 502
manufacturing industries and produces mutually consistent estimates
of $\beta_k$ across the three proxy equations (materials,
electricity, water), while $\beta_l$ estimates diverge
systematically, confirming the diagnostic pattern in
Figure~\ref{fig:recovery} that the exclusion restriction holds for
capital but not labor.

To identify $(\beta_k, \beta_l)$ jointly, I apply the Block~C
homothetic CES approach (Section~\ref{sec:homothetic}).
Block~C is the primary identification route for both $\beta_k$ and
$\beta_l$; the exclusion restriction provides an independent check
on $\beta_k$ only, since the restriction fails for labor
(Figure~\ref{fig:recovery}, Panel~b).
The two strategies yield mutually consistent estimates of $\beta_k$
for the \NconsGroup{} industries where the exclusion restriction
is validated for capital
(Table~\ref{tab:4group},
Appendix~\ref{sec:appendix_empirical}).

Table~\ref{tab:beta_kl_summary} summarizes the cross-industry
distributions of $\hat{\beta}_k$ and $\hat{\beta}_l$ across
three approaches: exclusion restriction (broken out by proxy input),
Block~C (homothetic CES), and ACF.

\begin{table}[htbp]
\centering
\caption{Cross-Industry Distribution of $\hat{\beta}_k$ and $\hat{\beta}_l$: Exclusion Restriction, Block~C, and ACF}
\label{tab:beta_kl_summary}
\small
\begin{tabular}{l r r r r r}
\toprule
 & Excl.\ ($m$) & Excl.\ ($e$) & Excl.\ ($w$) & Block~C & ACF \\
\midrule
$\hat{\beta}_k$ Median & 0.0086 & 0.0220 & 0.0106 & 0.0350 & 0.0297 \\
$\hat{\beta}_k$ Mean & 0.0121 & 0.0010 & -0.0077 & 0.0481 & 0.0528 \\
$\hat{\beta}_k$ SD & 0.2018 & 0.2724 & 0.1935 & 0.0544 & 0.0921 \\
\midrule
$\hat{\beta}_l$ Median & 0.2106 & 0.2600 & 0.2010 & 0.3316 & 0.2753 \\
$\hat{\beta}_l$ Mean & -0.1625 & -0.2815 & -0.2317 & 0.3357 & 0.3217 \\
$\hat{\beta}_l$ SD & 3.8991 & 5.1635 & 3.7362 & 0.2239 & 0.2657 \\
\midrule
$N$ & 389 & 389 & 389 & 502 & 499 \\
\bottomrule
\end{tabular}
\vspace{0.5em}
\begin{minipage}{0.95\textwidth}\footnotesize
\textit{Notes:} Industries with $|\hat{\beta}|>2$ excluded. Excl.\ ($m$/$e$/$w$): exclusion restriction OLS using each proxy (Proposition~\ref{prop:excl_ols}). Block~C: homothetic CES approach (Theorem~\ref{thm:homothetic_id}). ACF: \textcite{ackerberg2015identification}.
\end{minipage}
\end{table}

Estimates of $\hat{\beta}_k$ are broadly consistent across all three
approaches (median $\approx 0.01$--$0.04$), corroborating the
identification cross-check in Table~\ref{tab:4group}
(Appendix~\ref{sec:appendix_empirical}).
Labor elasticity estimates diverge more substantially:
Block~C yields a median $\hat{\beta}_l = 0.33$, while ACF produces
a higher median of $0.50$.
The Monte Carlo simulations (Table~\ref{tab:mc_part2_dgp3potential})
show that under DGP~3, ACF $\hat{\beta}_l$ collapses to near zero
(bias $\approx -0.30$), while the proposed method recovers the true
value accurately (bias $\approx +0.01$).
The empirical ACF estimate lies above the proposed estimate,
which is the opposite direction from the MC collapse.
Both patterns reflect the same fragility: ACF labor elasticity
identification breaks down when the Markov assumption is violated,
with the direction of the deviation depending on the specific
dynamics of the data-generating process.
Figure~\ref{fig:beta_kl_3methods} shows the full cross-industry
density distributions for all three methods.
A four-group comparison across identification strategies is
reported in Table~\ref{tab:4group}
(Appendix~\ref{sec:appendix_empirical}).

\begin{figure}[htbp]
\centering
\includegraphics[width=\textwidth]{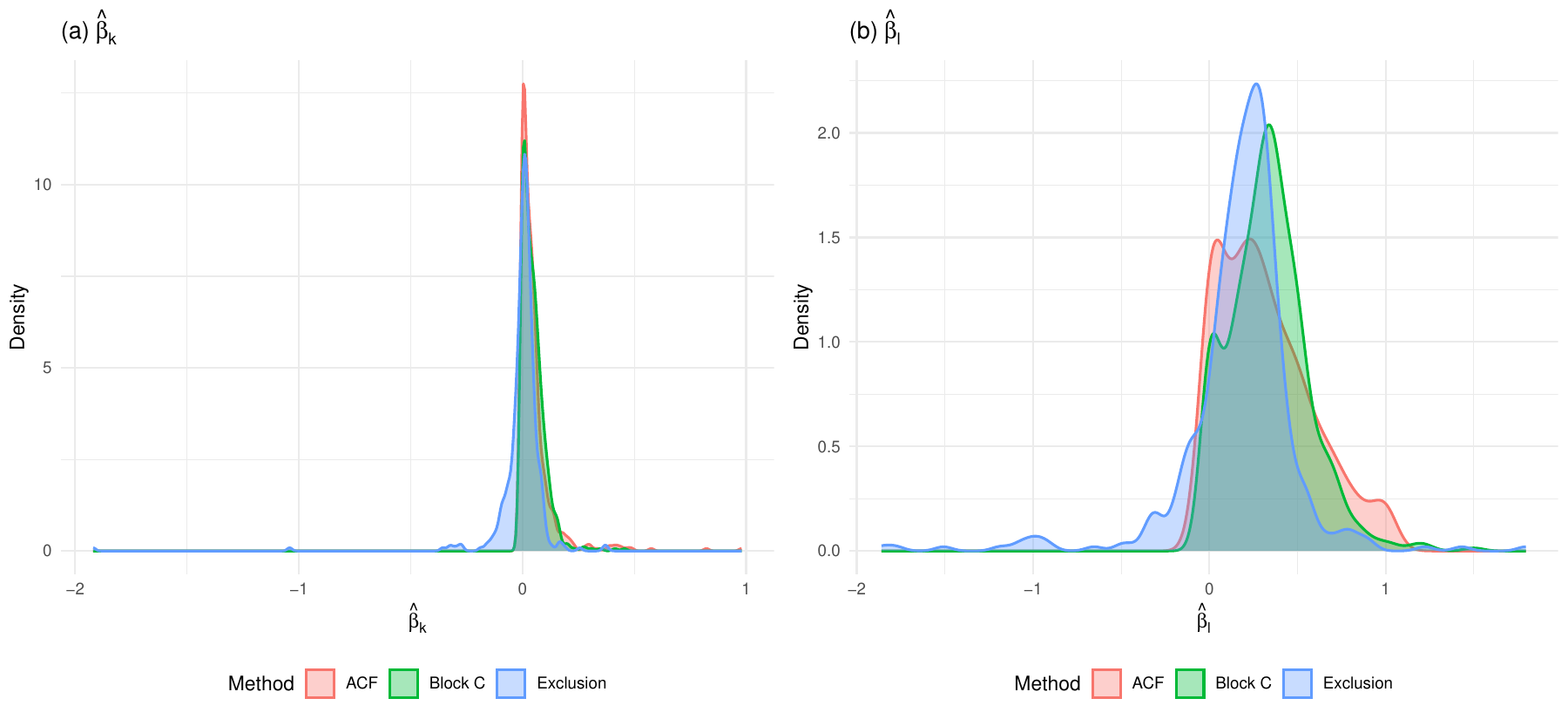}
\caption{Cross-Industry Distribution of $\hat{\beta}_k$ and $\hat{\beta}_l$: Three Methods}
\label{fig:beta_kl_3methods}

\small\raggedright \textit{Notes: Panel~(a) shows the density of
$\hat{\beta}_k$ from Exclusion, Homothetic (Block~C), and
ACF. Panel~(b) shows $\hat{\beta}_l$ for all three
methods; Exclusion estimates use the materials proxy
(Proposition~\ref{prop:excl_ols}).
Industries with $|\hat{\beta}| > 2$ are excluded.
Summary statistics in Table~\ref{tab:beta_kl_summary};
four-group identification cross-check in
Table~\ref{tab:4group} (Appendix~\ref{subsec:crosscheck_4group}).}
\end{figure}

\section{Conclusion}
\label{sec:conclusion}

Can the production function be identified without restricting how
productivity evolves over time? This paper answers in the
affirmative. Replacing the Markov assumption with conditional
independence across three intermediate inputs, the paper shows that
the production function and the distribution of productivity are
nonparametrically identified from a single cross-section. No
assumption on the law of motion for $\omega_{jt}$ is required at
any stage of estimation. The empirical analysis, covering 502
Japanese manufacturing industries, confirms that the choice between
the two identification strategies has quantitative consequences for
every downstream object: input elasticities, markups, allocative
efficiency, and the measured response of productivity to economic
shocks.

The consequences are economically large. The proposed method yields
systematically lower markups than the standard proxy variable
estimator across the entire distribution (median 0.93 vs.\ 1.03;
the share of industries above unity falls from 54 to 37 percent),
shifting the measured degree of market power in the manufacturing
sector. In the earthquake event
study (Section~\ref{subsec:event_study_body}), the
difference-in-differences estimate of the productivity effect on
plants in the three most severely affected prefectures is
$-1.28\%$ under the proposed method and $-1.68\%$ under the
standard method; the 0.40 percentage point gap corresponds to
roughly \$3.6 billion (\textyen{}400 billion) per year in
mismeasured productivity when scaled to aggregate manufacturing
output. The
\textcite{olley1996thedynamics} decomposition and the productivity
determinant regressions reinforce the same pattern: the log-wage
coefficient is roughly 25\% larger under the proposed method
(Table~\ref{tab:main_regressions}), consistent with a higher
signal-to-noise ratio in the recovered productivity measure once
input-specific demand shocks are separated out. The Monte Carlo
simulations, the convergence diagnostic, and the determinant
regressions all point in the same direction, and the underlying
mechanism is general: in any setting where a policy, shock, or
institutional change alters the transition path of productivity,
the Markov transition equation is structurally misspecified and
the resulting production function parameters lack a consistent
interpretation. Trade liberalization, R\&D subsidies, natural
disasters, and mergers all generate such dynamics. The proposed
method accommodates these settings because it imposes no
restriction on how productivity evolves. The Cobb--Douglas
functional form is shared by both the proposed method and the ACF
benchmark, so the gap between estimates reflects the difference in
identifying assumptions, not in functional form.

These findings connect to two broader debates. First, the recent
literature on rising global markups
\parencite{deloecker2020rise} relies on production function
estimates that impose the Markov assumption. The present results
suggest that markup levels, and potentially trends, are sensitive to
this assumption; replication of the global markup finding under
conditional independence identification is a natural next step.
Second, \textcite{chen2024identifying} show that standard proxy
variable estimators are structurally incompatible with a potential
outcomes framework: the Markov transition equation has no structural
interpretation when a treatment alters the productivity process, so
the resulting estimates lack economic meaning under policy
evaluation. The proposed method avoids this problem because it uses
no transition equation; Proposition~\ref{prop:omega_hat_D}
establishes that the recovered productivity measure retains a
causal interpretation under treatment assignment mechanisms
satisfying conditions~(i)--(ii) of that proposition (the treatment
does not alter demand function structure and is orthogonal to
input-specific demand shocks). The same static structure accommodates time-varying
parameters without additional assumptions, since no intertemporal
link is imposed.

A broader implication concerns the nature of identifying assumptions
in production function estimation. The Markov restriction is a
constraint on the time-series behavior of an unobservable; the
conditional independence restriction is a constraint on the
structure of input markets. The latter is grounded in economic
primitives (separate suppliers, distinct procurement channels,
independent regulatory regimes), and the researcher can specify
which observable controls restore the assumption when a particular
threat is identified. This transparency provides a basis for
evaluating the credibility of the estimates that has no analogue
under the Markov framework.

On the theoretical side, Theorem~\ref{thm:obs_equiv} characterizes
the residual indeterminacy that arises once the Markov assumption is
dropped. Two routes close this indeterminacy: an exclusion
restriction (Corollary~\ref{thm:exclusion}) with a testable
necessary condition (Remark~\ref{rem:testable_exclusion}), and a
homothetic regularity condition
(Theorem~\ref{thm:homothetic_id}). The two routes yield mutually
consistent estimates in industries where both apply.

Three limitations and corresponding directions for future work
deserve mention. First, the identification strategy requires at
least three intermediate inputs with separately observable quantity
data, though this requirement is met in several settings beyond the
Japanese Census of Manufactures, including the U.S.\ EIA Form~923
\parencite{fabrizio2007dothey,cicala2015when} and emissions data
in environmental economics.\footnote{%
  Additional datasets satisfying this requirement include
  India's Annual Survey of Industries (ASI), which reports
  firm-level electricity and fuel consumption alongside materials;
  Canada's Annual Survey of Manufacturing and Logging (ASML),
  which covers electricity and water use at the establishment level;
  and the World Bank Enterprise Survey (WBES), which collects
  firm-level electricity expenditure and water source data across
  over 100 countries. These datasets enable direct application
  of the proposed estimator in diverse institutional settings.}
When labor adjustment is rapid, labor itself serves as an additional
productivity signal, reducing the required number of intermediate
inputs from three to two (footnote~\ref{rem:static_labor});
extending the framework to such settings is a natural direction.
Second, no targeted test of the conditional independence assumption
alone exists; the convergence diagnostic of
Remark~\ref{rem:testable_exclusion} provides a necessary condition
for the exclusion restriction. Extending the moment system to
achieve overidentification (for instance via Block~C structural
constraints or cross-equation demand restrictions under
Cobb--Douglas) would enable formal specification testing. Third,
Block~C identification of $(\beta_k, \beta_l)$ requires
non-negligible curvature in $h(v)$; when the capital-labor ratio
varies little, the exclusion restriction route becomes preferable,
and combining the static identification of flexible input
elasticities with semiparametric methods for the capital-labor
component is left for future research.

\clearpage
\printbibliography

\clearpage
\thispagestyle{empty}
\vspace*{\fill}
\begin{center}
{\LARGE\bfseries Appendix}\\[1.5em]
{\large Nonparametric Identification and Estimation of Production Functions\\
Invariant to Productivity Dynamics}\\[1em]
{\normalsize Rentaro Utamaru}
\end{center}
\vspace*{\fill}
\clearpage

\begin{appendices}
\numberwithin{thm}{section}
\numberwithin{assumption}{section}
\numberwithin{remark}{section}
\numberwithin{proposition}{section}
\numberwithin{corollary}{section}
\setcounter{thm}{0}
\setcounter{assumption}{0}
\setcounter{remark}{0}
\setcounter{proposition}{0}
\setcounter{corollary}{0}


\section{Identification Details}
\label{app:identification_details}

This appendix collects the regularity conditions
(Assumptions~\ref{ass:injectivity}--\ref{ass:labeling}),
the density identification theorem (Theorem~\ref{thm:density_id}),
and selected propositions that supplement the main identification
results in Section~\ref{sec:model_identification}.

\subsection{Regularity Conditions}

\begin{assumption}[Injectivity]
\label{ass:injectivity}
The integral operator $L_{e_{jt} \mid \omega_{jt}, x_{jt}}$ with
kernel $f_{e_{jt} \mid \omega_{jt}, x_{jt}}$ and the integral
operator $L_{\omega_{jt} \mid w_{jt}, x_{jt}}$ with kernel
$f_{\omega_{jt} \mid w_{jt}, x_{jt}}$ are both injective.
\end{assumption}

\textit{Role and economic content.} Injectivity requires that distinct
productivity levels generate distinct conditional distributions of
$e_{jt}$ and $w_{jt}$: a firm cannot be more productive without
systematically altering its input demand. This is weaker than
the strict monotonicity plus scalar unobservability required by
\textcite{ackerberg2015identification}: strict monotonicity of
the input demand function in productivity is one sufficient
condition for injectivity, but injectivity holds more broadly in
the presence of idiosyncratic shocks $(\tau, \nu, \eta)$ that
would violate scalar unobservability. The condition could fail in
settings where input allocation is determined by administrative
rules rather than optimization; for example, publicly operated
utilities where electricity consumption follows fixed schedules
regardless of productivity. In competitive manufacturing, the
condition is generically satisfied.

\textit{Weak identification and eigenvalue decay.}
Injectivity is a point-identification condition, not a strength condition.
Even when the operators $L_{e \mid \omega, x}$ and $L_{\omega \mid w, x}$ are
injective, identification can be weak in finite samples if the
eigenvalues of the associated integral operators decay rapidly to
zero. Rapid eigenvalue decay arises when the demand shocks
$(\tau_{jt}, \nu_{jt}, \eta_{jt})$ have large variance relative to
the productivity signal: a high noise-to-signal ratio compresses the
spectrum of $L_{e \mid \omega, x}$, making the inversion ill-conditioned.
Formally, let $\{\lambda_s\}_{s=1}^{\infty}$ denote the eigenvalues of
$L_{e \mid \omega, x}$ in descending order. Rapid decay
$\lambda_s / \lambda_1 \to 0$ as $s \to \infty$ implies that the
conditional density of $e_{jt}$ given $\omega_{jt}$ concentrates its
informational content in a low-dimensional subspace, reducing effective
identification to that subspace.  In the current application, the use
of three independent intermediate inputs mitigates this concern: the
mutual independence of $(\tau, \nu, \eta)$ (Assumption~\ref{ass:cond_indep})
ensures that the information in $(e_{jt}, w_{jt})$ about $\omega_{jt}$
is not co-linear, and the empirical diagnostics in
Section~\ref{subsec:blockc_diagnostics} provide indirect evidence on
identification strength through the significance of $\hat{\rho}_2$ and
$\hat{\rho}_3$. Industries where these parameters are insignificant
should be interpreted with caution as potentially subject to weak
identification of the Block~C moments.

\begin{assumption}[Distinct Eigenvalues]
\label{ass:distinct_eigen}
For any $\omega_{jt} \neq \bar{\omega}_{jt}$, the conditional
densities $f_{w_{jt} \mid \omega_{jt}, x_{jt}}$ and
$f_{w_{jt} \mid \bar{\omega}_{jt}, x_{jt}}$ are not identical
as functions of $w_{jt}$.
\end{assumption}

\textit{Role and economic content.} This condition ensures that the
eigenvalues in the spectral decomposition are distinct, which is
required for unique identification of the eigenfunctions. It
requires that no two distinct productivity levels produce
identical distributions of water demand, conditional on
$(k, l, z)$. The condition fails if water demand is degenerate
or if the conditioning set $x_{jt}$ contains information that
renders $w_{jt}$ uninformative about $\omega_{jt}$. As with
Assumption~\ref{ass:injectivity}, this could be violated in
regulated industries where water allocation is rationed (e.g.,
public irrigation systems), but is generically satisfied in
manufacturing where water consumption responds to the scale of
production.

\begin{assumption}[Productivity Labeling]
\label{ass:labeling}
For each fixed $(k_{jt}, l_{jt})$, the location (labeling) of
$\omega_{jt}$ is fixed by a normalization corresponding to
Assumption~5 in HS08. Specifically, there exists a location
functional $M$ (e.g., the conditional mean) such that
\[
  M\!\bigl[f_{m_{jt} \mid \omega_{jt}, x_{jt}}
  (\cdot \mid \omega_{jt})\bigr] = \omega_{jt}
  \quad \text{for all } \omega_{jt} \in \Omega
\]
holds for each $(k_{jt}, l_{jt})$ separately.
\end{assumption}

\textit{Role and economic content.} Assumption~\ref{ass:labeling}
fixes the scale and location of the latent productivity index within
each capital-labor cell. The HS08 spectral decomposition identifies
the densities up to a relabeling of the latent variable;
Assumption~\ref{ass:labeling} resolves this by requiring that the
conditional mean of material demand (as a function of $\omega$)
equals $\omega$ itself, normalizing productivity to the metric
of material demand. The choice of $M$ as the conditional mean is
conventional; other location functionals that are equivariant
under location shifts yield equivalent identification results,
and the linear GMM in Section~\ref{sec:estimation}
is invariant to this choice.

Critically, the normalization is applied independently at each
$(k_{jt}, l_{jt})$: for each capital-labor cell, the HS08
decomposition solves a separate measurement error model instance
with its own latent variable. The consistency of $\omega$ levels
across different values of $(k, l)$ is not guaranteed by this
assumption alone; see Section~\ref{sec:obs_equiv}.

The analysis further requires that the conditional densities
$f_{m \mid \omega, k, l}$, $f_{e \mid \omega, k, l}$, and
$f_{w \mid \omega, k, l}$ depend continuously on $(k, l)$ in the
$L^1$ norm for each fixed $\omega$. This regularity condition is
satisfied whenever the demand functions $g_m, g_e, g_w$ are
continuous in $(k, l)$ and the demand shocks have smooth densities.

\subsection{Proof of Theorem~\ref{thm:density_id}}
\label{app:density_id_proof}

\begin{proof}[Proof of Theorem~\ref{thm:density_id}]
The proof applies HS08 to the present setting. Consider the joint density
$f_{m_{jt}, e_{jt} \mid x_{jt}, w_{jt}}$, which is directly
identifiable from the data. Using the law of total probability, I
introduce unobserved productivity $\omega_{jt}$ as a latent variable
of integration:
\[
  f_{m_{jt}, e_{jt} \mid x_{jt}, w_{jt}}
  = \int f_{m_{jt}, e_{jt}, \omega_{jt} \mid x_{jt}, w_{jt}}\,
  d\omega_{jt}.
\]
Since $w_{jt}$ is a function of $(\omega_{jt}, x_{jt}, \eta_{jt})$
and $\eta_{jt}$ is independent of $(\tau_{jt}, \nu_{jt})$
conditional on $(\omega_{jt}, x_{jt})$
(Assumption~\ref{ass:cond_indep}), conditioning on $w_{jt}$ does
not alter the conditional distribution of $m_{jt}$ or $e_{jt}$
given $(\omega_{jt}, x_{jt})$. Applying the chain rule and the
conditional independence assumption, I decompose the integrand
into the product of three unknown conditional densities:
\begin{equation}
  f_{m_{jt}, e_{jt} \mid x_{jt}, w_{jt}}
  = \int f_{e_{jt} \mid \omega_{jt}, x_{jt}}
  \cdot f_{m_{jt} \mid \omega_{jt}, x_{jt}}
  \cdot f_{\omega_{jt} \mid x_{jt}, w_{jt}}\, d\omega_{jt}.
  \label{eq:integral_eq}
\end{equation}
This equation has the structure of Equation~(5) in HS08, whose
Theorem~1 establishes uniqueness of the three unknown densities
under the maintained assumptions.

The conditional distributions of input demand and productivity
are therefore nonparametrically identified from static data alone.
\end{proof}

The roles of $m$, $e$, and $w$ in Theorem~\ref{thm:density_id}
are interchangeable: any permutation of the three inputs yields
the same identification result. The asymmetric instrument strategy
in Section~\ref{sec:gmm_approach} breaks this symmetry at the
estimation stage for efficiency, but the underlying identification
is symmetric.

As a consequence of Theorem~\ref{thm:density_id}, the density
$f_{w_{jt} \mid \omega_{jt}, x_{jt}}$ is also identified. This
follows from Bayes' rule:
\begin{equation}
  f_{w_{jt} \mid \omega_{jt}, x_{jt}}(w \mid \omega, x)
  = \frac{
    f_{\omega_{jt} \mid x_{jt}, w_{jt}}(\omega \mid x, w)
    \cdot f_{w_{jt} \mid x_{jt}}(w \mid x)
  }{
    f_{\omega_{jt} \mid x_{jt}}(\omega \mid x)
  },
  \label{eq:bayes_mpp}
\end{equation}
where the numerator's first factor is identified by
Theorem~\ref{thm:density_id}, $f_{w_{jt} \mid x_{jt}}$ is
directly computable from the data, and the denominator is obtained by
marginalizing over $w_{jt}$.

By applying Bayes' rule, the full posterior density of
productivity given all observable inputs is identified:
\begin{equation}
  f_{\omega_{jt} \mid x_{jt}, m_{jt}, e_{jt}, w_{jt}}
  (\omega \mid \cdot)
  = \frac{
    f_{m_{jt} \mid \omega_{jt}, x_{jt}}(m \mid \omega, \cdot)\;
    f_{e_{jt} \mid \omega_{jt}, x_{jt}}(e \mid \omega, \cdot)\;
    f_{\omega_{jt} \mid x_{jt}, w_{jt}}(\omega \mid \cdot)
  }{
    f_{m_{jt}, e_{jt} \mid x_{jt}, w_{jt}}(m, e \mid \cdot)
  },
  \label{eq:bayes_omega}
\end{equation}
where all densities in the numerator are identified by
Theorem~\ref{thm:density_id} and the denominator is directly
computable from the data. This conditional density is used in
Section~\ref{sec:prod_func_id} for production function
identification.

\subsection{Identification up to $\mathbb{E}[\omega \mid k, l]$}

\begin{thm}[{Identification up to $\mathbb{E}[\omega \mid k, l]$}]
\label{thm:id_up_to_E}
Under Assumptions~\ref{ass:additive_error}--\ref{ass:cond_indep}
and Assumptions~\ref{ass:injectivity}--\ref{ass:labeling}, the
production function $f_t$ is identified up to the specification of
$\mathbb{E}[\omega \mid k, l]$. Specifically:
\begin{enumerate}
\item[(a)] If an additional restriction specifying the functional
  form of $\mathbb{E}[\omega \mid k, l]$ is introduced, $f_t$ is
  point-identified.
\item[(b)] Without any restriction on
  $\mathbb{E}[\omega \mid k, l]$, $f_t$ is not point-identified.
\end{enumerate}
\end{thm}

\begin{proof}
(a) By Theorem~\ref{thm:obs_equiv}, observationally equivalent
structures are indexed by $\Delta(k, l)$. Since
$\mathbb{E}[\tilde{\omega} \mid k, l]
= \mathbb{E}[\omega \mid k, l] + \Delta(k, l)$, fixing
$\mathbb{E}[\omega \mid k, l]$ uniquely pins down
$\Delta(k, l) = 0$, and hence $f_t$ is point-identified.

\noindent (b) Without restrictions, $\Delta(k, l)$ can be any
continuous function, and Theorem~\ref{thm:obs_equiv} implies that
point identification cannot be achieved.
\end{proof}

\subsection{Identification via Exclusion Restriction (Proposition)}

\begin{proposition}[Identification via Exclusion Restriction]
\label{prop:excl_ols}
Under Assumptions~\ref{ass:additive_error}--\ref{ass:cond_indep},
Assumptions~\ref{ass:injectivity}--\ref{ass:labeling}, and the
linear demand specification, let
$\tilde{y}_{jt} \equiv y_{jt} - \hat{\beta}_m m_{jt}
- \hat{\beta}_e e_{jt} - \hat{\beta}_w w_{jt}$ denote the
partially identified output using Block~A+B estimates.

\medskip
\noindent\textbf{Case~1} (Joint exclusion).
If $a_{k}^{h} = a_{l}^{h} = 0$ for some input $h$, construct the
productivity proxy
\begin{equation}
  \hat{\omega}_{jt}^h
  = \frac{h_{jt} - \hat{a}_{z}^{h\prime}\,z_{jt}}
         {\hat{a}_{\omega}^{h}}.
  \label{eq:omega_proxy_joint}
\end{equation}
Then OLS regression of
$(\tilde{y}_{jt} - \hat{\omega}_{jt}^h)$ on
$(1, k_{jt}, l_{jt})$ consistently estimates
$(\beta_0, \beta_k, \beta_l)$.

\medskip
\noindent\textbf{Case~2} (Marginal exclusion).
If $a_{k}^{h_1} = 0$ for input $h_1$ and $a_{l}^{h_2} = 0$ for
input $h_2$ ($h_1 \neq h_2$), construct the proxies
\begin{align}
  \hat{\omega}_{jt}^{h_1}
  &= \frac{h_{1,jt} - \hat{a}_{l}^{h_1 *}\,l_{jt}
    - \hat{a}_{z}^{h_1\prime}\,z_{jt}}{\hat{a}_{\omega}^{h_1}},
  \label{eq:proxy_h1} \\
  \hat{\omega}_{jt}^{h_2}
  &= \frac{h_{2,jt} - \hat{a}_{k}^{h_2 *}\,k_{jt}
    - \hat{a}_{z}^{h_2\prime}\,z_{jt}}{\hat{a}_{\omega}^{h_2}},
  \label{eq:proxy_h2}
\end{align}
where $\hat{a}_{l}^{h_1 *}$ and $\hat{a}_{k}^{h_2 *}$ denote
the Block~A+B estimates (which include the
$\Delta(k,l)$ indeterminacy). Then:
\begin{enumerate}
\item[(i)] the coefficient on $k$ in the OLS regression of
  $(\tilde{y}_{jt} - \hat{\omega}_{jt}^{h_1})$ on
  $(1, k_{jt}, l_{jt})$ consistently estimates $\beta_k$;
\item[(ii)] the coefficient on $l$ in the OLS regression of
  $(\tilde{y}_{jt} - \hat{\omega}_{jt}^{h_2})$ on
  $(1, k_{jt}, l_{jt})$ consistently estimates $\beta_l$.
\end{enumerate}

This procedure requires neither the Markov assumption nor the
homothetic regularity condition
(Assumption~\ref{ass:homothetic}).
\end{proposition}

\begin{proof}[Proof sketch (Case~1; full proof in
Appendix~\ref{app:proof_excl_ols})]
Under the linear specification, the observational equivalence
of Theorem~\ref{thm:obs_equiv} implies that Block~A+B
estimates satisfy
$\hat{a}_{k}^{h *} \xrightarrow{p} a_{k}^{h} + a_{\omega}^{h}\,c_k$
for constants $(c_k, c_l)$.
The proxy~\eqref{eq:omega_proxy_joint} uses only the
invariant estimates $\hat{a}_{z}^{h}$ and
$\hat{a}_{\omega}^{h}$, yielding
$\hat{\omega}^h \xrightarrow{p} \omega + \eta^h/a_{\omega}^{h}$.
Subtracting from $\tilde{y}$ leaves a regression with
error $\varepsilon - \eta^h/a_{\omega}^{h}$, which is
mean-independent of $(k,l)$ by iterated expectations and the
exclusion restriction $a_{k}^{h} = a_{l}^{h} = 0$. Case~2
follows by a symmetric argument using separate proxies for
$\beta_k$ and $\beta_l$.
\end{proof}

Replacing the linear subtraction of $\hat{\omega}^h$ in
Proposition~\ref{prop:excl_ols} with a polynomial regression
is not consistent in general; see
Appendix~\ref{app:nonlinear_specs} for details.

\subsection{Validity of Recovered Productivity for Policy Evaluation}

\begin{proposition}[Validity of Recovered Productivity for Policy Evaluation]
\label{prop:omega_hat_D}
Under Assumptions~\ref{ass:additive_error}--\ref{ass:cond_indep}
and Assumptions~\ref{ass:injectivity}--\ref{ass:labeling}, suppose
additionally that:
\begin{enumerate}
\item[(i)] $D$ does not alter the functional form of the intermediate
  input demand functions $g_m, g_e, g_w$; that is, $D$ may
  affect the \textit{level} of $\omega$, but the mapping from
  $\omega$ to $(m, e, w)$ is structurally invariant to $D$.
\item[(ii)] The determination of $D$ may depend on $\omega$ and the
  state variables $(k, l, z)$, but does not depend on the
  input-specific demand shocks $(\tau, \nu, \eta)$.
\end{enumerate}
Then $\mathbb{E}[\hat{\omega}_{jt} \mid D_{jt}]
= \mathbb{E}[\omega_{jt} \mid D_{jt}]$.
\end{proposition}

The proof is given in Appendix~\ref{app:proof_omega_hat_D}.
The idea is that under conditions~(i) and~(ii), the
intermediate inputs serve as sufficient statistics for $\omega$:
once $(x, m, e, w)$ are observed, $D$ provides no additional
information about $\omega$, so the law of iterated expectations
yields the result.\footnote{%
Condition~(i) may fail if the policy fundamentally changes how firms
use intermediate inputs. In such cases,
Theorem~\ref{thm:density_id} should be applied separately to
subpopulations defined by $D = 0$ and $D = 1$. Condition~(ii) may
fail if, for example, a specific input demand shock influences the
firm's decision to participate in the policy.}

\paragraph{Implication for ATT identification in event studies.}
Proposition~\ref{prop:omega_hat_D} provides the measurement guarantee
needed to identify the average treatment effect on the treated (ATT)
using $\hat{\omega}_{jt}$ as the outcome variable.
Suppose that the standard parallel trends assumption holds for the
latent productivity $\omega_{jt}$: for all $t \geq t_0$,
\begin{equation}
  \mathbb{E}[\omega_{jt}(0) \mid D_j = 1]
  - \mathbb{E}[\omega_{jt_0 - 1}(0) \mid D_j = 1]
  =
  \mathbb{E}[\omega_{jt}(0) \mid D_j = 0]
  - \mathbb{E}[\omega_{jt_0 - 1}(0) \mid D_j = 0],
  \label{eq:parallel_trends}
\end{equation}
where $\omega_{jt}(0)$ denotes the potential outcome under no treatment.
Under this assumption, the ATT at period $t$,
\[
  \mathrm{ATT}_t
  \;=\;
  \mathbb{E}[\omega_{jt}(1) - \omega_{jt}(0) \mid D_j = 1],
\]
is identified by the difference-in-differences estimand
\begin{align}
  \hat{\mathrm{ATT}}_t
  &=
  \bigl(\mathbb{E}[\hat{\omega}_{jt} \mid D_j = 1]
        - \mathbb{E}[\hat{\omega}_{jt_0-1} \mid D_j = 1]\bigr)
  \notag\\
  &\quad
  -
  \bigl(\mathbb{E}[\hat{\omega}_{jt} \mid D_j = 0]
        - \mathbb{E}[\hat{\omega}_{jt_0-1} \mid D_j = 0]\bigr).
  \label{eq:att_did}
\end{align}
To see that~\eqref{eq:att_did} converges to $\mathrm{ATT}_t$,
apply Proposition~\ref{prop:omega_hat_D} to replace each
$\mathbb{E}[\hat{\omega}_{jt}\mid D_j]$ with
$\mathbb{E}[\omega_{jt}\mid D_j]$, and then invoke
parallel trends~\eqref{eq:parallel_trends}.
No assumption on the time-series dynamics of $\omega$ is required
beyond~\eqref{eq:parallel_trends} itself.

Standard proxy variable estimators do not share this property.
Their residual satisfies
\[
  \hat{\omega}^{\mathrm{proxy}}_{jt}
  = \omega_{jt}
  + (\beta_k^{\mathrm{true}} - \hat{\beta}_k^{\mathrm{proxy}})\,k_{jt}
  + (\beta_l^{\mathrm{true}} - \hat{\beta}_l^{\mathrm{proxy}})\,l_{jt}
  + o_p(1),
\]
so the DiD estimand~\eqref{eq:att_did} applied to
$\hat{\omega}^{\mathrm{proxy}}$ converges to
\[
  \mathrm{ATT}_t
  +
  (\beta_k^{\mathrm{true}} - \hat{\beta}_k^{\mathrm{proxy}})
  \,\Delta\mathbb{E}[k_{jt} \mid D_j]
  +
  (\beta_l^{\mathrm{true}} - \hat{\beta}_l^{\mathrm{proxy}})
  \,\Delta\mathbb{E}[l_{jt} \mid D_j],
\]
where $\Delta\mathbb{E}[k_{jt}\mid D_j]$ denotes the DiD in
capital between treatment and control groups.
The bias vanishes only if
$\hat{\beta}^{\mathrm{proxy}} = \beta^{\mathrm{true}}$
or if the treatment is orthogonal to $(k, l)$.
In settings where the treatment induces capital or labor
adjustment---as in the earthquake application of
Section~\ref{subsec:event_study_body}---neither condition is
guaranteed.

\subsection{GMM Moment Conditions}

\begin{assumption}[Moment Conditions for GMM]
\label{ass:gmm_uncorrelated}
Let $\Xi_{jt} = (\tau_{jt}, \nu_{jt}, \eta_{jt}, \varepsilon_{jt})$
denote the vector of all structural shocks, and define
$W_{jt} = (1, \omega_{jt}, k_{jt}, l_{jt}, z_{jt})$
as the vector of relevant state variables and functions thereof. The following conditions are maintained:
\begin{enumerate}
\item \textit{Zero Mean Shocks}:
  $\mathbb{E}[\Xi^n_{jt}] = 0$ for all $n$.
\item \textit{State Exogeneity}: All shocks are uncorrelated with
  productivity $\omega_{jt}$ and primary inputs $k_{jt}, l_{jt}$
  (and functions thereof):
  \[
    \mathbb{E}[\Xi^n_{jt} \cdot W^p_{jt}] = 0
    \quad \text{for all } n, p.
  \]
\item \textit{Mutual Exogeneity of Shocks}: Different structural
  shocks are mutually uncorrelated:
  \[
    \mathbb{E}[\Xi^n_{jt} \cdot \Xi^p_{jt}] = 0
    \quad \text{for all } n \neq p.
  \]
\end{enumerate}
\end{assumption}

Assumption~\ref{ass:gmm_uncorrelated} is implied by the zero
conditional mean condition
$\mathbb{E}[\Xi^n_{jt} \mid \omega_{jt}, x_{jt}] = 0$ together
with Assumption~\ref{ass:cond_indep} (conditional
independence).\footnote{Strictly speaking,
Assumption~\ref{ass:gmm_uncorrelated} is weaker than imposing zero
conditional mean on top of Assumption~\ref{ass:cond_indep}. The zero
conditional mean condition implies uncorrelatedness with
\textit{any} measurable function of $(\omega, x)$, whereas
Assumption~\ref{ass:gmm_uncorrelated}(2) requires uncorrelatedness
only with the specific functions in $W_{jt}$.}

\section{Microfoundations of Intermediate Input Factor Demand}
\label{sec:micro_fnd}

In this appendix, I provide microfoundations for the factor demand
functions of the three intermediate inputs introduced in 
Section~\ref{sec:model_identification} 
(Equations~\eqref{eq:demand_m}--\eqref{eq:demand_mpp}). 
Specifically, I demonstrate that the unobserved shock terms
$\tau_{jt}, \nu_{jt}, \eta_{jt}$ included in each demand function 
are structurally derived from the firm's optimization behavior and 
market frictions. I thereby offer a theoretical rationale for the
independence conditions required for identification.

\subsection{Primitives}

The production technology of firm $j$ at time $t$ is described by 
the following general nonparametric production function:
\begin{equation}
  Y_{jt} = F_t(K_{jt}, L_{jt}, M_{jt}, E_{jt}, W_{jt},
  \Omega_{jt}) \exp(\varepsilon_{jt})
  \label{eq:production_tech}
\end{equation}
where $\Omega_{jt}$ is productivity observed by the firm but 
unobserved by the econometrician, and $\varepsilon_{jt}$ is an 
ex-post production shock realized after input decisions are made. 
The firm faces an inverse demand function 
$P_{jt}(Y_{jt}, A_{jt})$ in the product market, where $A_{jt}$ 
denotes a demand shock. Additionally, the firm faces an inverse 
supply function $P_{h,jt}(h_{jt})$ in the market for each 
intermediate input 
$h \in \{M_{jt}, E_{jt}, W_{jt}\}$.

I allow for deviations from perfect competition. I define the
markup $\mu_{jt}$ in the product market and the markdown 
$\psi_{h,jt}$ for input $h$ as follows:
\begin{equation}
  \mu_{jt} \equiv \frac{P_{jt}}{MC_{jt}}, \quad 
  \psi_{h,jt} \equiv \frac{ME_{h,jt}}{P_{h,jt}}
  \label{eq:markup_markdown_def}
\end{equation}
where $MC_{jt}$ denotes marginal cost and $ME_{h,jt}$ denotes 
marginal expenditure. By definition, $\mu_{jt} \geq 1$ and 
$\psi_{h,jt} \geq 1$ hold.

Following \textcite{hsieh2009misallocation}, I define
a wedge $\Upsilon_{h,jt}$ representing exogenous distortions 
specific to each input $h$ (e.g., taxes, adjustment costs, or 
optimization errors).

\subsection{Expected Cost Minimization}

The firm minimizes the cost of achieving a target expected output 
$\bar{Y}_{jt}$:
\begin{align}
  \min_{M_{jt}, E_{jt}, W_{jt}} \quad
  & \sum_{h \in \{M, E, W\}} 
  P_{h,jt}(h_{jt})\, \Upsilon_{h,jt}\, h \\
  \text{s.t.} \quad 
  & F_t(K_{jt}, L_{jt}, M_{jt}, E_{jt}, W_{jt}, \Omega_{jt}) 
  \geq \bar{Y}_{jt}
  \label{eq:constraint_expected}
\end{align}
The Lagrange multiplier $\lambda_{jt}$ is interpreted as the 
marginal cost ($MC_{jt}$). The first-order condition with respect to 
input $h$ is:
\begin{equation}
  ME_{h,jt}\, \Upsilon_{h,jt} = MC_{jt} 
  \frac{\partial F_t}{\partial h_{jt}}
  \label{eq:foc_raw}
\end{equation}
Rewriting using~\eqref{eq:markup_markdown_def}, I obtain
$\psi_{h,jt} P_{h,jt} \Upsilon_{h,jt} 
= \frac{P_{jt}}{\mu_{jt}} (\partial F_t / \partial h_{jt})$. 
Taking logarithms:
\begin{equation}
  \ln\left(\frac{\partial F_t}{\partial h}\right) 
  = \underbrace{\ln(\psi_{h,jt}\, P_{h,jt}\, \Upsilon_{h,jt})}_{%
  \text{Idiosyncratic Input Cost}} 
  - \underbrace{\ln\frac{P_{jt}}{\mu_{jt}}}_{%
  \text{Common Market Factor}}
  \label{eq:structural_decomposition}
\end{equation}
The second term, $\ln(P_{jt}/\mu_{jt})$, is a ``Common Market 
Factor'' that affects the demand for all variable inputs 
symmetrically. This term aggregates the effects of product market 
demand shocks and markups. The markdown $\psi_{h,jt}$
is input-specific: different intermediate inputs may face different 
degrees of buyer power. This heterogeneity in markdowns across 
inputs generates different productivity loading coefficients 
$(\gamma_\omega, \delta_\omega, \zeta_\omega)$ in the reduced-form 
demand functions, even when the underlying production technology is 
common.

\subsection{Correspondence with Factor Demand Functions}

From equation~\eqref{eq:structural_decomposition}, the optimal 
input $M_{jt}$ is a function of state variables, productivity, the 
input price and wedge, and the common factor 
$\ln(P_{jt}/\mu_{jt})$. Assuming strict concavity of the production 
function and applying the Implicit Function Theorem yields:
\begin{equation}
  m_{jt} = g_m\left(k_{jt}, l_{jt}, \omega_{jt}, 
  \ln(\psi_{m,jt}\, P_{m,jt}\, \Upsilon_{m,jt}) 
  - \ln\frac{P_{jt}}{\mu_{jt}}\right)
  \label{eq:structural_demand_general}
\end{equation}

The challenge for identification is that the Common Market Factor 
$\ln(P_{jt}/\mu_{jt})$ may induce correlation among the error terms 
across inputs $(\tau_{jt}, \nu_{jt}, \eta_{jt})$, threatening the 
conditional independence assumption 
(Assumption~\ref{ass:cond_indep}). To address this, I include
additional control variables $z_{jt}$ in the observable state 
variables $x_{jt}$.

The econometric error term $\tau_{jt}$ is defined as the 
``orthogonal residual'' obtained by projecting the combined term 
onto $(\omega_{jt}, k_{jt}, l_{jt}, z_{jt})$:
\begin{equation}
  \tau_{jt} \equiv \left[\ln(\psi_{m,jt}\, P_{m,jt}\, 
  \Upsilon_{m,jt}) - \ln\frac{P_{jt}}{\mu_{jt}}\right] 
  - \mathbb{E}\left[\ln(\psi_{m,jt}\, P_{m,jt}\, \Upsilon_{m,jt}) 
  - \ln\frac{P_{jt}}{\mu_{jt}} \;\middle|\; 
  \omega_{jt}, k_{jt}, l_{jt}, z_{jt}\right]
  \label{eq:tau_definition_projected}
\end{equation}
Since the common factor $\ln(P_{jt}/\mu_{jt})$ is spanned by 
$z_{jt}$, it is removed from the residual $\tau_{jt}$. The residual 
reflects only idiosyncratic input costs: input-specific price 
fluctuations, the outcomes of negotiations with individual 
suppliers, procurement frictions, or optimization errors. It is 
economically reasonable to assume that specific supply shocks in the 
markets for raw materials, industrial water, and electricity are 
mutually independent. This definition provides the microfoundations 
for $\tau_{jt} \perp \nu_{jt} \perp \eta_{jt} \mid 
(\omega_{jt}, x_{jt})$.

Applying this definition to 
equation~\eqref{eq:structural_demand_general} yields the form of 
equation~\eqref{eq:demand_m} in the main text:
\[
  m_{jt} = g_m(x_{jt}, \omega_{jt}, \tau_{jt}).
\]


\section{Production Function Recovery}
\label{app:prod_recovery}

This appendix shows that, for each fixed $(k_0, l_0)$, the 
production function $f_t(k_0, l_0, m, e, w, \omega)$ is 
identified as a function of $(m, e, w, \omega)$. I treat three 
cases of increasing generality.

\begin{thm}[Hicks-Neutral Case]
\label{thm:hicks_recovery}
Under Assumptions~\ref{ass:cond_indep}--\ref{ass:labeling}, in the 
additively separable model 
$y = g_t(k, l, m, e, w) + \omega + \varepsilon$, for each fixed 
$(k_0, l_0)$, $g_t(k_0, l_0, m, e, w)$ is nonparametrically 
identified as a function of $(m, e, w)$.
\end{thm}

\begin{proof}
By Assumption~\ref{ass:additive_error},
\[
  \mathbb{E}[y \mid x, m, e, w] 
  = g_t(k_0, l_0, m, e, w) 
  + \mathbb{E}[\omega \mid x, m, e, w],
\]
where $\mathbb{E}[\varepsilon \mid x, m, e, w] = 0$ follows 
from the law of iterated expectations. Therefore
\[
  g_t(k_0, l_0, m, e, w) 
  = \mathbb{E}[y \mid x, m, e, w] 
  - \mathbb{E}[\omega \mid x, m, e, w].
\]
The first term is identified from the data, and the second is 
computable from $f_{\omega \mid x, m, e, w}$ 
(equation~\eqref{eq:bayes_omega}).
\end{proof}

\begin{thm}[Non-Hicks-Neutral Case: No Ex-Post Shock]
\label{thm:non_hicks_no_eps}
Under Assumptions~\ref{ass:cond_indep}--\ref{ass:distinct_eigen} 
and~\ref{ass:labeling}, in the model 
$y = f_t(k, l, m, e, w, \omega)$ with no ex-post shock, suppose 
$f_t(k_0, l_0, m, e, w, \omega)$ is strictly monotone in 
$\omega$ for each fixed $(m, e, w)$. Then, for each fixed 
$(k_0, l_0)$, $f_t$ is nonparametrically identified as a function 
of $(m, e, w, \omega)$.
\end{thm}

\begin{proof}
Fix $(k_0, l_0)$ and $(m_0, e_0, w_0)$. With $\varepsilon = 0$, 
$y = f_t(k_0, l_0, m_0, e_0, w_0, \omega)$, and strict 
monotonicity in $\omega$ yields
\[
  F_{y \mid x, m_0, e_0, w_0}(y \mid \cdot) 
  = F_{\omega \mid x, m_0, e_0, w_0}
  \bigl(f_t^{-1}(\cdot, y) \,\big|\, \cdot\bigr).
\]
Both sides are identified (the left from data, the right from 
Theorem~\ref{thm:density_id} and 
equation~\eqref{eq:bayes_omega}). By quantile matching:
\[
  f_t(k_0, l_0, m_0, e_0, w_0, \omega) 
  = F_{y \mid \cdot}^{-1}\!\bigl(
  F_{\omega \mid \cdot}(\omega) \,\big|\, \cdot\bigr).
\]
\end{proof}

\begin{thm}[Non-Hicks-Neutral Case: Known $f_\varepsilon$]
\label{thm:non_hicks_known_eps}
Under Assumptions~\ref{ass:cond_indep}--\ref{ass:labeling}, suppose 
additionally that $\varepsilon$ is independent of 
$(\omega, x, m, e, w)$, $f_\varepsilon$ is known, its 
characteristic function has no zeros on $\mathbb{R}$, and 
$f_t(k_0, l_0, m, e, w, \omega)$ is strictly monotone in 
$\omega$ for each fixed $(m, e, w)$. Then, for each fixed 
$(k_0, l_0)$, $f_t$ is nonparametrically identified.
\end{thm}

\begin{proof}
Fix $(k_0, l_0)$ and $(m_0, e_0, w_0)$. Setting 
$s = f_t(k_0, l_0, m_0, e_0, w_0, \omega)$, the observed 
conditional density becomes a convolution:
$f_{y \mid \cdot}(y) = (f_\varepsilon * \tilde{K})(y)$. Taking 
characteristic functions and using 
$\varphi_\varepsilon(t) \neq 0$, one can recover
$\tilde{K}$ by deconvolution. With 
$K \equiv f_{\omega \mid x, m_0, e_0, w_0}$ identified, 
quantile matching recovers $f_t$.
\end{proof}

\begin{remark}
Theorem~\ref{thm:non_hicks_known_eps} assumes $f_\varepsilon$ is 
fully known. When $f_\varepsilon$ belongs to a parametric family 
(e.g., $N(0, \sigma^2)$) with unknown parameters, these can 
typically be recovered from the decay rate of the characteristic 
function.
\end{remark}


\section{Proof of Theorem~\ref{thm:obs_equiv}}
\label{app:proof_obs_equiv}

\begin{proof}
\textit{Sufficiency.} Take any continuous function $\Delta(k, l)$ 
and define $\tilde{\omega}$ and $\tilde{f}_t$ 
by~\eqref{eq:obs_equiv}. Then 
$\tilde{f}_t(k, l, m, e, w, \tilde{\omega}) + \varepsilon 
= f_t(k, l, m, e, w, \omega) + \varepsilon = y$, so the 
distribution of output is unchanged. For intermediate input demands, 
$f_{m \mid \tilde{\omega}, k, l}(m \mid \tilde{\omega}) 
= f_{m \mid \omega, k, l}
(m \mid \tilde{\omega} - \Delta(k, l))$, so the observable joint 
distribution $f_{y, m, e, w \mid k, l}$ is invariant for each 
fixed $(k, l)$.

\medskip

\textit{Necessity.} Suppose $(\tilde{f}_t, \tilde{\omega})$ is 
observationally equivalent to $(f_t, \omega)$. Fix $(k_0, l_0)$ and 
apply the identification procedure of 
Theorem~\ref{thm:density_id} (HS08, Theorem~1).

The proof of HS08's Theorem~1 proceeds in four stages: 
(1)~uniqueness of the spectral decomposition 
(\textcite{dunford1971linear}, Theorem~XV.4.5); 
(2)~fixing the scale of eigenfunctions by the density integration 
condition; (3)~resolving degenerate eigenvalues via 
Assumption~\ref{ass:distinct_eigen}; and (4)~fixing the indexing 
via the location normalization (Assumption~\ref{ass:labeling}).

Stages~(1)--(3) determine the family of conditional densities 
$\{f_{m \mid \omega, k_0, l_0}(\cdot \mid \omega) : 
\omega \in \Omega\}$ uniquely \textit{as an unordered set}. Since 
$(\tilde{f}_t, \tilde{\omega})$ is observationally equivalent, the 
spectral decomposition based on $\tilde{\omega}$ must produce the 
same unordered set. Therefore, there exists a bijection 
$R_{k_0, l_0} \colon \Omega \to \Omega$ such that
\begin{equation}
  f_{m \mid \tilde{\omega}, k_0, l_0}
  (\cdot \mid \tilde{\omega}) 
  = f_{m \mid \omega, k_0, l_0}
  (\cdot \mid R_{k_0, l_0}^{-1}(\tilde{\omega}))
  \quad \text{for all } \tilde{\omega}.
  \label{eq:relabeling_app}
\end{equation}

Applying Stage~(4): under the $\omega$-normalization 
$M[f_{m \mid \omega}(\cdot \mid \omega)] = \omega$, one obtains
\[
  M\!\bigl[f_{m \mid \tilde{\omega}}
  (\cdot \mid \tilde{\omega})\bigr] 
  = R^{-1}(\tilde{\omega}).
\]
If the $\tilde{\omega}$-normalization is imposed, then 
$R^{-1}(\tilde{\omega}) = \tilde{\omega}$ and $R = \mathrm{id}$.

However, the normalization is defined independently for each 
$(k_0, l_0)$. In the true structure, the normalization functional 
$M[f_{m \mid \omega^{\mathrm{true}}, k, l}
(\cdot \mid \omega^{\mathrm{true}})]$ generally depends on 
$(k, l)$, so the normalizations at different $(k, l)$ fix $\omega$ 
at reference points differing by
\[
  c(k_0, l_0) \equiv 
  M\!\bigl[f_{m \mid \omega^{\mathrm{true}}, k_0, l_0}
  (\cdot \mid \omega^{\mathrm{true}})\bigr] 
  - \omega^{\mathrm{true}}.
\]
It follows that 
$\tilde{\omega} = \omega + \Delta(k, l)$ where 
$\Delta(k, l) = c_{\tilde{\omega}}(k, l) - c_\omega(k, l)$ is
continuous (by the assumed continuous dependence of
$f_{m \mid \omega, k, l}$ on $(k, l)$), and 
$\tilde{f}_t(\cdot, \tilde{\omega}) 
= f_t(\cdot, \tilde{\omega} - \Delta(k, l))$.
\end{proof}


\section{Data Generating Process for Monte Carlo Simulation}
\label{app:dgp}

This appendix describes the detailed parameter settings for the 
Data Generating Process (DGP) of the Monte Carlo simulation 
outlined in Section~\ref{sec:simulation}.

\subsection{Common Parameter Settings}

The following parameters are common to all DGPs.

\paragraph{Production Function 
(Equation~(\ref{eq:sim_prod_func})):}
\begin{itemize}
\item Intercept $\beta_0 = 0.1$
\item Capital $\beta_k = 0.2$
\item Labor $\beta_l = 0.3$
\item Intermediate Input $m$: $\beta_m = 0.3$
\item Intermediate Input $e$: $\beta_e = 0.15$
\item Intermediate Input $w$: $\beta_w = 0.1$
\item Measurement Error 
  $\varepsilon_{jt} \sim N(0, \sigma_\varepsilon^2)$, 
  $\sigma_\varepsilon = 0.05$
\end{itemize}

\paragraph{Intermediate Input Demand Functions.}
The demand function coefficients are derived from the first-order 
conditions of cost minimization under input-specific markdowns 
(Appendix~\ref{sec:micro_fnd}). The theoretical benchmark under 
perfect competition yields 
$\gamma_\omega^{\mathrm{bench}} = 1/(1 - S_m) \approx 2.22$, 
where $S_m = \beta_m + \beta_e + \beta_w = 0.55$. 
Input-specific markdowns $\psi_h$ generate heterogeneity in the 
productivity loading coefficients across inputs.
\begin{itemize}
\item $m_{jt} = \gamma_k k_{jt} + \gamma_l l_{jt} 
  + \gamma_\omega \omega_{jt} + \tau_{jt}$
  \begin{itemize}
  \item $(\gamma_k, \gamma_l, \gamma_\omega) = (0.45, 0.65, 2.2)$
  \item $\tau_{jt} = \rho_\tau \tau_{j,t-1} + e_{\tau,jt}$, 
    $\rho_\tau = 0.5$, 
    $\mathrm{Var}(\tau_{jt}) = \sigma_\tau^2 = 0.15^2$
  \end{itemize}
\item $e_{jt} = \delta_k k_{jt} + \delta_l l_{jt} 
  + \delta_\omega \omega_{jt} + \nu_{jt}$
  \begin{itemize}
  \item $(\delta_k, \delta_l, \delta_\omega) = (0.40, 0.60, 2.0)$
  \item $\nu_{jt} = \rho_\nu \nu_{j,t-1} + e_{\nu,jt}$, 
    $\rho_\nu = 0.5$, 
    $\mathrm{Var}(\nu_{jt}) = \sigma_\nu^2 = 0.15^2$
  \end{itemize}
\item $w_{jt} = \zeta_k k_{jt} + \zeta_l l_{jt} 
  + \zeta_\omega \omega_{jt} + \eta_{jt}$
  \begin{itemize}
  \item $(\zeta_k, \zeta_l, \zeta_\omega) = (0.50, 0.70, 1.8)$
  \item $\eta_{jt} = \rho_\eta \eta_{j,t-1} + e_{\eta,jt}$, 
    $\rho_\eta = 0.5$, 
    $\mathrm{Var}(\eta_{jt}) = \sigma_\eta^2 = 0.15^2$
  \end{itemize}
\end{itemize}
The demand shock innovations $e_{\tau,jt}, e_{\nu,jt}, e_{\eta,jt}$ 
follow mutually independent normal distributions. The AR(1) 
structure preserves the conditional independence assumption 
(Assumption~\ref{ass:cond_indep}) at each time point, since each 
shock is generated from its own independent chain.

\paragraph{Dynamic Decisions and Firms' Beliefs:}
\begin{itemize}
\item \textbf{Firms' Beliefs} (Assumed AR(1) in all DGPs):
  $\omega_{jt} = \rho_{\mathrm{belief}}\, \omega_{j,t-1} 
  + \xi_{\mathrm{belief},jt}$, 
  $\rho_{\mathrm{belief}} = 0.8$, stationary variance 
  $\sigma_{\omega,\mathrm{belief}}^2 
  = 0.2^2 / (1 - 0.8^2) \approx 0.111$.
\item \textbf{Capital Accumulation:} 
  $k_{j,t+1} = \log((1 - \delta_{\mathrm{capital}}) 
  \exp(k_{jt}) + i_{jt})$, $\delta_{\mathrm{capital}} = 0.2$.
\item \textbf{Investment Function:} $i_{jt}$ is determined to 
  maximize expected returns under the AR(1) belief, following the 
  mechanism in the Monte Carlo simulation of 
  \textcite{ackerberg2015identification}.
\item \textbf{Labor Decision:} $l_{j,t+1}$ is determined based on 
  the productivity forecast under the AR(1) belief, the 
  predetermined $k_{j,t+1}$, and the exogenous wage 
  $\ln w_{j,t+1}$.
\item \textbf{Wage Process (AR(1)):} 
  $\ln w_{jt} = \rho_{\ln w}\, \ln w_{j,t-1} + \xi_{\ln w,jt}$, 
  $\rho_{\ln w} = 0.3$, $\sigma_{\ln w} = 0.1$.
\item \textbf{Others:} Discount factor $0.95$, investment cost 
  heterogeneity $\sigma_b = 0.6$.
\end{itemize}

\subsection{DGP-Specific Productivity Process Settings}

\paragraph{DGP1: AR(1) Process (Baseline)}
\begin{itemize}
\item $\omega_{jt} = \rho_{\mathrm{dgp1}}\, \omega_{j,t-1} 
  + \xi_{\mathrm{dgp1},jt}$
\item $\rho_{\mathrm{dgp1}} = 0.8$
\item Innovation standard deviation $\sigma_{\xi,\mathrm{dgp1}} 
  = 0.2$ (consistent with firms' beliefs)
\end{itemize}

\paragraph{DGP2: AR(2) Process}
\begin{itemize}
\item $\omega_{jt} = \rho_{1,\mathrm{dgp2}}\, \omega_{j,t-1} 
  + \rho_{2,\mathrm{dgp2}}\, \omega_{j,t-2} 
  + \xi_{\mathrm{dgp2},jt}$
\item $\rho_{1,\mathrm{dgp2}} = 0.6$, 
  $\rho_{2,\mathrm{dgp2}} = 0.3$
\item Innovation standard deviation 
  $\sigma_{\xi,\mathrm{dgp2}} = 0.15$
\end{itemize}

\paragraph{DGP3: Potential Outcome Model}
\begin{itemize}
\item \textbf{Potential Process (Untreated $\omega^0$):}
  $\omega^0_{jt} = \rho_0\, \omega^0_{j,t-1} + \xi_{0,jt}$, 
  $\rho_0 = 0.8$, $\sigma_{\xi 0} = 0.2$
\item \textbf{Potential Process (Treated $\omega^1$):}
  \begin{itemize}
  \item From treated state ($D_{j,t-1} = 1$):
    $\omega^1_{jt} = \rho_1\, \omega^1_{j,t-1} + \Delta
    + \xi_{1,jt}$
  \item From untreated state ($D_{j,t-1} = 0$):
    $\omega^1_{jt} = \rho_1\, \omega^0_{j,t-1} + \Delta
    + \xi_{1,jt}$
  \item $\rho_1 = 0.5$, $\Delta = 0.15$, $\sigma_{\xi 1} = 0.25$
  \end{itemize}
\item \textbf{Reversible Selection:}
  $D_{jt} = \mathbb{I}(\omega^0_{jt} > 0)$.
  Treatment is not an absorbing state: firms enter and exit
  treatment as $\omega^0_{jt}$ crosses the threshold.
  Both potential processes $\omega^0, \omega^1$ evolve
  independently regardless of the current treatment state
  (Diagonal Markov; \textcite{chen2024identifying},
  Assumption~2.1).
\item \textbf{Realized Productivity:} 
  $\omega_{jt} = (1 - D_{jt})\, \omega^0_{jt} 
  + D_{jt}\, \omega^1_{jt}$
\end{itemize}

\subsection{Simulation Execution Settings}
\begin{itemize}
\item Burn-in period $T_{\mathrm{burnin}} = 30$
\item \textbf{Part~1} (Block~A+B): $R = 100$,
  $N \in \{50, 200, 500\}$, $T_{\mathrm{obs}} \in \{10, 20, 50\}$
\item \textbf{Part~2} (Block~A+B+C): $R = 100$,
  $(N, T) = (200, 50)$
\end{itemize}

\section{Additional Results for Monte Carlo Simulation}
\label{sec:mc_additional_results}

This appendix presents detailed simulation results.
Tables report bias, standard deviation, and RMSE for each
parameter, estimation method, and DGP at $(N, T) = (500, 50)$.

\subsection{Part~1: Flexible Input Parameters (Block~A+B)}

\begin{table}
\centering
\caption{\label{tab:mc_part1_dgp1ar1}Part 1: DGP 1: AR(1) (N=500, T=50)}
\centering
\resizebox{\linewidth}{!}{
\fontsize{9}{11}\selectfont
\begin{tabular}[t]{lcrrrrrrrrrrrr}
\toprule
\multicolumn{1}{c}{\textbf{ }} & \multicolumn{1}{c}{\textbf{ }} & \multicolumn{3}{c}{\textbf{ACF}} & \multicolumn{3}{c}{\textbf{ACF-Mod}} & \multicolumn{3}{c}{\textbf{GNR}} & \multicolumn{3}{c}{\textbf{Proposed}} \\
\cmidrule(l{3pt}r{3pt}){3-5} \cmidrule(l{3pt}r{3pt}){6-8} \cmidrule(l{3pt}r{3pt}){9-11} \cmidrule(l{3pt}r{3pt}){12-14}
Parameter & True & Bias & SD & RMSE & Bias & SD & RMSE & Bias & SD & RMSE & Bias & SD & RMSE\\
\midrule
$\beta_m$ & 0.30 & 0.0004 & 0.0009 & 0.0010 & 0.0004 & 0.0010 & 0.0011 & 0.5886 & 0.0011 & 0.5886 & 0.0007 & 0.0035 & 0.0035\\
$\beta_e$ & 0.15 & 0.0006 & 0.0015 & 0.0016 & 0.0005 & 0.0012 & 0.0013 & 0.5480 & 0.0019 & 0.5480 & -0.0008 & 0.0034 & 0.0034\\
$\beta_w$ & 0.10 & 0.0001 & 0.0007 & 0.0007 & 0.0001 & 0.0008 & 0.0008 & 1.0366 & 0.0027 & 1.0366 & 0.0001 & 0.0037 & 0.0037\\
\bottomrule
\end{tabular}}
\end{table}

\begin{table}
\centering
\caption{\label{tab:mc_part1_dgp2ar2}Part 1: DGP 2: AR(2) (N=500, T=50)}
\centering
\resizebox{\linewidth}{!}{
\fontsize{9}{11}\selectfont
\begin{tabular}[t]{lcrrrrrrrrrrrr}
\toprule
\multicolumn{1}{c}{\textbf{ }} & \multicolumn{1}{c}{\textbf{ }} & \multicolumn{3}{c}{\textbf{ACF}} & \multicolumn{3}{c}{\textbf{ACF-Mod}} & \multicolumn{3}{c}{\textbf{GNR}} & \multicolumn{3}{c}{\textbf{Proposed}} \\
\cmidrule(l{3pt}r{3pt}){3-5} \cmidrule(l{3pt}r{3pt}){6-8} \cmidrule(l{3pt}r{3pt}){9-11} \cmidrule(l{3pt}r{3pt}){12-14}
Parameter & True & Bias & SD & RMSE & Bias & SD & RMSE & Bias & SD & RMSE & Bias & SD & RMSE\\
\midrule
$\beta_m$ & 0.30 & 0.0157 & 0.0190 & 0.0245 & 0.0181 & 0.0200 & 0.0268 & 0.5877 & 0.0013 & 0.5877 & 0.0004 & 0.0037 & 0.0036\\
$\beta_e$ & 0.15 & 0.0178 & 0.0219 & 0.0280 & 0.0206 & 0.0229 & 0.0306 & 0.5476 & 0.0029 & 0.5476 & -0.0005 & 0.0032 & 0.0033\\
$\beta_w$ & 0.10 & 0.0021 & 0.0032 & 0.0038 & 0.0014 & 0.0030 & 0.0033 & 1.0347 & 0.0033 & 1.0347 & 0.0004 & 0.0041 & 0.0041\\
\bottomrule
\end{tabular}}
\end{table}

\begin{table}
\centering
\caption{\label{tab:mc_part1_dgp3potential}Part 1: DGP 3: Potential (N=500, T=50)}
\centering
\resizebox{\linewidth}{!}{
\fontsize{9}{11}\selectfont
\begin{tabular}[t]{lcrrrrrrrrrrrr}
\toprule
\multicolumn{1}{c}{\textbf{ }} & \multicolumn{1}{c}{\textbf{ }} & \multicolumn{3}{c}{\textbf{ACF}} & \multicolumn{3}{c}{\textbf{ACF-Mod}} & \multicolumn{3}{c}{\textbf{GNR}} & \multicolumn{3}{c}{\textbf{Proposed}} \\
\cmidrule(l{3pt}r{3pt}){3-5} \cmidrule(l{3pt}r{3pt}){6-8} \cmidrule(l{3pt}r{3pt}){9-11} \cmidrule(l{3pt}r{3pt}){12-14}
Parameter & True & Bias & SD & RMSE & Bias & SD & RMSE & Bias & SD & RMSE & Bias & SD & RMSE\\
\midrule
$\beta_m$ & 0.30 & 0.1947 & 0.0144 & 0.1952 & 0.2325 & 0.0285 & 0.2342 & 0.5888 & 0.0014 & 0.5888 & 0.0012 & 0.0039 & 0.0041\\
$\beta_e$ & 0.15 & 0.1868 & 0.0090 & 0.1871 & 0.1707 & 0.0204 & 0.1719 & 0.5408 & 0.0029 & 0.5408 & -0.0014 & 0.0037 & 0.0039\\
$\beta_w$ & 0.10 & 0.0886 & 0.0117 & 0.0893 & 0.0602 & 0.0122 & 0.0614 & 1.0299 & 0.0027 & 1.0300 & 0.0001 & 0.0031 & 0.0031\\
\bottomrule
\end{tabular}}
\end{table}

\subsection{Part~2: All Parameters (Block~A+B+C)}

\begin{table}
\centering
\caption{\label{tab:mc_part2_dgp1ar1}Part 2: DGP 1: AR(1) (N=200, T=50)}
\centering
\resizebox{\linewidth}{!}{
\fontsize{9}{11}\selectfont
\begin{tabular}[t]{lcrrrrrr}
\toprule
\multicolumn{1}{c}{\textbf{ }} & \multicolumn{1}{c}{\textbf{ }} & \multicolumn{3}{c}{\textbf{ACF}} & \multicolumn{3}{c}{\textbf{Proposed}} \\
\cmidrule(l{3pt}r{3pt}){3-5} \cmidrule(l{3pt}r{3pt}){6-8}
Parameter & True & Bias & SD & RMSE & Bias & SD & RMSE\\
\midrule
$\beta_k$ & 0.20 & 0.0005 & 0.0064 & 0.0063 & 0.0016 & 0.0119 & 0.0119\\
$\beta_l$ & 0.30 & -0.0017 & 0.0047 & 0.0050 & 0.0019 & 0.0226 & 0.0225\\
$\beta_m$ & 0.30 & 0.0005 & 0.0021 & 0.0021 & 0.0002 & 0.0126 & 0.0125\\
$\beta_e$ & 0.15 & 0.0009 & 0.0037 & 0.0038 & 0.0023 & 0.0089 & 0.0091\\
$\beta_w$ & 0.10 & -0.0000 & 0.0009 & 0.0009 & -0.0012 & 0.0115 & 0.0115\\
\bottomrule
\end{tabular}}
\end{table}

\begin{table}
\centering
\caption{\label{tab:mc_part2_dgp2ar2}Part 2: DGP 2: AR(2) (N=200, T=50)}
\centering
\resizebox{\linewidth}{!}{
\fontsize{9}{11}\selectfont
\begin{tabular}[t]{lcrrrrrr}
\toprule
\multicolumn{1}{c}{\textbf{ }} & \multicolumn{1}{c}{\textbf{ }} & \multicolumn{3}{c}{\textbf{ACF}} & \multicolumn{3}{c}{\textbf{Proposed}} \\
\cmidrule(l{3pt}r{3pt}){3-5} \cmidrule(l{3pt}r{3pt}){6-8}
Parameter & True & Bias & SD & RMSE & Bias & SD & RMSE\\
\midrule
$\beta_k$ & 0.20 & -0.0211 & 0.0269 & 0.0340 & 0.0108 & 0.0133 & 0.0170\\
$\beta_l$ & 0.30 & -0.0394 & 0.0540 & 0.0665 & -0.0021 & 0.0192 & 0.0191\\
$\beta_m$ & 0.30 & 0.0187 & 0.0239 & 0.0301 & -0.0019 & 0.0128 & 0.0128\\
$\beta_e$ & 0.15 & 0.0228 & 0.0303 & 0.0376 & 0.0047 & 0.0102 & 0.0111\\
$\beta_w$ & 0.10 & 0.0022 & 0.0045 & 0.0050 & -0.0002 & 0.0099 & 0.0098\\
\bottomrule
\end{tabular}}
\end{table}

\begin{table}
\centering
\caption{\label{tab:mc_part2_dgp3potential}Part 2: DGP 3: Potential (N=200, T=50)}
\centering
\resizebox{\linewidth}{!}{
\fontsize{9}{11}\selectfont
\begin{tabular}[t]{lcrrrrrr}
\toprule
\multicolumn{1}{c}{\textbf{ }} & \multicolumn{1}{c}{\textbf{ }} & \multicolumn{3}{c}{\textbf{ACF}} & \multicolumn{3}{c}{\textbf{Proposed}} \\
\cmidrule(l{3pt}r{3pt}){3-5} \cmidrule(l{3pt}r{3pt}){6-8}
Parameter & True & Bias & SD & RMSE & Bias & SD & RMSE\\
\midrule
$\beta_k$ & 0.20 & -0.1972 & 0.0190 & 0.1981 & 0.0123 & 0.0104 & 0.0161\\
$\beta_l$ & 0.30 & -0.2954 & 0.0287 & 0.2967 & 0.0116 & 0.0139 & 0.0180\\
$\beta_m$ & 0.30 & 0.1902 & 0.0254 & 0.1918 & -0.0057 & 0.0110 & 0.0122\\
$\beta_e$ & 0.15 & 0.1860 & 0.0218 & 0.1872 & -0.0009 & 0.0087 & 0.0087\\
$\beta_w$ & 0.10 & 0.0878 & 0.0144 & 0.0890 & -0.0033 & 0.0083 & 0.0088\\
\bottomrule
\end{tabular}}
\end{table}

\subsection{Additional Monte Carlo Figures}

\begin{figure}[htbp]
\centering
\includegraphics[width=\textwidth]{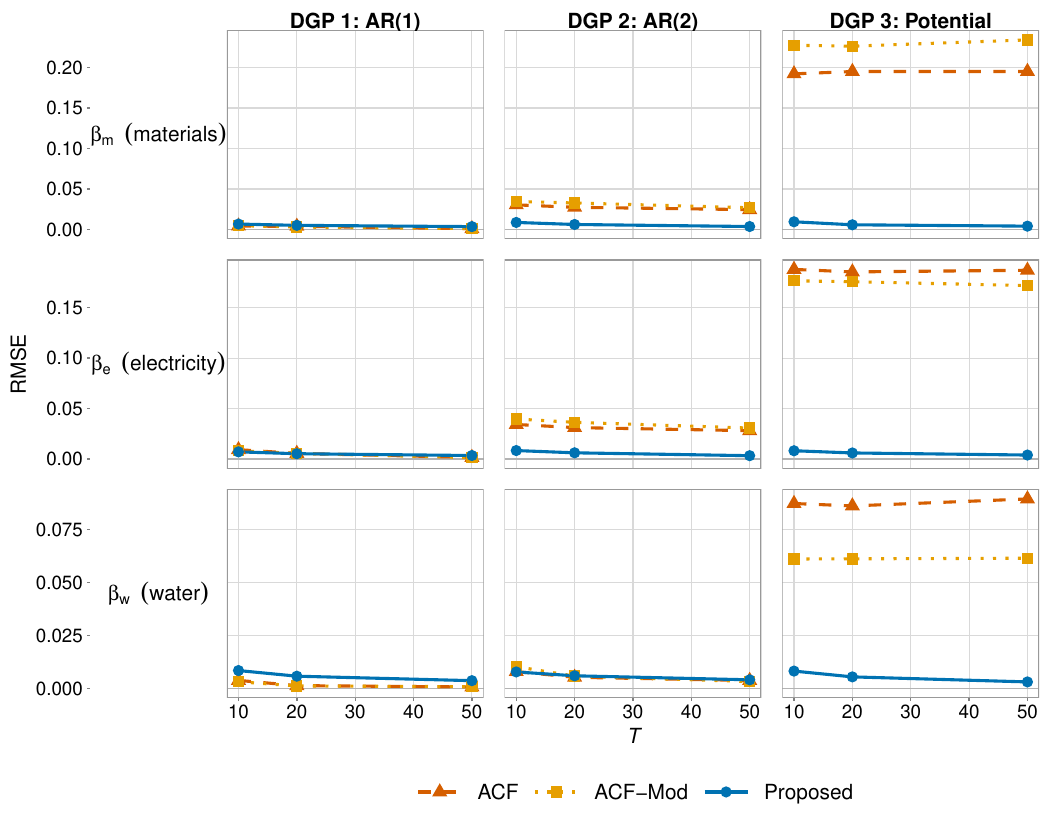}
\caption{Part~1: Mean RMSE Convergence ($N = 500$)}
\label{fig:rmse_convergence_app}

\small\raggedright \textit{Notes: Mean RMSE of $(\hat{\beta}_m, \hat{\beta}_e, \hat{\beta}_w)$ as a function of $T$ ($N = 500$, $R = 100$). Under DGP~2 and DGP~3, the ACF and ACF-Mod RMSEs do not vanish with $T$, reflecting asymptotic bias. The proposed method's RMSE declines monotonically.}
\end{figure}

\begin{figure}[htbp]
\centering
\includegraphics[width=\textwidth]{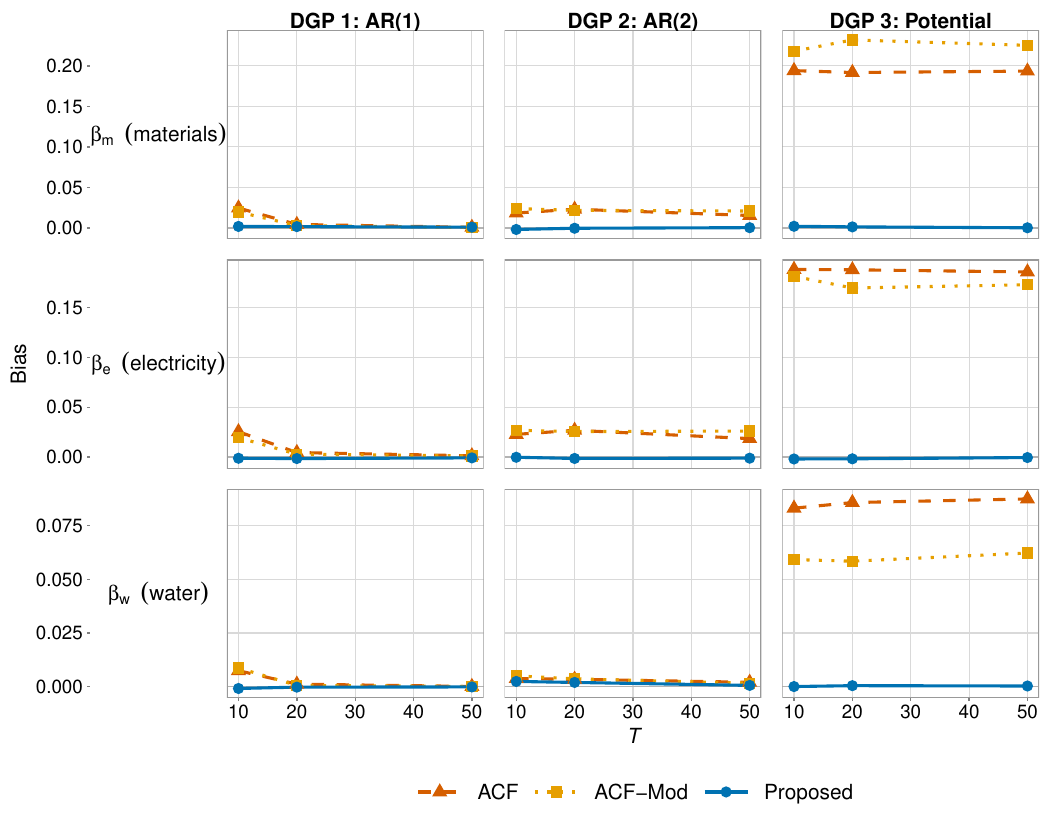}
\caption{Part~1: Mean Bias Convergence ($N = 200$)}
\label{fig:bias_convergence_N200}

\small\raggedright \textit{Notes: Same as Figure~\ref{fig:bias_convergence} but for $N = 200$. The qualitative patterns are preserved; the larger variance reflects the smaller cross-section.}
\end{figure}

\begin{figure}[htbp]
\centering
\includegraphics[width=\textwidth]{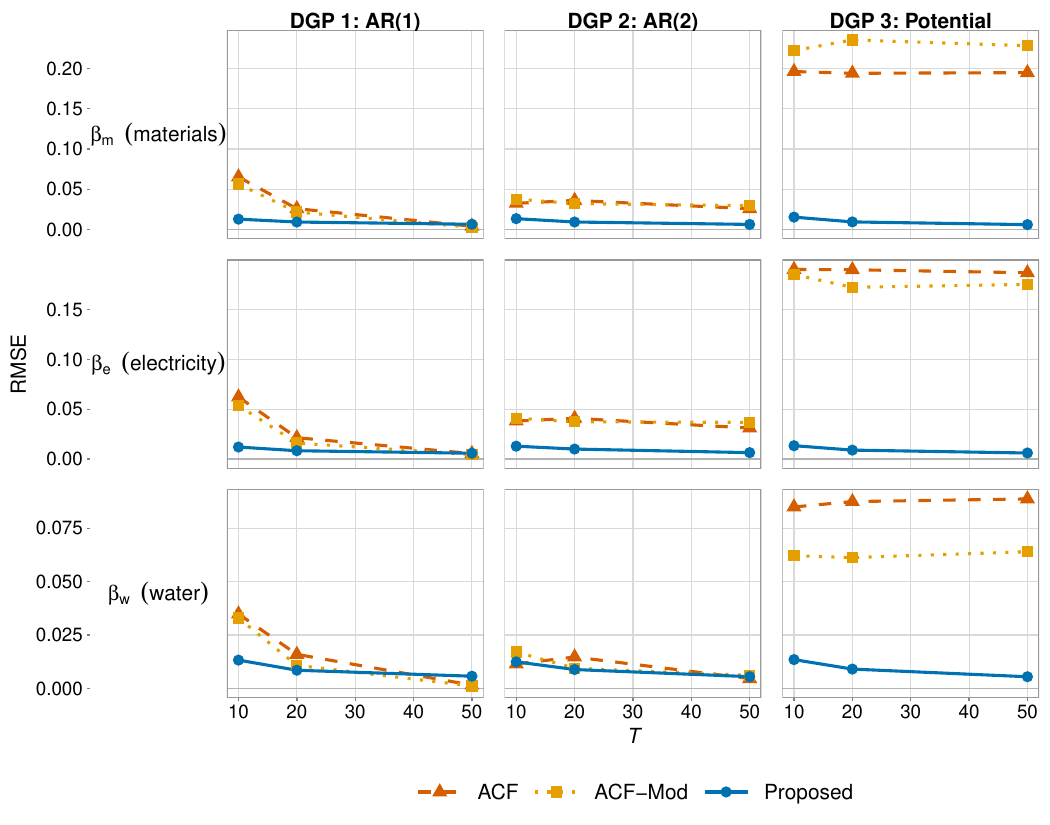}
\caption{Part~1: Mean RMSE Convergence ($N = 200$)}
\label{fig:rmse_convergence_N200}

\small\raggedright \textit{Notes: Mean RMSE of $(\hat{\beta}_m, \hat{\beta}_e, \hat{\beta}_w)$ as a function of $T$ ($N = 200$, $R = 100$).}
\end{figure}

\begin{figure}[htbp]
\centering
\includegraphics[width=\textwidth]{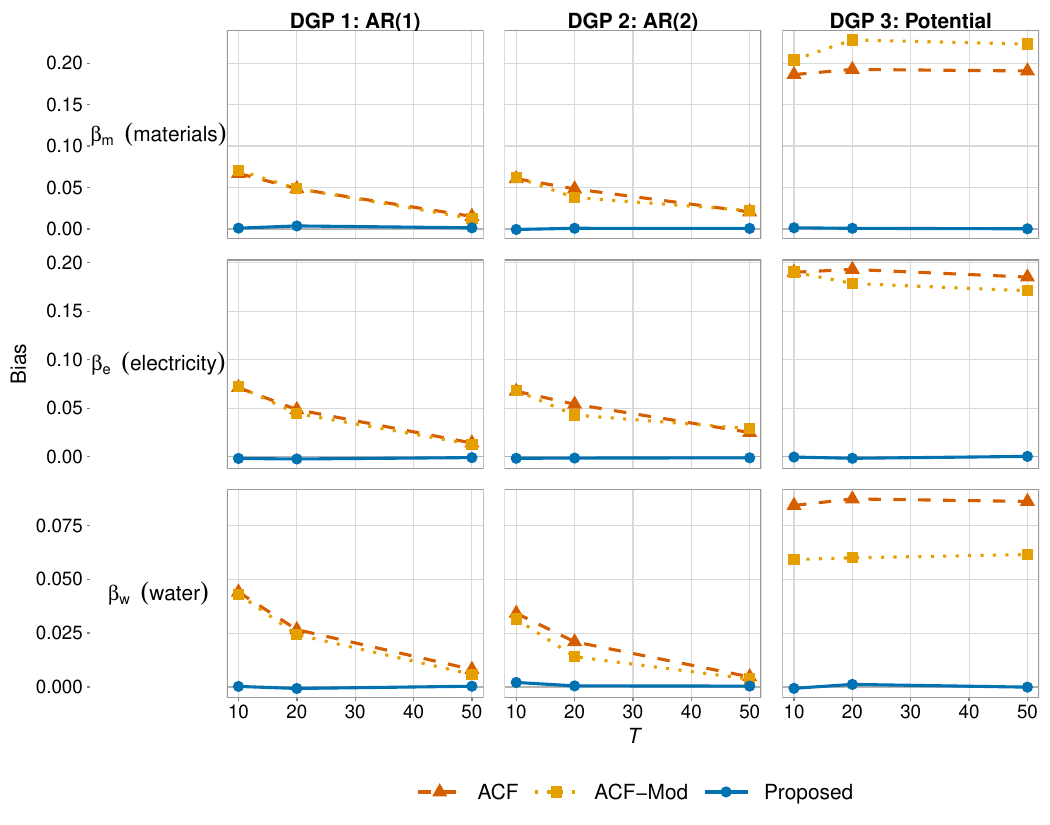}
\caption{Part~1: Mean Bias Convergence ($N = 50$)}
\label{fig:bias_convergence_N50}

\small\raggedright \textit{Notes: Same as Figure~\ref{fig:bias_convergence} but for $N = 50$. The qualitative patterns are preserved; the larger variance reflects the smaller cross-section.}
\end{figure}

\begin{figure}[htbp]
\centering
\includegraphics[width=\textwidth]{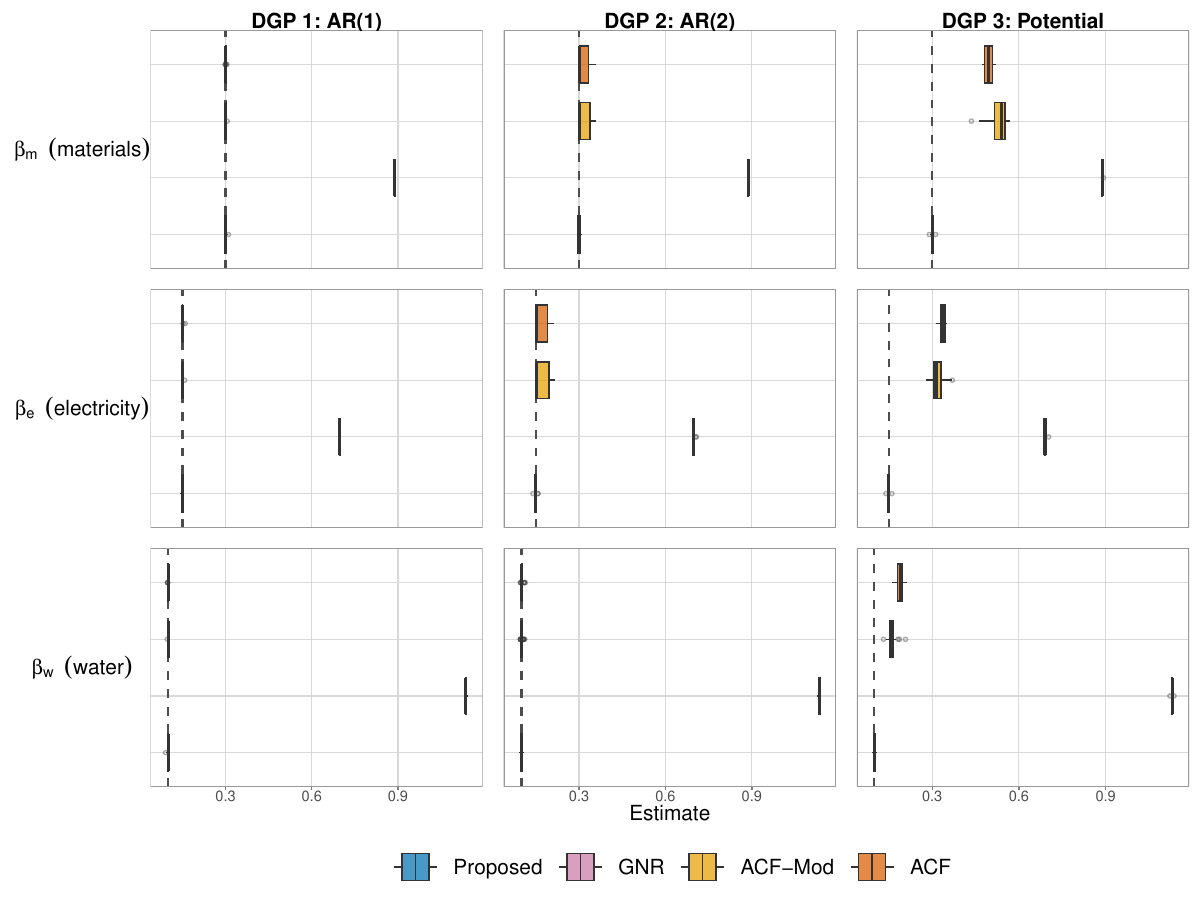}
\caption{Part~1: Four-Method Comparison ($N=500$, $T=50$)}
\label{fig:boxplot_all4_app}

\small\raggedright \textit{Notes: Distribution of $(\hat{\beta}_m, \hat{\beta}_e, \hat{\beta}_w)$ including the GNR estimator ($N = 500$, $T = 50$). The GNR estimates are severely biased ($\text{Bias}(\hat{\beta}_m) \approx 0.59$) due to persistent demand shocks ($\rho_\tau = 0.5$) violating the scalar unobservability assumption. The main text figures (Figures~\ref{fig:bias_convergence}--\ref{fig:boxplot_part1}) exclude GNR to preserve visual clarity for the ACF--Proposed comparison.}
\end{figure}

\begin{figure}[htbp]
\centering
\includegraphics[width=\textwidth]{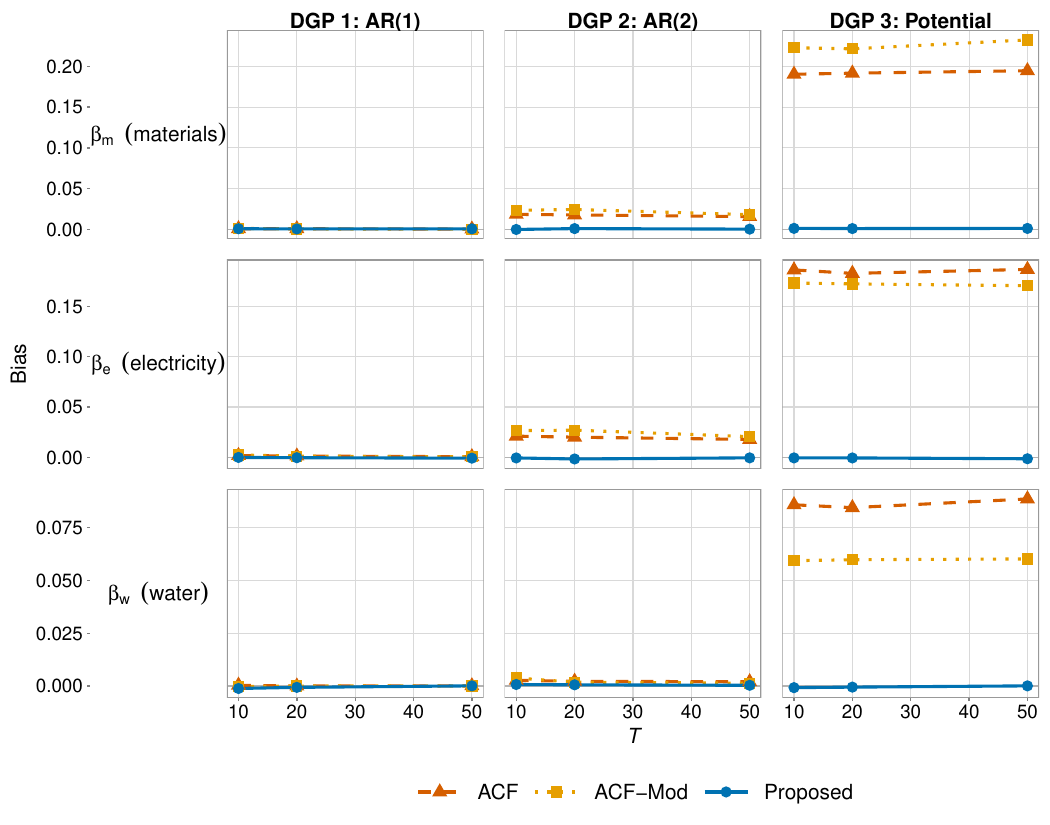}
\caption{Part~1: Three-Method Bias Convergence ($N = 500$)}
\label{fig:bias_convergence_3methods}

\small\raggedright \textit{Notes: Same as Figure~\ref{fig:bias_convergence} but including the ACF-Mod (oracle) estimator that observes the true demand shock $\tau_{jt}$. Under DGP~2 and~3, ACF-Mod bias is comparable to or larger than standard ACF.}
\end{figure}

\begin{figure}[htbp]
\centering
\includegraphics[width=\textwidth]{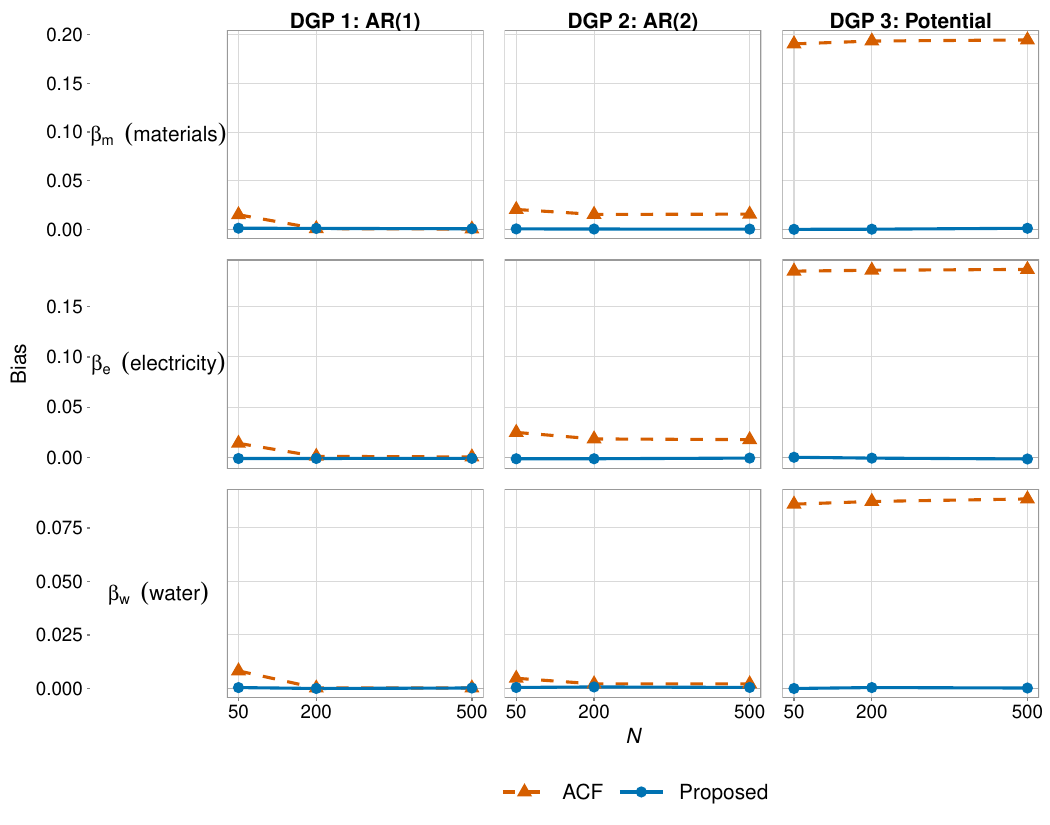}
\caption{Part~1: Mean Bias as a Function of $N$ ($T = 50$)}
\label{fig:bias_convergence_byN}

\small\raggedright \textit{Notes: Mean bias of $(\hat{\beta}_m, \hat{\beta}_e, \hat{\beta}_w)$ as a function of $N$ for $T = 50$ ($R = 100$). Increasing $N$ reduces variance for all estimators. ACF bias under DGP~2 and~3 does not vanish with $N$, confirming asymptotic bias.}
\end{figure}

\begin{figure}[htbp]
\centering
\includegraphics[width=\textwidth]{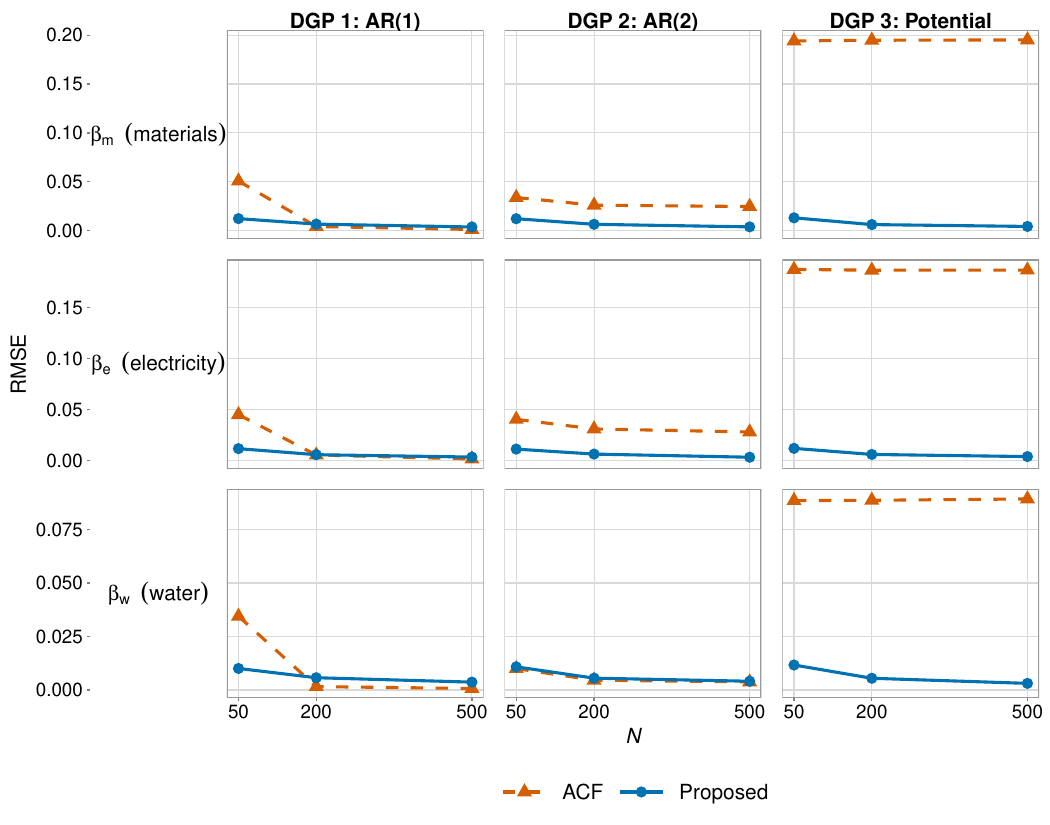}
\caption{Part~1: RMSE as a Function of $N$ ($T = 50$)}
\label{fig:rmse_convergence_byN}

\small\raggedright \textit{Notes: Root mean squared error of $(\hat{\beta}_m, \hat{\beta}_e, \hat{\beta}_w)$ as a function of $N$ for $T = 50$ ($R = 100$). Under DGP~1 (correct Markov specification), the RMSE of both estimators converges to zero at similar rates. Under DGP~2 and~3, ACF RMSE is bounded away from zero because bias dominates, whereas the proposed estimator's RMSE continues to decrease with $N$.}
\end{figure}

\subsection{DGP4: Conditional Independence Violation}
\label{subsec:mc_dgp4}

\begin{figure}[htbp]
\centering
\includegraphics[width=0.7\textwidth]{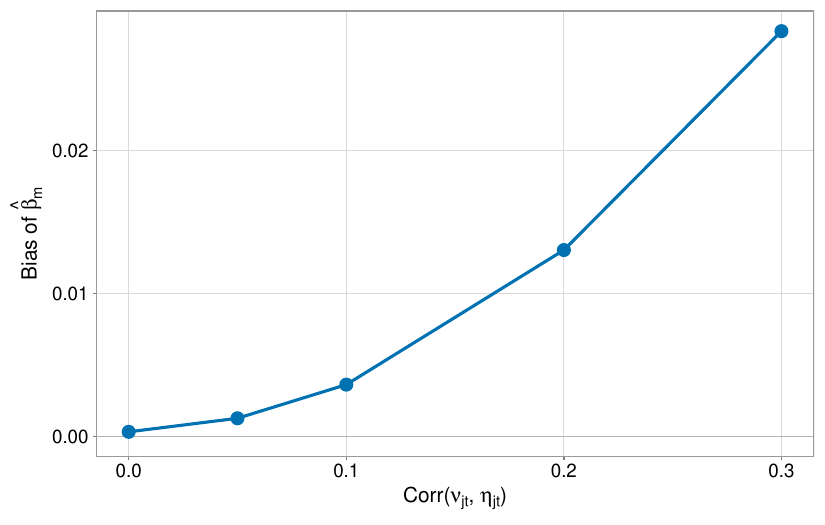}
\caption{DGP~4: Bias of $\hat{\beta}_m$ as a Function of $\mathrm{Corr}(\nu,\eta)$}
\label{fig:dgp4_bias_vs_rho}

\small\raggedright \textit{Notes: Mean bias of $\hat{\beta}_m$ as a function of $\mathrm{Corr}(\nu_{jt}, \eta_{jt})$ for the proposed method ($N=200$, $T=50$, $R=20$). When $\mathrm{Corr}(\nu, \eta) = 0$, the estimator is approximately unbiased at the true value $\beta_m = 0.30$. As the electricity--water correlation increases, $\hat{\zeta}_\omega$ is overestimated, causing upward bias in $\hat{\beta}_m$ (Appendix~\ref{app:ci_violation}). The bias direction is the same as the Markov misspecification bias in ACF, so the empirical gap (ACF $>$ Proposed) cannot be attributed to CI violation.}
\end{figure}

\begin{table}[!h]
\centering
\caption{\label{tab:mc_dgp4_ci}DGP 4 (CI Violation): $\hat{\beta}_m$ Bias, SD, and RMSE by $\mathrm{Corr}(\nu, \eta)$. $N=200$, $T=50$, $R=20$.}
\centering
\fontsize{9}{11}\selectfont
\begin{tabular}[t]{crrr}
\toprule
$\text{Corr}(\nu,\eta)$ & Bias & SD & RMSE\\
\midrule
0.0000 & 0.0003 & 0.0056 & 0.0054\\
0.0500 & 0.0013 & 0.0052 & 0.0052\\
0.1000 & 0.0036 & 0.0052 & 0.0062\\
0.2000 & 0.0130 & 0.0048 & 0.0139\\
0.3000 & 0.0283 & 0.0044 & 0.0287\\
\bottomrule
\end{tabular}
\end{table}

\begin{figure}[htbp]
\centering
\includegraphics[width=\textwidth]{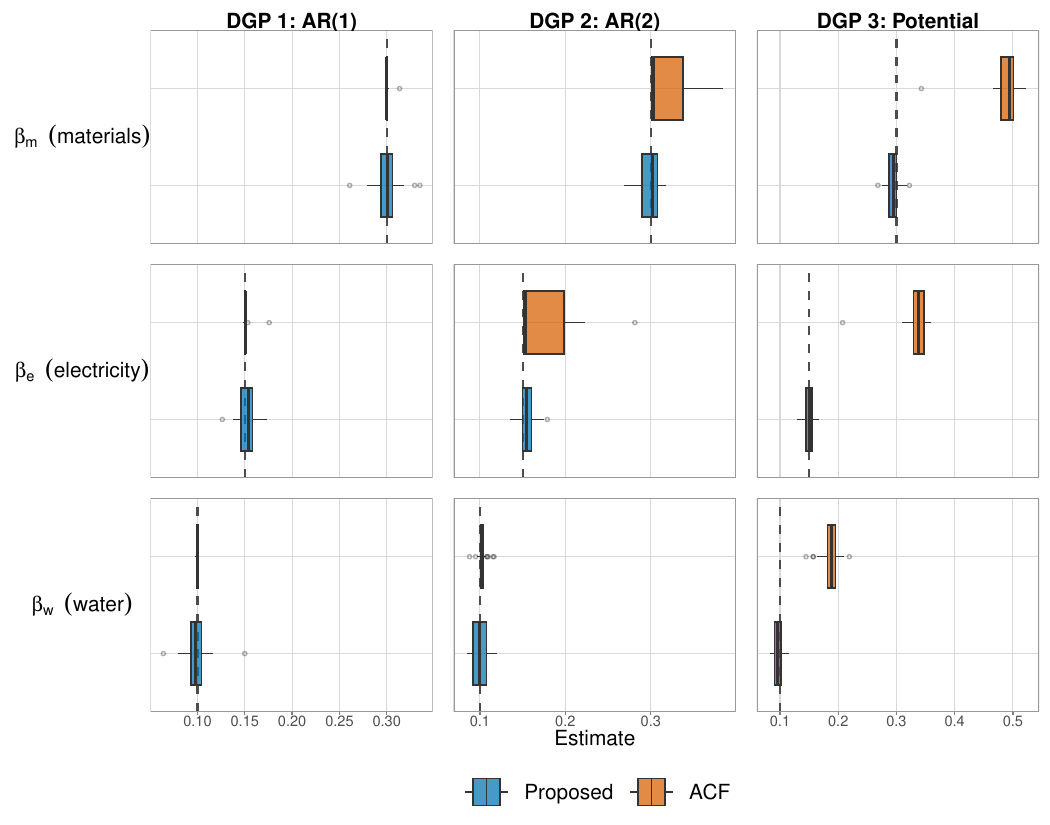}
\caption{Part~2: Flexible Input Parameters Under Block~A+B+C}
\label{fig:boxplot_part2_mew_app}

\small\raggedright \textit{Notes: Distribution of $(\hat{\beta}_m, \hat{\beta}_e, \hat{\beta}_w)$ from Part~2 (Block~A+B+C). These parameters are identified by Blocks~A and B alone; comparison with Part~1 confirms that the addition of Block~C moments does not contaminate the flexible input estimates.}
\end{figure}

\subsection{Block~C Identification Diagnostics}
\label{subsec:blockc_diagnostics}

The significance of $\hat{\rho}_2$ and $\hat{\rho}_3$ provides a
practical diagnostic for Block~C identification strength.
When these coefficients are statistically significant, the nonlinear
component $h(v)$ is identified beyond the linear term, enabling
separate estimation of $\beta_k$ and $\beta_l$.
When both are statistically insignificant, the capital-labor aggregate
$v(k,l)$ is empirically indistinguishable from a Cobb--Douglas form,
and the rank condition in Theorem~\ref{thm:homothetic_id} is not met
in the data: $(\beta_k, \beta_l)$ cannot be separately identified
through Block~C for such industries. This is a testable,
data-driven indicator of the proximity to the identification boundary
described in Section~\ref{sec:homothetic}.
In the empirical application (Section~\ref{sec:empirical}), I report
$t$-statistics for $\hat{\rho}_2$ and $\hat{\rho}_3$ alongside the
$J$-test to assess Block~C reliability for each industry.

\subsection{Block~A+B Only vs.\ Block~A+B+C Comparison}

A robustness check compares estimates obtained using Block~A+B moments alone against estimates using the full Block~A+B+C system.
Parameters identified by Block~A+B (namely $(\beta_m, \beta_e, \beta_w)$ and derived quantities such as the markup $\mu = \beta_m / s_m$) should remain stable across the two specifications.
Instability would indicate misspecification of the Block~C moments or weak identification of the CES aggregator structure.
See Table~\ref{tab:4group} and Section~\ref{sec:empirical} for the empirical comparison.


\section{Additional Results for Empirical Analysis}
\label{sec:empirical_additional}
\label{sec:appendix_empirical}

\FloatBarrier
\subsection{Cross-Industry Distribution of Parameter Estimates}

Table~\ref{tab:all_industry_beta_dist} reports the cross-industry
distribution (median, mean, SD) of the intermediate input elasticity
estimates ($\hat{\beta}_m$, $\hat{\beta}_e$, $\hat{\beta}_w$)
for both the proposed method (Panel~A) and ACF (Panel~B).
Capital and labor elasticities ($\hat{\beta}_k$, $\hat{\beta}_l$)
across all three methods are reported in Table~\ref{tab:beta_kl_summary}.

\begin{table}[!h]
\centering
\caption{\label{tab:all_industry_beta_dist}Cross-Industry Distribution of Intermediate Input Elasticity Estimates}
\centering
\begin{threeparttable}
\fontsize{11}{13}\selectfont
\begin{tabular}[t]{lrrrr}
\toprule
Parameter & N & Median & Mean & SD\\
\midrule
\addlinespace[0.3em]
\multicolumn{5}{l}{\textbf{Panel~A: Proposed Method ($N=502$)}}\\
\hspace{1em}$\hat{\beta}_m$ (Material) & 502 & 0.491 & 0.422 & 0.271\\
\hspace{1em}$\hat{\beta}_e$ (Electricity) & 502 & 0.001 & 0.079 & 0.191\\
\hspace{1em}$\hat{\beta}_w$ (Water) & 502 & 0.006 & 0.142 & 0.300\\
\addlinespace[0.3em]
\multicolumn{5}{l}{\textbf{Panel~B: ACF ($N=500$)}}\\
\hspace{1em}$\hat{\beta}_m$ (Material) & 500 & 0.565 & 0.535 & 0.191\\
\hspace{1em}$\hat{\beta}_e$ (Electricity) & 500 & 0.072 & 0.115 & 0.143\\
\hspace{1em}$\hat{\beta}_w$ (Water) & 500 & 0.017 & 0.055 & 0.095\\
\bottomrule
\end{tabular}
\begin{tablenotes}
\item \textit{Note: } 
\item Outliers $|\hat{\beta}|>2$ excluded from ACF panel.
\end{tablenotes}
\end{threeparttable}
\end{table}

\FloatBarrier
\subsection{Identification Cross-Check: Exclusion Restriction vs.\ Block~C}
\label{subsec:crosscheck_4group}

Table~\ref{tab:4group} presents the four-group comparison of
$\hat{\beta}_k$ and $\hat{\beta}_l$.
Groups~(i) and~(i-a) form the identification cross-check: the
exclusion restriction and Block~C applied to the same
\NconsGroup{} exclusion-consistent industries yield closely aligned
estimates of $\beta_k$ (median 0.010 vs.\ 0.039) and $\beta_l$
(median 0.219 vs.\ 0.341), indicating that the two conceptually
distinct identification strategies converge.
Groups~(ii) and~(iii) provide diagnostic contrast.

\begin{table}[ht]
\centering
\resizebox{\linewidth}{!}{\begin{tabular}{llrrrrrrr}
  \toprule
Group & Method & $N$ & $\hat{\beta}_k$ Median & $\hat{\beta}_k$ Mean & $\hat{\beta}_k$ SD & $\hat{\beta}_l$ Median & $\hat{\beta}_l$ Mean & $\hat{\beta}_l$ SD \\ 
  \midrule
(i)~~Excl. (consistent, $N=302$) & Excl. restriction & 302 & 0.010 & 0.010 & 0.038 & 0.219 & 0.204 & 0.208 \\ 
  (i-a) Block~C (consistent, $N=302$) & Block C & 302 & 0.039 & 0.048 & 0.046 & 0.341 & 0.339 & 0.192 \\ 
  (ii)~Excl. (inconsistent, $N=200$) & Excl. restriction & 200 & -0.013 & -0.073 & 0.555 & 0.184 & 0.093 & 0.573 \\ 
  (iii) Block~C (all, $N=502$) & Block C & 502 & 0.035 & 0.048 & 0.054 & 0.332 & 0.336 & 0.224 \\ 
   \midrule
\multicolumn{9}{l}{\footnotesize \textit{Notes:} Outliers $|\hat{\beta}|>2$ excluded. Groups~(i) and~(i-a) are applied to the identical set of industries, forming the identification cross-check.}\\[2pt]
 \bottomrule
\end{tabular}}
\caption{Four-Group Comparison of $\hat{\beta}_k$ and $\hat{\beta}_l$: Exclusion Restriction vs.\ Block~C} 
\label{tab:4group}
\end{table}

\FloatBarrier
\subsection{Demand Function Parameter Estimates}

Table~\ref{tab:all_industry_demand_dist} reports the cross-industry
distribution (median, mean, SD) of the input demand function
parameter estimates. Panel~A covers the productivity loading
parameters ($\hat{\gamma}_\omega$, $\hat{\delta}_\omega$,
$\hat{\zeta}_\omega$) identified by Block~B; Panel~B covers the
demand slopes on capital and labor.

\begin{table}[!h]
\centering
\caption{\label{tab:all_industry_demand_dist}Cross-Industry Distribution of Input Demand Function Parameter Estimates (Block~A+B)}
\centering
\begin{threeparttable}
\fontsize{11}{13}\selectfont
\begin{tabular}[t]{lrrrr}
\toprule
Parameter & N & Median & Mean & SD\\
\midrule
\addlinespace[0.3em]
\multicolumn{5}{l}{\textbf{Panel~A: Productivity Loading (Block~B)}}\\
\hspace{1em}$\hat{\gamma}_{\omega}$ (Material) & 502 & 1.547 & 2.498 & 3.239\\
\hspace{1em}$\hat{\delta}_{\omega}$ (Water) & 502 & 1.334 & 4.537 & 5.819\\
\hspace{1em}$\hat{\zeta}_{\omega}$ (Electricity) & 502 & 1.325 & 3.210 & 4.676\\
\addlinespace[0.3em]
\multicolumn{5}{l}{\textbf{Panel~B: Demand Slopes on $(k, l)$}}\\
\hspace{1em}$\hat{\gamma}_k$ (Capital) & 502 & 0.017 & 0.002 & 0.252\\
\hspace{1em}$\hat{\gamma}_l$ (Labor) & 502 & 0.010 & -0.253 & 1.541\\
\hspace{1em}$\hat{\delta}_k$ (Capital) & 502 & -0.050 & -0.045 & 0.484\\
\hspace{1em}$\hat{\delta}_l$ (Labor) & 502 & -0.978 & -1.280 & 2.444\\
\hspace{1em}$\hat{\zeta}_k$ (Capital) & 502 & -0.028 & -0.040 & 0.466\\
\hspace{1em}$\hat{\zeta}_l$ (Labor) & 502 & -0.246 & -0.604 & 2.203\\
\bottomrule
\end{tabular}
\begin{tablenotes}
\item \textit{Note: } 
\item 502 manufacturing industries (Block~A+B convergence). Outliers $|\hat{\theta}| > 10$ excluded from summary statistics.
\end{tablenotes}
\end{threeparttable}
\end{table}

\FloatBarrier
\subsection{Event Study Results}
\label{subsec:earthquake_es}

See Section~\ref{subsec:event_study_body} for the main DiD analysis
(Table~\ref{tab:earthquake_did}).
Table~\ref{tab:earthquake_es} below reports the full
\textcite{sun2021estimating} year-by-year coefficient estimates,
confirming flat pre-trends.

%

\begin{table}[htbp]
   \caption{\label{tab:earthquake_es} 2011 Tohoku Earthquake: Sun-Abraham (2021) Event Study}
   \bigskip
   \centering
   \begin{tabular}{lcc}
      \toprule
                                      & Proposed       & ACF \\   
                                      & (1)            & (2)\\  
      \midrule 
      time $=$ -8                     & 0.0022         & 0.0124\\   
                                      & (0.0129)       & (0.0104)\\   
      time $=$ -7                     & 0.0024         & 0.0167$^{*}$\\   
                                      & (0.0126)       & (0.0099)\\   
      time $=$ -6                     & 0.0086         & 0.0192$^{**}$\\   
                                      & (0.0123)       & (0.0095)\\   
      time $=$ -5                     & 0.0015         & 0.0047\\   
                                      & (0.0121)       & (0.0098)\\   
      time $=$ -4                     & 0.0116         & 0.0099\\   
                                      & (0.0121)       & (0.0098)\\   
      time $=$ -3                     & 0.0131         & -0.0037\\   
                                      & (0.0097)       & (0.0080)\\   
      time $=$ -2                     & 0.0014         & 0.0042\\   
                                      & (0.0096)       & (0.0080)\\   
      time $=$ 0                      & -0.0288$^{**}$ & -0.0253$^{**}$\\   
                                      & (0.0138)       & (0.0122)\\   
      time $=$ 1                      & -0.0010        & -0.0029\\   
                                      & (0.0097)       & (0.0085)\\   
      time $=$ 2                      & -0.0024        & -0.0042\\   
                                      & (0.0096)       & (0.0081)\\   
      time $=$ 3                      & -0.0031        & -0.0052\\   
                                      & (0.0096)       & (0.0081)\\   
      time $=$ 4                      & -0.0120        & -0.0295$^{***}$\\   
                                      & (0.0130)       & (0.0114)\\   
      time $=$ 5                      & -0.0061        & -0.0117\\   
                                      & (0.0102)       & (0.0082)\\   
      time $=$ 6                      & -0.0054        & -0.0062\\   
                                      & (0.0102)       & (0.0081)\\   
      time $=$ 7                      & 0.0005         & -0.0097\\   
                                      & (0.0100)       & (0.0081)\\   
      time $=$ 8                      & -0.0004        & -0.0126\\   
                                      & (0.0103)       & (0.0082)\\   
      time $=$ 9                      & -0.0380$^{**}$ & -0.0261$^{**}$\\   
                                      & (0.0154)       & (0.0114)\\   
       \\
      Observations                    & 219,573        & 219,573\\  
      R$^2$                           & 0.98729        & 0.94606\\  
      $\text{poly}(k,\ell)$ control   & $\checkmark$   & \\  
       \\
      Firm fixed effects              & $\checkmark$   & $\checkmark$\\   
      Ind.$\times$Year fixed effects  & $\checkmark$   & $\checkmark$\\   
      \bottomrule
   \end{tabular}
   
   \par \raggedright 
   \textcite{sun2021estimating} estimator. Single cohort 2011; SA estimator coincides with plain event study numerically. Treatment: Iwate, Miyagi, Fukushima (seismic intensity $\geq$ 6-strong). Control: West Japan (prefectures 25--47). Fixed effects: firm + industry$\times$year. Proposed: nonparametric $\text{poly}(k,\ell)$ degree~3 control for $\Delta(k,\ell)$. ACF: no polynomial control ($k,\ell$ already subtracted in $\hat{\omega}^{\mathrm{ACF}}$). Heteroskedasticity-robust standard errors. Reference period: $t=-1$.
\end{table}

\FloatBarrier
\subsection{Time-Varying Parameter Estimates}
\label{subsec:annual_params}

Because the proposed estimator identifies the production function
from static covariances alone, it can be applied to each
cross-section separately, tracking parameter evolution over time
without imposing structural stability.
Table~\ref{tab:annual_params} reports annual Block~A+B estimates
of $(\beta_m, \beta_e, \beta_w)$ for four representative industries.
Production elasticities exhibit variation across years; where
confidence intervals do not overlap, the variation is statistically
significant and inconsistent with time-invariant parameters.
Annual cross-sections are smaller than the pooled sample,
yielding wider confidence intervals in some periods.

\begin{table}[!htbp]
\centering
\caption{\label{tab:annual_params}Annual Production Function Parameter Estimates (Proposed Method)}
\centering
\resizebox{\textwidth}{!}{
\fontsize{10}{12}\selectfont
\begin{tabular}[t]{lcccccccccccccccccc}
\toprule
Parameter & 2003 & 2004 & 2005 & 2006 & 2007 & 2008 & 2009 & 2010 & 2011 & 2012 & 2013 & 2014 & 2015 & 2016 & 2017 & 2018 & 2019 & 2020\\
\midrule
\addlinespace[0.3em]
\multicolumn{19}{l}{\textbf{Bread}}\\
\hspace{1em}$\beta_m$ (Material) & \makecell[c]{0.679\\(0.037)} & \makecell[c]{0.633\\(0.160)} & \makecell[c]{0.681\\(0.030)} & \makecell[c]{0.000\\(0.394)} & \makecell[c]{0.681\\(0.041)} & \makecell[c]{0.638\\(0.076)} & \makecell[c]{0.649\\(0.054)} & \makecell[c]{0.000\\(1.142)} & \makecell[c]{0.486\\(0.452)} & \makecell[c]{0.000\\(0.236)} & \makecell[c]{0.679\\(0.036)} & \makecell[c]{0.000\\(0.245)} & \makecell[c]{0.000\\(1.042)} & \makecell[c]{0.266\\(0.224)} & \makecell[c]{0.692\\(0.046)} & \makecell[c]{0.547\\(0.069)} & \makecell[c]{0.553\\(0.049)} & \makecell[c]{0.000\\(0.625)}\\
\hspace{1em}$\beta_e$ (Electricity) & \makecell[c]{0.000\\} & \makecell[c]{0.001\\} & \makecell[c]{0.005\\} & \makecell[c]{0.000\\} & \makecell[c]{0.000\\} & \makecell[c]{0.000\\} & \makecell[c]{0.000\\} & \makecell[c]{0.000\\} & \makecell[c]{0.011\\} & \makecell[c]{0.000\\} & \makecell[c]{0.000\\} & \makecell[c]{0.000\\} & \makecell[c]{0.000\\} & \makecell[c]{0.057\\} & \makecell[c]{0.001\\} & \makecell[c]{0.000\\} & \makecell[c]{0.000\\} & \makecell[c]{0.000\\}\\
\hspace{1em}$\beta_w$ (Water) & \makecell[c]{0.013\\} & \makecell[c]{0.017\\} & \makecell[c]{0.023\\} & \makecell[c]{0.019\\} & \makecell[c]{0.040\\} & \makecell[c]{0.026\\} & \makecell[c]{0.039\\} & \makecell[c]{0.000\\} & \makecell[c]{0.013\\} & \makecell[c]{0.000\\} & \makecell[c]{0.049\\} & \makecell[c]{0.027\\} & \makecell[c]{0.000\\} & \makecell[c]{0.014\\} & \makecell[c]{0.001\\} & \makecell[c]{0.000\\} & \makecell[c]{0.000\\} & \makecell[c]{0.000\\}\\
\addlinespace[0.3em]
\multicolumn{19}{l}{\textbf{Corrugated board boxes}}\\
\hspace{1em}$\beta_m$ (Material) & - & - & - & - & - & \makecell[c]{0.663\\(0.045)} & \makecell[c]{0.611\\(0.032)} & \makecell[c]{0.655\\(0.051)} & \makecell[c]{0.060\\(1.147)} & \makecell[c]{0.547\\(0.074)} & \makecell[c]{0.539\\(0.070)} & \makecell[c]{0.682\\(0.055)} & \makecell[c]{0.556\\(0.068)} & \makecell[c]{0.591\\(0.030)} & \makecell[c]{0.536\\(0.026)} & \makecell[c]{0.000\\(0.322)} & \makecell[c]{0.636\\(0.035)} & \makecell[c]{0.587\\(0.057)}\\
\hspace{1em}$\beta_e$ (Electricity) & - & - & - & - & - & \makecell[c]{0.000\\} & \makecell[c]{0.000\\} & \makecell[c]{0.000\\} & \makecell[c]{0.113\\} & \makecell[c]{0.000\\} & \makecell[c]{0.000\\} & \makecell[c]{0.010\\} & \makecell[c]{0.008\\} & \makecell[c]{0.000\\} & \makecell[c]{0.000\\} & \makecell[c]{0.000\\} & \makecell[c]{0.000\\} & \makecell[c]{0.000\\}\\
\hspace{1em}$\beta_w$ (Water) & - & - & - & - & - & \makecell[c]{0.000\\} & \makecell[c]{0.000\\} & \makecell[c]{0.002\\} & \makecell[c]{0.000\\} & \makecell[c]{0.007\\} & \makecell[c]{0.016\\} & \makecell[c]{0.000\\} & \makecell[c]{0.000\\} & \makecell[c]{0.000\\} & \makecell[c]{0.000\\} & \makecell[c]{0.000\\} & \makecell[c]{0.000\\} & \makecell[c]{0.007\\}\\
\addlinespace[0.3em]
\multicolumn{19}{l}{\textbf{Plastic film}}\\
\hspace{1em}$\beta_m$ (Material) & - & - & - & - & - & \makecell[c]{0.502\\(0.052)} & \makecell[c]{0.421\\(0.042)} & \makecell[c]{0.503\\(0.040)} & \makecell[c]{0.000\\(0.189)} & \makecell[c]{0.552\\(0.049)} & \makecell[c]{0.492\\(0.552)} & \makecell[c]{0.104\\(0.275)} & \makecell[c]{0.009\\(0.712)} & \makecell[c]{0.426\\(0.046)} & \makecell[c]{0.425\\(0.050)} & \makecell[c]{0.455\\(0.040)} & \makecell[c]{0.103\\(0.194)} & \makecell[c]{0.573\\(0.082)}\\
\hspace{1em}$\beta_e$ (Electricity) & - & - & - & - & - & \makecell[c]{0.000\\} & \makecell[c]{0.000\\} & \makecell[c]{0.000\\} & \makecell[c]{0.000\\} & \makecell[c]{0.000\\} & \makecell[c]{0.012\\} & \makecell[c]{0.000\\} & \makecell[c]{0.000\\} & \makecell[c]{0.000\\} & \makecell[c]{0.001\\} & \makecell[c]{0.000\\} & \makecell[c]{0.000\\} & \makecell[c]{0.000\\}\\
\hspace{1em}$\beta_w$ (Water) & - & - & - & - & - & \makecell[c]{0.000\\} & \makecell[c]{0.001\\} & \makecell[c]{0.001\\} & \makecell[c]{0.000\\} & \makecell[c]{0.008\\} & \makecell[c]{0.002\\} & \makecell[c]{0.016\\} & \makecell[c]{0.000\\} & \makecell[c]{0.028\\} & \makecell[c]{0.021\\} & \makecell[c]{0.020\\} & \makecell[c]{0.000\\} & \makecell[c]{0.024\\}\\
\addlinespace[0.3em]
\multicolumn{19}{l}{\textbf{Robots}}\\
\hspace{1em}$\beta_m$ (Material) & \makecell[c]{0.425\\(0.040)} & \makecell[c]{0.395\\(0.037)} & \makecell[c]{0.413\\(0.060)} & \makecell[c]{0.384\\(0.039)} & \makecell[c]{0.406\\(0.041)} & \makecell[c]{1.000\\(0.167)} & \makecell[c]{0.715\\(0.182)} & \makecell[c]{0.588\\(0.381)} & - & \makecell[c]{0.000\\(0.563)} & \makecell[c]{0.000\\(0.486)} & \makecell[c]{0.585\\(0.070)} & - & \makecell[c]{0.439\\(0.584)} & \makecell[c]{0.304\\(0.331)} & \makecell[c]{0.445\\(1.013)} & \makecell[c]{0.000\\(0.306)} & -\\
\hspace{1em}$\beta_e$ (Electricity) & \makecell[c]{0.034\\} & \makecell[c]{0.135\\} & \makecell[c]{0.000\\} & \makecell[c]{0.007\\} & \makecell[c]{0.000\\} & \makecell[c]{0.067\\} & \makecell[c]{0.241\\} & \makecell[c]{0.009\\} & - & \makecell[c]{0.000\\} & \makecell[c]{0.000\\} & \makecell[c]{0.002\\} & - & \makecell[c]{0.403\\} & \makecell[c]{0.000\\} & \makecell[c]{0.079\\} & \makecell[c]{0.000\\} & -\\
\hspace{1em}$\beta_w$ (Water) & \makecell[c]{0.008\\} & \makecell[c]{0.000\\} & \makecell[c]{0.004\\} & \makecell[c]{0.002\\} & \makecell[c]{0.000\\} & \makecell[c]{0.130\\} & \makecell[c]{0.094\\} & \makecell[c]{0.950\\} & - & \makecell[c]{0.000\\} & \makecell[c]{0.124\\} & \makecell[c]{0.094\\} & - & \makecell[c]{0.324\\} & \makecell[c]{0.000\\} & \makecell[c]{0.000\\} & \makecell[c]{0.000\\} & -\\
\bottomrule
\end{tabular}}

\begin{minipage}{\linewidth}\footnotesize\raggedright \textit{Note:} Estimates from annual cross-sectional GMM (Proposed Method). Analytical standard errors in parentheses; SEs $<$ 0.001 reported as such.\end{minipage}
\end{table}

\FloatBarrier
\subsection{Block C: Homothetic Recovery of $(\beta_k, \beta_l)$}
\label{subsec:blockc_results}

Table~\ref{tab:blockc_diagnostics} reports \emph{internal} Block~C
diagnostics: the cross-industry distribution of the estimated CES
substitution parameter $\hat\rho_v$ and capital share $\hat\alpha$
(left panel), and the significance rates of the higher-order curvature
terms $\hat\rho_2, \hat\rho_3$ together with the stability of
$\hat\beta_m$ when Block~C moments are added (right panel).
These diagnostics assess whether the CES curvature assumption
holds within each industry.
The subsequent cross-tabulation (Table~\ref{tab:4group})
is an \emph{external} cross-check: it asks whether industries
that pass the internal exclusion diagnostic also tend to pass
the Block~C $J$-test, providing evidence that the two
identification strategies are consistent with each other.

\begin{table}[htbp]
\centering
\caption{Block~C Estimation Diagnostics ($N=502$ industries)}
\label{tab:blockc_diagnostics}
\small
\begin{tabular}{l r r l r}
\toprule
\multicolumn{3}{c}{\textit{CES parameters}} & \multicolumn{2}{c}{\textit{Specification diagnostics}} \\
\cmidrule(r){1-3}\cmidrule(l){4-5}
Parameter & Median & Mean & Statistic & Value \\
\midrule
$\hat{\rho}_v$ & -1.000 & -0.149 & $\hat{\rho}_2$ significant ($|t|>1.96$) & 153/502 \\
$\hat{\alpha}$ & 0.500 & 0.533 & $\hat{\rho}_3$ significant ($|t|>1.96$) & 178/502 \\
 &  &  & $\Delta\hat{\beta}_m$ median (A+B$\to$A+B+C) & 0.0003 \\
\bottomrule
\end{tabular}
\vspace{0.5em}
\begin{minipage}{0.95\textwidth}\footnotesize
\textit{Notes:} $\hat{\rho}_v$: CES substitution parameter; $\hat{\alpha}$: capital share in CES aggregator. $\hat{\rho}_2, \hat{\rho}_3$: higher-order CES terms (Block~C curvature instruments). $\Delta\hat{\beta}_m$: change in materials elasticity when Block~C moments are added.
\end{minipage}
\end{table}

\paragraph{Cross-validation of two identification strategies.}

The exclusion restriction (Section~\ref{sec:exclusion}) and the
homothetic regularity condition
(Section~\ref{sec:homothetic}) provide two independent
routes to identifying $(\beta_k, \beta_l)$. Their joint
behavior across industries provides an indirect
validity check that does not rely on either restriction alone.

I classify each industry along two dimensions:
(i)~\emph{exclusion consistency}, defined as the maximum
absolute gap in $\hat{\beta}_k$ (or $\hat{\beta}_l$) across
the three proxy-specific OLS estimates being below~0.2
(chosen as roughly one within-group standard deviation of
$\hat{\beta}_k$ across all industries; results are qualitatively
robust to thresholds of 0.1 and 0.3); and
(ii)~\emph{Block~C specification}, defined as non-rejection
of the Block~C $J$-test at the 5\% level (the full Block~A+B+C
system is overidentified).

Among the 502 industries, \PctConsPass\% of exclusion-consistent
industries also pass the Block~C $J$-test, compared to \PctInconsPass\%
among exclusion-inconsistent industries
(\NconsGroup{} consistent, \NinconsGroup{} inconsistent).
Among the \NbothPass~industries where both criteria are satisfied,
the cross-method correlation of $\hat{\beta}_k$ reaches \CorBkCrossVal,
indicating that the two conceptually distinct identification strategies
converge to similar estimates when their respective
maintained assumptions are empirically supported.

This pattern is informative about the source of Block~C
$J$-test failures. A logistic regression of Block~C
$J$-test passage on industry characteristics finds that
$|d_k|$---the exclusion restriction diagnostic
(Remark~\ref{rem:testable_exclusion})---is the only
statistically significant predictor ($p = 0.004$); sample
size, $\hat{\rho}_v$, and $\hat{\alpha}$ are all
insignificant. Industries with large $|d_k|$ (median~1.16
among those failing both criteria) exhibit demand functions
that respond strongly to capital and labor conditional on
productivity, violating the exclusion restriction. In such
industries, the Block~A+B estimates of the $g$-function
slopes carry substantial contamination from the
$(k,l)$-direction, which propagates into Block~C through
the constructed productivity index.

These findings support the interpretation that the Block~C
$J$-test rejection primarily reflects misspecification
transmitted from the demand-side moment conditions, rather
than failure of the homothetic production function
assumption per se. The \NbothPass~industries satisfying both
criteria serve as an internally validated subsample in
which the full three-block GMM system is well-specified.

\section{Identification Proofs and Technical Remarks}
\label{app:id_proofs}

This appendix collects proofs and technical remarks that
supplement the identification results in
Section~\ref{sec:model_identification}.

\subsection{Testability of the Exclusion Restriction}
\label{app:testability_exclusion}

This subsection provides the detailed derivation supporting
Remark~\ref{rem:testable_exclusion}.

The Wald test of $d_k = d_l = 0$ described in the main text
has 2 degrees of freedom, corresponding to the two parametric
restrictions.

With three inputs, two independent pairwise differences
for $\beta_k$ and two for $\beta_l$ provide four
testable implications; the formal test has two degrees of
freedom, corresponding to the parametric constraints
$d_k = d_l = 0$. The null hypothesis is
slightly weaker than the full exclusion restriction: it
requires $a_{k}^{h}/a_{\omega}^{h}$ to be common across
inputs, a condition satisfied by the exclusion restriction
but also by a knife-edge proportional response with no
structural basis when the three inputs involve distinct
production technologies.

\subsection{Proof of Proposition~\ref{prop:excl_ols}}
\label{app:proof_excl_ols}

\begin{proof}
\textbf{Preliminary: the observational equivalence and demand
function parameters.}
Under the linear specification, the observational equivalence
of Theorem~\ref{thm:obs_equiv} implies that Block~A+B
estimates satisfy, for each input $h$,
\begin{equation}
  \hat{a}_{k}^{h *} \xrightarrow{p} a_{k}^{h} + a_{\omega}^{h}\,c_k,
  \qquad
  \hat{a}_{l}^{h *} \xrightarrow{p} a_{l}^{h} + a_{\omega}^{h}\,c_l,
  \label{eq:indeterminacy_demand}
\end{equation}
for some constants $(c_k, c_l)$ characterizing the equivalence
class, while $\hat{a}_{z}^{h} \xrightarrow{p} a_{z}^{h}$ and
$\hat{a}_{\omega}^{h} \xrightarrow{p} a_{\omega}^{h}$ are
unaffected by the indeterminacy (since $z$ and $\omega$ are
orthogonal to the $(k,l)$-direction of the shift).

\medskip
\noindent\textbf{Case~1.}
Under $a_{k}^{h} = a_{l}^{h} = 0$,
equation~\eqref{eq:indeterminacy_demand} gives
$a_{k}^{h*} = a_{\omega}^{h}\,c_k$ and
$a_{l}^{h*} = a_{\omega}^{h}\,c_l$. Since
$a_{\omega}^{h} \neq 0$, the exclusion restriction forces
$c_k = c_l = 0$, resolving the indeterminacy completely.
Note, however, that the OLS consistency established below does
not require this global identification: it follows directly
from the proxy construction~\eqref{eq:omega_proxy_joint},
which uses only the invariant estimates $\hat{a}_{z}^{h}$ and
$\hat{a}_{\omega}^{h}$ and does not involve $(k, l)$.

The proxy~\eqref{eq:omega_proxy_joint} satisfies
\[
  \hat{\omega}_{jt}^h
  \xrightarrow{p}
  \frac{a_{z}^{h\prime}\,z + a_{\omega}^{h}\,\omega + \eta^h
    - a_{z}^{h\prime}\,z}{a_{\omega}^{h}}
  = \omega + \frac{\eta^h}{a_{\omega}^{h}},
\]
and the regression equation becomes
\[
  \tilde{y}_{jt} - \hat{\omega}_{jt}^h
  = \beta_0 + \beta_k\,k_{jt} + \beta_l\,l_{jt}
  + \varepsilon_{jt} - \frac{\eta_{jt}^h}{a_{\omega}^{h}}.
\]
OLS consistency requires
$\mathbb{E}[\varepsilon - \eta^h/a_{\omega}^{h} \mid k, l] = 0$.
The first term vanishes by
Assumption~\ref{ass:additive_error}. For the second, the
law of iterated expectations yields
\[
  \mathbb{E}[\eta^h \mid k, l]
  = \mathbb{E}\bigl[\,
    \underbrace{%
      \mathbb{E}[\eta^h \mid k, l, \omega, z]
    }_{= 0}
    \;\bigr|\; k, l\,\bigr] = 0,
\]
where the inner expectation vanishes because $a_{k}^{h} =
a_{l}^{h} = 0$ ensures $\eta^h = h - a_{z}^{h\prime}z
- a_{\omega}^{h}\,\omega$ has no residual dependence on
$(k, l)$. Both $\beta_k$ and $\beta_l$ are identified.

\medskip
\noindent\textbf{Case~2.}
Under $a_{k}^{h_1} = 0$ (with $a_{l}^{h_1}$ possibly nonzero),
the Block~A+B estimate satisfies
$\hat{a}_{l}^{h_1 *} \to a_{l}^{h_1} + a_{\omega}^{h_1}\,c_l$.
The proxy~\eqref{eq:proxy_h1} satisfies
\begin{align*}
  \hat{\omega}^{h_1}
  &\xrightarrow{p}
  \frac{a_{l}^{h_1}\,l + a_{z}^{h_1\prime}\,z
    + a_{\omega}^{h_1}\,\omega + \eta^{h_1}
    - (a_{l}^{h_1} + a_{\omega}^{h_1}\,c_l)\,l
    - a_{z}^{h_1\prime}\,z}{a_{\omega}^{h_1}} \\
  &= \omega - c_l\,l + \frac{\eta^{h_1}}{a_{\omega}^{h_1}},
\end{align*}
The regression equation becomes
\[
  \tilde{y}_{jt} - \hat{\omega}_{jt}^{h_1}
  = \beta_0 + \beta_k\,k_{jt}
  + (\beta_l + c_l)\,l_{jt}
  + \varepsilon_{jt} - \frac{\eta_{jt}^{h_1}}{a_{\omega}^{h_1}}.
\]
By the same iterated expectations argument,
$\mathbb{E}[\varepsilon - \eta^{h_1}/a_{\omega}^{h_1}
\mid k, l] = 0$, so OLS consistently estimates the
coefficient on $k$ as $\beta_k$ and the coefficient on $l$
as $\beta_l + c_l$. The \emph{capital} elasticity $\beta_k$
is identified; the labor coefficient carries the
indeterminacy $c_l$.

By a symmetric argument using $h_2$ (with $a_{l}^{h_2} = 0$),
the proxy~\eqref{eq:proxy_h2} yields a regression where the
coefficient on $l$ equals $\beta_l$ (identified) and the
coefficient on $k$ equals $\beta_k + c_k$ (biased).
Combining the two regressions identifies both $\beta_k$
and $\beta_l$.
\end{proof}

\subsection{Nonlinear Specifications}
\label{app:nonlinear_specs}

\begin{remark}[On nonlinear specifications]
\label{rem:poly_warning}
One might consider replacing the subtraction of
$\hat{\omega}^h$ in Proposition~\ref{prop:excl_ols} with a
polynomial regression of
$\tilde{y}$ on $(k, l, \hat{\omega}^h,
(\hat{\omega}^h)^2, (\hat{\omega}^h)^3)$. This is
\emph{not} consistent for $\beta_k$ and $\beta_l$ in
general. Since $\hat{\omega}^h = \omega + \eta^h /
a_{\omega}^{h}$ is a noisy proxy, the conditional expectation
$\mathbb{E}[\omega \mid k, l, \hat{\omega}^h]$ depends on
$(k, l)$ through signal extraction, and a polynomial in
$\hat{\omega}^h$ alone cannot absorb this component. By
contrast, fixing the coefficient on $\hat{\omega}^h$ to
unity (as in Proposition~\ref{prop:excl_ols}) eliminates
$\omega$ algebraically, avoiding this problem.
\end{remark}

\subsection{Proof of Proposition~\ref{prop:omega_hat_D}}
\label{app:proof_omega_hat_D}

\begin{proof}
Under conditions~(i) and~(ii), the intermediate inputs serve as
sufficient statistics for $\omega$ given the state variables: once
$(x, m, e, w)$ are observed, knowing $D$ provides no additional
information about $\omega$.
Formally,
\begin{equation}
  \mathbb{E}[\omega_{jt} \mid D_{jt}, x_{jt}, m_{jt}, e_{jt},
  w_{jt}]
  = \mathbb{E}[\omega_{jt} \mid x_{jt}, m_{jt}, e_{jt}, w_{jt}].
  \label{eq:suff_stat}
\end{equation}
To see this, note that condition~(i) implies that the demand
functions $m = g_m(x, \omega, \tau)$,
$e = g_e(x, \omega, \nu)$, and
$w = g_w(x, \omega, \eta)$ have the same structure regardless
of $D$. Hence, for any given $(x, \omega)$, the conditional
distribution of $(m, e, w)$ is unaffected by $D$.
Condition~(ii) implies that $D = d(\omega, x)$ for some function $d$
that does not depend on $(\tau, \nu, \eta)$. It follows that
conditional on $(x, m, e, w)$, the posterior distribution of
$\omega$ already incorporates all the information that $D$ could
provide about $\omega$. By the law of iterated expectations:
\begin{equation}
  \mathbb{E}[\hat{\omega}_{jt} \mid D_{jt}]
  = \mathbb{E}\bigl[\mathbb{E}[\omega_{jt} \mid x, m, e, w]
  \,\big|\, D_{jt}\bigr]
  = \mathbb{E}[\omega_{jt} \mid D_{jt}].
  \label{eq:omega_hat_D}
\end{equation}
\end{proof}

\subsection{Assumption~\ref{ass:homothetic}: Necessity and Testability Details}
\label{app:assumption6_details}

Each condition in Assumption~\ref{ass:homothetic} is necessary for
Theorem~\ref{thm:homothetic_id}. The failure mode differs by
condition.

\paragraph{Necessity of condition~(A).}
If $h' \equiv \gamma$ is constant, then $\bar{\omega}(k,l) = \gamma
v(k,l) + c_0$ is linear in $v$. For any $(c_k, c_l)$, define
$\tilde{v}(k,l) = v(k,l) - (c_k k + c_l l)/\gamma$; then
$\bar{\tilde{\omega}}(k,l) = \gamma \tilde{v}(k,l) + c_0$
preserves the structure. The shift is fully absorbed, and
$(\beta_k, \beta_l)$ remain unidentified.

\paragraph{Necessity of condition~(B).}
Under translation homogeneity, the alternative index
$\tilde{v}(k,l) = v(k,l) - (c_k k + c_l l)/\gamma$ (from the
necessity argument for~(A)) is also translation homogeneous only if
$c_k k + c_l l$ is translation homogeneous, which requires
$c_k + c_l = 1$; for general $(c_k, c_l)$ this need not hold.
The requirement that $\bar{\tilde{\omega}} = \tilde{h}(\tilde{v})$ for
some $\tilde{h}$ and translation homogeneous $\tilde{v}$ constrains
$(c_k, c_l)$ through the nonlinearity of $h$ (condition~(A)) and the
non-constancy of the MRS (condition~(C)). Without translation
homogeneity, $\tilde{v}$ can absorb the shift through
higher-order terms without contradicting the structural form.

\paragraph{Necessity of condition~(C).}
If $v_k/v_l$ is constant on $(k,l)$, then under translation
homogeneity $v(k,l) = \alpha k + (1-\alpha) l$ (the
Cobb--Douglas form in logs). The proof of
Theorem~\ref{thm:homothetic_id} requires
$c_l v_k - c_k v_l = 0$ everywhere, which is automatically satisfied
when $v_k/v_l = \alpha/(1-\alpha)$ for any $(c_k, c_l)$ satisfying
$c_l/c_k = \alpha/(1-\alpha)$. The one-dimensional manifold of solutions
$\{(c_k, c_l) : c_l/c_k = \alpha/(1-\alpha)\}$ represents a residual
indeterminacy that cannot be eliminated.

\paragraph{Economic interpretation.}
Condition~(A) requires that the cross-sectional relationship between
the capital-labor index and expected productivity is nonlinear:
identical absolute increases in the log index at different levels have
different effects on expected productivity. Condition~(B) corresponds
to constant returns to scale in the level variables (translation
homogeneity on the log scale is equivalent to degree-one homogeneity
in levels); if relaxed, the aggregator can absorb shifts that would
otherwise identify the parameters. Condition~(C) requires a finite
and non-unit elasticity of substitution between capital and labor:
firms with different capital-labor ratios face different marginal
rates of technical substitution, and this variation provides the
cross-sectional nonlinearity needed for identification.

\paragraph{Testability.}
Conditions~(A)--(C) concern $\bar{\omega}(k,l)$, which is a function
of the structural parameters estimated in Blocks~A and~B. Using
Block~A+B estimates, one can recover $\hat{\bar{\omega}}(k,l)$ up to
the $\Delta(k,l)$ shift. Condition~(C) can be assessed by testing
whether the estimated $\hat{\rho}_v$ differs significantly from zero.
In the CES specification, $\hat{\rho}_v$ is directly estimated by
grid search in Block~C (Section~\ref{sec:gmm_approach}); a
likelihood ratio or information criterion comparison between the CES
and Cobb--Douglas nested models tests~(C) directly. Condition~(A)
can be assessed by examining whether the relationship between the
estimated index and production residuals exhibits significant
nonlinearity, using a RESET-type test on the Block~C residuals.
Section~\ref{sec:diagnostics} describes the implementation.

\section{Estimation Details}
\label{app:estimation_details}

This appendix collects technical details of the GMM estimation
procedure that supplement Section~\ref{sec:estimation}.

\subsection{Block~A: Instrument Assignment and Invariance}
\label{app:block_a_details}

\paragraph{Instrument assignment logic.}
Since $u_{1,jt}$ consists only of
$\tau_{jt}$ and $\nu_{jt}$, it is uncorrelated with $\eta_{jt}$
(Assumption~\ref{ass:gmm_uncorrelated}(3)) and hence with
$w_{jt}$ (which contains $\eta_{jt}$ as the sole unobserved
component uncorrelated with $Z_{\mathrm{base}}$). Thus $w_{jt}$
serves as an additional instrument for $u_{1,jt}$. The reasoning
for $u_{2,jt}$ and $u_{3,jt}$ is analogous: $u_{2,jt}$ contains
$\tau_{jt}$ and $\eta_{jt}$ but not $\nu_{jt}$, so $e_{jt}$ is a
valid instrument; $u_{3,jt}$ contains $\varepsilon_{jt}$ and
$\tau_{jt}$ but not $\nu_{jt}$ or $\eta_{jt}$, so both $e_{jt}$
and $w_{jt}$ are valid instruments.

\paragraph{Invariance to $\Delta(k,l)$.}
Block~A is invariant to the observationally equivalent
transformation
$\tilde{\omega} = \omega + c_k k + c_l l$
(Theorem~\ref{thm:obs_equiv}). Under this relabeling,
$\tilde{\beta}_k = \beta_k + c_k$ and
$\tilde{\gamma}_k = \gamma_k + \gamma_\omega c_k$ (and analogously
for $l$ and for $e, w$). All residuals
$\tilde{m}_{jt}$, $\tilde{e}_{jt}$, $\tilde{w}_{jt}$, and
$\tilde{y}_{jt}$ are individually invariant to this transformation:
for instance,
$\tilde{y}_{jt} = \tilde{\beta}_k k + \tilde{\beta}_l l
+ \tilde{\omega} + \varepsilon = \beta_k k + \beta_l l + \omega
+ \varepsilon$. Hence $u_{i,jt}$ ($i = 1,2,3$) are
invariant, and Block~A cannot separately identify $\beta_k$ and the
demand slopes on $(k, l)$.

\subsection{Block~B: Covariance Derivation}
\label{app:block_b_derivation}

Using the mutual exogeneity of shocks
(Assumption~\ref{ass:gmm_uncorrelated}(3)), the cross-covariances
among residuals satisfy:
\begin{align}
  \mathrm{Cov}(\tilde{m}, \tilde{e})
  &= \gamma_\omega\,\delta_\omega\,\mathrm{Var}(\omega),
  \label{eq:covB1_app} \\
  \mathrm{Cov}(\tilde{m}, \tilde{w})
  &= \gamma_\omega\,\zeta_\omega\,\mathrm{Var}(\omega),
  \label{eq:covB2_app} \\
  \mathrm{Cov}(\tilde{e}, \tilde{w})
  &= \delta_\omega\,\zeta_\omega\,\mathrm{Var}(\omega),
  \label{eq:covB3_app} \\
  \mathrm{Cov}(\tilde{y}, \tilde{m})
  &= \gamma_\omega\,\mathrm{Var}(\omega),
  \label{eq:covB4_app} \\
  \mathrm{Cov}(\tilde{y}, \tilde{e})
  &= \delta_\omega\,\mathrm{Var}(\omega),
  \label{eq:covB5_app} \\
  \mathrm{Cov}(\tilde{y}, \tilde{w})
  &= \zeta_\omega\,\mathrm{Var}(\omega).
  \label{eq:covB6_app}
\end{align}
Eliminating $\mathrm{Var}(\omega)$ across pairs yields six moment
conditions. For instance, dividing~\eqref{eq:covB1_app}
by~\eqref{eq:covB4_app} gives
$\mathrm{Cov}(\tilde{m}, \tilde{e})
= \delta_\omega \mathrm{Cov}(\tilde{y}, \tilde{m})$. Expressing
all six relations in expectation form (using de-meaned data):
\begin{align}
  \mathbb{E}\bigl[\tilde{m}\,\tilde{e}
  - \delta_\omega\,\tilde{y}\,\tilde{m}\bigr] &= 0,
  \label{eq:momentB1_app} \\
  \mathbb{E}\bigl[\tilde{m}\,\tilde{e}
  - \gamma_\omega\,\tilde{y}\,\tilde{e}\bigr] &= 0,
  \label{eq:momentB2_app} \\
  \mathbb{E}\bigl[\tilde{m}\,\tilde{w}
  - \zeta_\omega\,\tilde{y}\,\tilde{m}\bigr] &= 0,
  \label{eq:momentB3_app} \\
  \mathbb{E}\bigl[\tilde{m}\,\tilde{w}
  - \gamma_\omega\,\tilde{y}\,\tilde{w}\bigr] &= 0,
  \label{eq:momentB4_app} \\
  \mathbb{E}\bigl[\tilde{e}\,\tilde{w}
  - \zeta_\omega\,\tilde{y}\,\tilde{e}\bigr] &= 0,
  \label{eq:momentB5_app} \\
  \mathbb{E}\bigl[\tilde{e}\,\tilde{w}
  - \delta_\omega\,\tilde{y}\,\tilde{w}\bigr] &= 0.
  \label{eq:momentB6_app}
\end{align}

\paragraph{Redundancy with Block~A.}
Of the six moment conditions
\eqref{eq:momentB1_app}--\eqref{eq:momentB6_app}, four are
algebraically implied by the Block~A instrumental variable moments.
Specifically, any four of the six conditions that involve
cross-products of the demand residuals $\tilde{e}$ or $\tilde{w}$
with $\tilde{m}$ (or equivalently with $\tilde{y}$ via $u_3$)
are already encoded in the Block~A moment conditions
through the instruments $Z_3 = (k, l, \tilde{e}, \tilde{w})$
(equation~\eqref{eq:momentA3}); which four are labeled ``redundant''
depends on the chosen basis, but the rank reduction by four is
basis-independent.
The two independent contributions (in any basis) correspond to
cross-covariance ratios not captured by Block~A instruments.
Consequently, the combined
Block~A+B system has $\mathrm{rank}(\Omega) = 12$, matching the
12 free parameters, and is just-identified. The concentrated
covariance-ratio formulas remain useful for obtaining closed-form
scale parameter estimates, improving computational efficiency.

\paragraph{Invariance to $\Delta(k,l)$.}
As with Block~A, Block~B is invariant to the $\Delta(k,l)$
transformation, since each residual $\tilde{m}_{jt}$,
$\tilde{e}_{jt}$, $\tilde{w}_{jt}$, and $\tilde{y}_{jt}$ is
individually invariant to the relabeling
$\tilde{\omega} = \omega + c_k k + c_l l$
(Appendix~\ref{app:block_a_details}).

\subsection{De-Meaning, Intercept Recovery, and Two-Step Procedure}
\label{app:demeaning_details}

\paragraph{De-meaning.}
Non-zero demand intercepts $(\gamma_0, \delta_0, \zeta_0)$ cause the
raw-level moment conditions to be misspecified. To see this, note
that
$\mathbb{E}[u_{1,jt}]
= \delta_\omega \gamma_0 - \gamma_\omega \delta_0$, which is
generally nonzero. All variables are therefore de-meaned prior to
estimation and the constant is excluded from all instrument vectors.

\paragraph{Post-estimation intercepts.}
After obtaining the GMM estimates $\hat{\Theta}$, the production
function intercept is recovered as:
\begin{equation}
  \hat{\beta}_0^{\mathrm{new}}
  = \bar{y} - \hat{\beta}_k\,\bar{k} - \hat{\beta}_l\,\bar{l}
  - \hat{\beta}_m\,\bar{m} - \hat{\beta}_{e}\,\bar{e}
  - \hat{\beta}_{w}\,\bar{w},
  \label{eq:intercept_recovery_app}
\end{equation}
where overbars denote sample means of the original (non-de-meaned)
data. This formula follows from the normalization
$\mathbb{E}[\omega] = 0$, which implies
$\mathbb{E}[h(v)] = \mathbb{E}[\mathbb{E}[\omega \mid k, l]]
= \mathbb{E}[\omega] = 0$ by the law of iterated expectations.
Hence the function $h$ does not appear in the intercept. The
intercept absorbs the demand function intercepts and the constant
$\rho_0$.

\paragraph{Stacked GMM objective.}
Define the integrated moment vector by stacking all three blocks:
\[
  g_{jt}(\Theta) = \bigl[
  g_{jt,A}(\Theta)',\;
  g_{jt,B}(\Theta)',\;
  g_{jt,C}(\Theta)'
  \bigr]',
\]
where $g_{jt,A}$ collects the Block~A moments, $g_{jt,B}$ collects
the Block~B covariance moments, and $g_{jt,C}$
collects the Block~C structural moments. All
parameters $\Theta = (\theta_1, \theta_2)$ are estimated
simultaneously by minimizing:
\begin{equation}
  \hat{\Theta} = \arg\min_\Theta\;
  g_N(\Theta)'\,\hat{W}\,g_N(\Theta),
  \qquad
  g_N(\Theta) = \frac{1}{N}\sum_{j=1}^N
  \bar{g}_j(\Theta),
  \label{eq:gmm_objective_app}
\end{equation}
where $\bar{g}_j(\Theta) = T^{-1}\sum_{t=1}^T g_{jt}(\Theta)$
is the time-averaged moment for firm $j$.

\paragraph{Two-step procedure.}
\begin{enumerate}
\item \textbf{Step~1:} Minimize~\eqref{eq:gmm_objective_app} with an
  initial weighting matrix $\hat{W}^{(0)}$ (e.g., a block-diagonal
  matrix that normalizes the scale of each block) to obtain
  $\hat{\Theta}^{(1)}$.
\item \textbf{Optimal weight:} Estimate the long-run covariance
  matrix
  $\hat{\Sigma} = \frac{1}{N}\sum_{j=1}^N
  \bar{g}_j(\hat{\Theta}^{(1)})\,
  \bar{g}_j(\hat{\Theta}^{(1)})'$ and set
  $\hat{W}_{\mathrm{opt}} = \hat{\Sigma}^{-1}$.
\item \textbf{Step~2:} Re-minimize~\eqref{eq:gmm_objective_app} with
  $\hat{W}_{\mathrm{opt}}$ to obtain the efficient estimator
  $\hat{\Theta}^{(2)}$.
\item \textbf{Intercepts:} Recover
  $\hat{\beta}_0^{\mathrm{new}}$
  via~\eqref{eq:intercept_recovery_app}.
\end{enumerate}

\subsection{Computational Details}
\label{app:computation}

\paragraph{Computational cost.}
The two-step GMM requires numerical optimization over $\Theta$
($\dim \Theta = 18 + 3\,d_z$ when control variables with a
polynomial basis of dimension $d_z$ are included;
$d_z = 0$ in the baseline specification). Block~C
introduces nonlinearity through the CES index~\eqref{eq:index_v},
and $(\rho_v, \alpha)$ are optimized by profile GMM over a
discrete grid. Each grid point requires a standard GMM optimization
over the remaining parameters, making the total cost approximately
$|\text{grid}| \times$ cost of a single GMM evaluation. In the
empirical application with $N \approx 5{,}000$ and $T \approx 20$,
a single industry estimation completes in under one minute on a
12-core workstation. Bootstrap standard errors (200 replications)
require proportionally more time.

\paragraph{Iterative profile estimation of scale parameters.}
The scale parameters $(\gamma_\omega, \delta_\omega, \zeta_\omega)$
and the slope parameters $(\beta_m, \beta_e, \beta_w, \theta_g)$
enter the moment conditions multiplicatively, creating a ridge in
the GMM objective surface. I employ an iterative profile strategy:
given current scale values, the slope parameters are estimated by
minimizing the GMM objective; the scale parameters are then updated
via closed-form covariance ratios derived from the Block~B
conditions~\eqref{eq:momentB_compact}; and the procedure iterates
until convergence. A final joint optimization step refines all
parameters simultaneously, using the iterative profile solution as
starting values.

\subsection{Regularity Conditions and Asymptotic Proof}
\label{app:asymptotic_proof}

\begin{assumption}[Standard Conditions for Asymptotics]
\label{ass:asymptotics_app}
\begin{enumerate}
\item The sample
  $\{(y_{jt}, k_{jt}, l_{jt}, m_{jt}, e_{jt},
  w_{jt})_{t=1}^T\}_{j=1}^N$ consists of $N$ independent draws
  (independence across firms), with $T$ fixed.
\item The weighting matrix $\hat{W}$ converges in probability to a
  positive definite matrix $W$
  ($\hat{W} \xrightarrow{p} W$).
\item The true parameter vector $\Theta_0$ lies in the interior of a
  compact parameter space.
\item Identification Condition:
  $\mathbb{E}[\bar{g}_j(\Theta)] = 0$ if and only if
  $\Theta = \Theta_0$.
\item The variables necessary to compute the moment function
  $g_{jt}(\Theta)$ have finite moments of sufficiently high order.
\item $g_{jt}(\Theta)$ is continuously differentiable in $\Theta$
  in a neighborhood of $\Theta_0$, and the expected Jacobian matrix
  $G \equiv \mathbb{E}[\nabla_\Theta \bar{g}_j(\Theta_0)]$ has
  full column rank.
\end{enumerate}
\end{assumption}

\begin{proof}[Proof of Theorem~\ref{thm:asymptotics}]
The result follows from Theorems~2.6 and~3.4 in
\textcite{newey1994chapter}, with
Assumption~\ref{ass:asymptotics_app}(1) ensuring the applicability of
cross-sectional LLN and CLT. The time-averaged moment
$\bar{g}_j(\Theta) = T^{-1}\sum_{t=1}^T g_{jt}(\Theta)$ treats
each firm's $T$-period panel as a single observation, so the
asymptotic framework is cross-sectional ($N \to \infty$, $T$ fixed).
The optimal weighting matrix $W = \Sigma^{-1}$ yields the efficient
two-step GMM estimator with asymptotic variance
$(G'\Sigma^{-1}G)^{-1}$.

Standard errors are computed from a consistent estimate $\hat{V}$
of the asymptotic variance~\eqref{eq:asymptotic_variance}, using
the second-step estimates $\hat{\Theta}^{(2)}$ to evaluate the
sample Jacobian $\hat{G}$ and the moment covariance
$\hat{\Sigma}$. The standard error for the post-estimation
intercept $\hat{\beta}_0^{\mathrm{new}}$ is computed via the delta
method. In practice, standard errors are clustered at the firm level.
Although the time-averaged moment $\bar{g}_j$ aggregates across
periods, within-firm serial dependence can still inflate the variance
of $\bar{g}_j$ relative to the i.i.d.\ case. Clustering at the firm
level provides a heteroskedasticity- and autocorrelation-consistent
estimate of $\Sigma$ that accommodates arbitrary within-firm temporal
dependence, analogous to cluster-robust variance estimation in panel
regressions.
\end{proof}


\section{Direction of Bias under Conditional Independence Violation}
\label{app:ci_violation}

This appendix derives the direction of bias in $\hat{\beta}_m$
when the conditional independence assumption
(Assumption~\ref{ass:cond_indep}) is violated through a positive
covariance between the electricity and water demand shocks.

\subsection{Setup}

After the Frisch--Waugh--Lovell projection, the residual
structure takes the form:
\begin{align}
  \tilde{y}_{jt} &= \omega_{jt} + \varepsilon_{jt}, \label{eq:ci_y} \\
  \tilde{m}_{jt} &= \gamma_\omega\, \omega_{jt} + \tau_{jt}, \label{eq:ci_m} \\
  \tilde{e}_{jt} &= \delta_\omega\, \omega_{jt} + \nu_{jt}, \label{eq:ci_e} \\
  \tilde{w}_{jt} &= \zeta_\omega\, \omega_{jt} + \eta_{jt}, \label{eq:ci_w}
\end{align}
where $(\tau_{jt}, \nu_{jt}, \eta_{jt})$ are the input-specific
demand shocks. Under Assumption~\ref{ass:cond_indep}, all pairwise
covariances among these shocks are zero.

\subsection{CI Violation: Electricity--Water Common Utility Shock}

Suppose:
\begin{equation}
  \sigma_{\nu\eta} \equiv \operatorname{Cov}(\nu_{jt}, \eta_{jt}) > 0,
  \label{eq:ci_violation}
\end{equation}
while $\operatorname{Cov}(\tau, \nu) = \operatorname{Cov}(\tau, \eta) = 0$.
This arises naturally when a common energy price shock or seasonal
supply constraint raises both electricity and water costs
simultaneously---the most economically salient threat to conditional
independence, since electricity and water are both utility services
subject to common regulatory and infrastructure conditions.
The materials demand shock $\tau_{jt}$, which reflects raw material
procurement through distinct supply chains, remains independent.

\subsection{Bias in Scale Parameters}

The concentrated scale estimator for $\zeta_\omega$ uses the
cross-covariance between the electricity and water residuals.
From~\eqref{eq:ci_e} and~\eqref{eq:ci_w}:
\begin{equation}
  \mathbb{E}[\tilde{e} \cdot \tilde{w}]
  = \delta_\omega \zeta_\omega \operatorname{Var}(\omega)
    + \sigma_{\nu\eta}.
  \label{eq:cross_moment}
\end{equation}
The normalizing moment
$\mathbb{E}[\tilde{y} \cdot \tilde{e}]
= \delta_\omega \operatorname{Var}(\omega)$
is unaffected by $\sigma_{\nu\eta}$. The ratio gives:
\begin{equation}
  \hat{\zeta}_\omega \xrightarrow{p}\;
  \zeta_\omega + \frac{\sigma_{\nu\eta}}
  {\delta_\omega\, \operatorname{Var}(\omega)}.
  \label{eq:zeta_bias}
\end{equation}
The bias is positive when $\sigma_{\nu\eta} > 0$:
$\hat{\zeta}_\omega$ overestimates the true scale parameter.
Since $\operatorname{Cov}(\tau, \nu)
= \operatorname{Cov}(\tau, \eta) = 0$, the other two scale
parameters $\hat{\gamma}_\omega$ and $\hat{\delta}_\omega$
remain consistently estimated.

\subsection{Propagation to $\hat{\beta}_m$: Upward Bias}

The overestimation of $\zeta_\omega$ propagates to
$\hat{\beta}_m$ through the Block~A moment conditions.
The Block~A error $u_{2,jt} = \zeta_\omega \tilde{m}_{jt}
- \gamma_\omega \tilde{w}_{jt}$ eliminates $\omega_{jt}$ when
the scale parameters are correctly specified. When
$\hat{\zeta}_\omega > \zeta_\omega$, a positive fraction of
$\omega_{jt}$ leaks into $\hat{u}_{2,jt}$:
\[
  \hat{u}_{2,jt}
  = (\zeta_\omega + b)\,\tilde{m}_{jt}
    - \gamma_\omega\,\tilde{w}_{jt}
  = u_{2,jt} + b\,\tilde{m}_{jt},
\]
where $b > 0$ is the bias in~\eqref{eq:zeta_bias} and
$\tilde{m}_{jt} = \gamma_\omega \omega_{jt} + \tau_{jt}$ is
positively correlated with productivity.
This contamination biases the moment conditions for the
production function coefficients. In particular, the GMM
estimator compensates for the positive $\omega$ leakage in
the $u_2$-based moments by \emph{increasing} $\hat{\beta}_m$,
producing an \textbf{upward bias}.

Monte Carlo simulations (Section~\ref{sec:simulation},
Table~\ref{tab:mc_dgp4_ci}) confirm this direction:
$\hat{\beta}_m$ increases monotonically with
$\mathrm{Corr}(\nu, \eta)$.

\subsection{Implications}

\begin{enumerate}
  \item The bias is \emph{upward}: if CI is violated through a
    common electricity--water utility shock, the proposed estimator
    \emph{overestimates} $\beta_m$ and implied markups.
  \item This bias direction is the \emph{same} as the Markov
    misspecification bias in ACF-type estimators (which also
    overestimates $\beta_m$ under DGPs~2 and~3). Therefore, the
    empirical finding that the proposed estimator yields
    \emph{lower} $\hat{\beta}_m$ than ACF cannot be attributed
    to CI violation; it must reflect Markov misspecification
    bias in ACF.
  \item Including additional control variables in $z_{jt}$
    (e.g., regional energy price indices, seasonal indicators)
    reduces $\sigma_{\nu\eta}$ by absorbing common sources of
    utility cost variation, providing a partial remedy.
\end{enumerate}


\section{Parametric Implementation under Flexible Functional Forms}
\label{app:translog}

This appendix extends the parametric GMM implementation of
Section~\ref{sec:estimation} to flexible functional forms. The
identification source throughout is the conditional independence of
demand shocks (Assumption~\ref{ass:cond_indep})---the parametric
counterpart of the HS08 spectral decomposition
(Theorems~\ref{thm:density_id}--\ref{thm:obs_equiv}). Under
Cobb--Douglas, conditional independence yields the linear covariance
structure of Blocks~A and~B. Under translog, the same condition yields
nonlinear moment conditions derived from the structure of input demand
functions.

\subsection{Translog Production Function}
\label{app:translog_setup}

Consider the translog production function:
\begin{align}
y_{jt} &= \beta_k k_{jt} + \beta_l l_{jt} + \beta_m m_{jt}
        + \beta_e e_{jt} + \beta_w w_{jt} \notag \\
       &\quad + \beta_{kk} k_{jt}^2 + \beta_{ll} l_{jt}^2
        + \beta_{mm} m_{jt}^2 + \beta_{ee} e_{jt}^2
        + \beta_{ww} w_{jt}^2 \notag \\
       &\quad + \beta_{kl} k_{jt} l_{jt} + \beta_{km} k_{jt} m_{jt}
        + \beta_{ke} k_{jt} e_{jt} + \beta_{kw} k_{jt} w_{jt} \notag \\
       &\quad + \beta_{lm} l_{jt} m_{jt} + \beta_{le} l_{jt} e_{jt}
        + \beta_{lw} l_{jt} w_{jt} \notag \\
       &\quad + \beta_{me} m_{jt} e_{jt} + \beta_{mw} m_{jt} w_{jt}
        + \beta_{ew} e_{jt} w_{jt} + \omega_{jt} + \varepsilon_{jt}.
\label{eq:translog_prod}
\end{align}
The log marginal products are:
\begin{align}
\frac{\partial f}{\partial m} &= \beta_m + 2\beta_{mm} m + \beta_{km} k
  + \beta_{lm} l + \beta_{me} e + \beta_{mw} w,
\label{eq:mp_m} \\
\frac{\partial f}{\partial e} &= \beta_e + 2\beta_{ee} e + \beta_{ke} k
  + \beta_{le} l + \beta_{me} m + \beta_{ew} w,
\label{eq:mp_e} \\
\frac{\partial f}{\partial w} &= \beta_w + 2\beta_{ww} w + \beta_{kw} k
  + \beta_{lw} l + \beta_{mw} m + \beta_{ew} e.
\label{eq:mp_w}
\end{align}

\subsection{Demand Structure under Translog}
\label{app:translog_demand}

From the first-order condition for cost minimization
(Appendix~\ref{sec:micro_fnd}), each intermediate input $h \in \{m, e, w\}$
satisfies:
\begin{equation}
f + \omega - h + \ln\left(\frac{\partial f}{\partial h}\right)
= \phi_h(z_{jt}) + \tau_h,
\label{eq:foc_translog}
\end{equation}
where $\phi_h(z)$ captures price and markdown terms absorbed by control
variables, and $\tau_h$ is the input-specific demand shock. Under
translog, the log marginal product $\ln(\partial f / \partial h)$ depends
on the levels of all inputs, so input demands are implicitly defined and
nonlinear in productivity.

\subsection{Moment Conditions from Conditional Independence}
\label{app:translog_moments}

The key observation is that subtracting the first-order conditions for
two inputs eliminates both $f$ and $\omega$. For inputs $m$ and $e$:
\begin{equation}
\ln\left(\frac{\partial f/\partial m}{\partial f/\partial e}\right)
- (m - e) = (\phi_m - \phi_e) + (\tau_m - \nu_e).
\label{eq:foc_diff_me}
\end{equation}
The production function $f$ and productivity $\omega$ cancel exactly.
De-meaning removes the price terms $\phi_m - \phi_e$, yielding:
\begin{equation}
\widetilde{\ln\left(\frac{\partial f/\partial m}{\partial f/\partial e}\right)}
- (\tilde{m} - \tilde{e}) = \tilde{\tau}_m - \tilde{\nu}_e,
\label{eq:foc_diff_me_dm}
\end{equation}
where tildes denote de-meaned variables.

By Assumption~\ref{ass:cond_indep}, $\tau_m$, $\nu_e$, and $\eta_w$ are
mutually independent conditional on $(\omega, k, l, z)$. Therefore,
$\eta_w$ is uncorrelated with $\tau_m - \nu_e$, and $w$ serves as a
valid instrument. Defining:
\begin{equation}
u_{me}^{TL} \equiv \ln\left(
\frac{\beta_m + 2\beta_{mm} m + \beta_{km} k + \beta_{lm} l
      + \beta_{me} e + \beta_{mw} w}
     {\beta_e + 2\beta_{ee} e + \beta_{ke} k + \beta_{le} l
      + \beta_{me} m + \beta_{ew} w}
\right) - (m - e) - \text{mean},
\label{eq:u_me_tl}
\end{equation}
the moment condition is:
\begin{equation}
\mathbb{E}\bigl[u_{me}^{TL} \cdot (k,\, l,\, w,\, z)\bigr] = 0.
\label{eq:moment_me_tl}
\end{equation}
Analogous conditions hold for the pairs $(m, w)$ and $(e, w)$:
\begin{align}
\mathbb{E}\bigl[u_{mw}^{TL} \cdot (k,\, l,\, e,\, z)\bigr] &= 0,
\label{eq:moment_mw_tl} \\
\mathbb{E}\bigl[u_{ew}^{TL} \cdot (k,\, l,\, m,\, z)\bigr] &= 0.
\label{eq:moment_ew_tl}
\end{align}
These three sets of nonlinear moment conditions identify the
intermediate input parameters
$(\beta_m, \beta_e, \beta_w, \beta_{mm}, \beta_{ee}, \beta_{ww},
\beta_{me}, \beta_{mw}, \beta_{ew})$.

\subsection{Identification of Primary Input Parameters}
\label{app:translog_kl}

The $\Delta(k,l)$ indeterminacy of Theorem~\ref{thm:obs_equiv} persists
under translog: the moment conditions~\eqref{eq:moment_me_tl}--\eqref{eq:moment_ew_tl}
do not identify parameters involving $(k, l)$, namely
$(\beta_k, \beta_l, \beta_{kk}, \beta_{ll}, \beta_{kl}, \beta_{km},
\beta_{ke}, \beta_{kw}, \beta_{lm}, \beta_{le}, \beta_{lw})$.

Corollary~\ref{thm:exclusion} (exclusion restrictions) extends directly
to translog. Condition~(i)---that some input demand is independent of
$(k, l)$ conditional on productivity---implies:
\begin{equation}
\beta_{kw} = \beta_{lw} = 0.
\label{eq:excl_translog}
\end{equation}
Under this restriction, the log marginal product of $w$
(equation~\ref{eq:mp_w}) does not depend on $(k, l)$:
\begin{equation}
\frac{\partial f}{\partial w} = \beta_w + 2\beta_{ww} w
+ \beta_{mw} m + \beta_{ew} e.
\label{eq:mp_w_excl}
\end{equation}
The productivity proxy constructed from input $w$ is then independent
of $(k, l)$, and the $(k, l)$ parameters can be recovered by the
following procedure.

Define the partially residualized output using the intermediate input
estimates from Section~\ref{app:translog_moments}:
\begin{align}
\tilde{y}^{TL}_{jt} &\equiv y_{jt}
  - \hat{\beta}_m m - \hat{\beta}_e e - \hat{\beta}_w w \notag \\
&\quad - \hat{\beta}_{mm} m^2 - \hat{\beta}_{ee} e^2 - \hat{\beta}_{ww} w^2
  - \hat{\beta}_{me} me - \hat{\beta}_{mw} mw - \hat{\beta}_{ew} ew.
\label{eq:y_tilde_tl}
\end{align}
This equals:
\begin{align}
\tilde{y}^{TL}_{jt} &= \beta_k k + \beta_l l
  + \beta_{kk} k^2 + \beta_{ll} l^2 + \beta_{kl} kl \notag \\
&\quad + \beta_{km} km + \beta_{ke} ke + \beta_{lm} lm + \beta_{le} le
  + \omega + \varepsilon.
\label{eq:y_tilde_tl_expanded}
\end{align}
Construct the productivity proxy from input $w$:
\begin{equation}
\hat{\omega}^w_{jt} \equiv w_{jt} - f_{jt}
- \ln\left(\frac{\partial f}{\partial w}\right)_{jt} - \hat{\phi}_w(z_{jt}),
\label{eq:omega_hat_w_tl}
\end{equation}
where all terms on the right-hand side are evaluated at the estimated
parameters. Under the exclusion restriction~\eqref{eq:excl_translog},
$\hat{\omega}^w$ does not depend on $(k, l)$.

The moment condition for the $(k, l)$ parameters is:
\begin{equation}
\mathbb{E}\Bigl[\bigl(\tilde{y}^{TL} - \hat{\omega}^w
- g_{kl}(k, l, m, e; \beta_{kl})\bigr) \cdot Z_{kl}\Bigr] = 0,
\label{eq:moment_kl_tl}
\end{equation}
where $g_{kl}$ collects all terms involving $(k, l)$:
\begin{equation}
g_{kl} \equiv \beta_k k + \beta_l l + \beta_{kk} k^2 + \beta_{ll} l^2
+ \beta_{kl} kl + \beta_{km} km + \beta_{ke} ke + \beta_{lm} lm + \beta_{le} le,
\label{eq:g_kl}
\end{equation}
and $Z_{kl} = (k, l, k^2, l^2, kl, km, ke, lm, le)$. The error term
$\varepsilon - \eta_w / \zeta_\omega$ is orthogonal to $Z_{kl}$ under
the exclusion restriction, identifying all $(k, l)$ parameters.

\subsection{Reduction to Cobb--Douglas}
\label{app:translog_cd}

Setting all second-order coefficients to zero
($\beta_{hh'} = 0$ for all $h, h'$), the translog reduces to
Cobb--Douglas. The log marginal product ratio in
equation~\eqref{eq:u_me_tl} becomes:
\begin{equation}
\ln\left(\frac{\partial f/\partial m}{\partial f/\partial e}\right)
= \ln\frac{\beta_m}{\beta_e},
\label{eq:mp_ratio_cd}
\end{equation}
a constant that vanishes under de-meaning. The residual
$u_{me}^{TL}$ reduces to $-(\tilde{m} - \tilde{e})$, and the nonlinear
moment conditions collapse to linear orthogonality conditions.

Similarly, the $(k, l)$ identification
(equation~\ref{eq:moment_kl_tl}) reduces to:
\begin{equation}
\tilde{y} - \hat{\omega}^w = \beta_k k + \beta_l l + \text{error},
\label{eq:ols_cd}
\end{equation}
which is the OLS regression of Remark~\ref{rem:testable_exclusion}.
The Cobb--Douglas implementation of Section~\ref{sec:estimation} is
thus a computationally tractable special case of the general framework.

\section{Empirical Estimation Flowchart}
\label{app:flowchart}

Figure~\ref{fig:flowchart} provides an overview of the full
empirical estimation and inference pipeline.

\begin{figure}[htbp]
\centering
\resizebox{\textwidth}{!}{%
\begin{tikzpicture}[
  node distance = 0.5cm,
  box/.style    = {rectangle, draw, rounded corners=3pt,
                   text width=5.8cm, align=center,
                   minimum height=0.9cm, font=\small},
  diag/.style   = {diamond, draw, aspect=2.8,
                   text width=3.8cm, align=center, font=\small},
  side/.style   = {rectangle, draw, rounded corners=3pt,
                   text width=4.6cm, align=left,
                   minimum height=0.8cm, font=\small, inner sep=5pt},
  avail/.style  = {rectangle, draw, rounded corners=3pt,
                   text width=5.2cm, align=left,
                   font=\small, inner sep=5pt, fill=blue!8},
  head/.style   = {rectangle, draw, rounded corners=3pt,
                   text width=5.8cm, align=center,
                   minimum height=0.9cm, font=\small\bfseries,
                   fill=gray!20},
  arr/.style    = {-Stealth, thick},
  darr/.style   = {-Stealth, thick, dashed}
]


\node[head] (input)
  {Step 1: Data requirements\\
   {\normalfont\small Firm-level panel $(y,k,l,m,e,w)_{jt}$;\\
   at least two intermediate input proxies}};

\node[box, below=0.6cm of input] (ab)
  {Step 2: Block~A+B GMM\\
   {\footnotesize $\hat\beta_m,\hat\beta_e,\hat\beta_w$; demand slopes;\\
   productivity loading parameters.\\
   \textbf{No assumption on productivity dynamics.}}};

\node[box, below=0.6cm of ab] (omega)
  {Step 3: Productivity index $\hat\omega_{jt}$\\
   {\footnotesize Identified up to $\Delta(k,l)$ indeterminacy\\
   (Theorem~\ref{thm:obs_equiv})}};

\node[diag, below=0.6cm of omega] (excl)
  {Diagnostic~(i): Exclusion test\\$d_k\!\neq\!0$ or $d_l\!\neq\!0$?\\
   {\footnotesize ($d_k,d_l$: OLS slopes of $\hat\omega$ on $k,l$)}};

\node[box, below=0.6cm of excl] (blockc)
  {Step 4: Block~C GMM\\
   {\footnotesize CES curvature $\Rightarrow$ $\hat\beta_k, \hat\beta_l$}};

\node[diag, below=0.6cm of blockc] (ces)
  {Diagnostic~(ii)\\CES curvature: $\hat\rho_v \neq 0$?};

\node[side, below=0.6cm of ces, fill=green!10] (klfinal)
  {\textit{Pass:} $\hat\beta_k, \hat\beta_l$ from Block~C identified.};


\node[avail, right=2.8cm of omega] (down3)
  {\textbf{Available after Step~3:}\\
   $\bullet$ Markup estimation\\
   $\bullet$ Productivity determinants\\
   $\bullet$ Event study (policy evaluation)\\[2pt]
   {\footnotesize All invariant to $\Delta(k,l)$
   (Thm.~\ref{thm:obs_equiv})}};

\node[side, right=2.8cm of excl, fill=green!10] (cons)
  {\textit{No ($d_k\!=\!d_l\!=\!0$):}\\
   OLS recovery (Prop.~\ref{prop:excl_ols})\\
   gives $\hat\beta_k, \hat\beta_l$ directly.\\
   Block~C as cross-check only.};

\node[side, right=2.8cm of ces, fill=orange!15] (cesfail)
  {\textit{Fail:} Cobb-Douglas ($\rho_v\!=\!0$).\\
   $\beta_k,\beta_l$ not identified from\\
   Block~C (Thm.~\ref{thm:homothetic_id}).\\
   {\footnotesize If Diag.~(ii) passed, use\\
   OLS recovery instead.}};

\draw[arr] (input)  -- (ab);
\draw[arr] (ab)     -- (omega);
\draw[arr] (omega)  -- (excl);
\draw[arr] (excl)   -- node[left,font=\footnotesize]{yes} (blockc);
\draw[arr] (blockc) -- (ces);
\draw[arr] (ces)    -- node[left,font=\footnotesize]{pass} (klfinal);

\draw[arr] (omega.east)  -- (down3.west);
\draw[arr] (excl)   -- node[above,font=\footnotesize]{no}  (cons);
\draw[arr] (ces)    -- node[above,font=\footnotesize]{fail}  (cesfail);

\end{tikzpicture}}%
\caption{Implementation Guide: Proposed Estimation Pipeline}
\label{fig:flowchart}

\small\raggedright \textit{Notes: Step-by-step guide for applying
the proposed estimator to firm-level panel data.
The key message is that \textbf{Step~3 alone} (Block~A+B GMM) suffices
for the most common downstream applications---markup estimation,
productivity analysis, event studies, and Olley--Pakes
decomposition---because these rely only on $\hat\beta_m$ and
$\hat\omega_{jt}$, which are invariant to the $\Delta(k,l)$
indeterminacy (Theorem~\ref{thm:obs_equiv}).
Step~4 (Block~C GMM) is needed only when capital and labor elasticities
$(\beta_k, \beta_l)$ are themselves of interest.
Diagnostics (i)--(iii) check the maintained assumptions at
each stage.}
\end{figure}
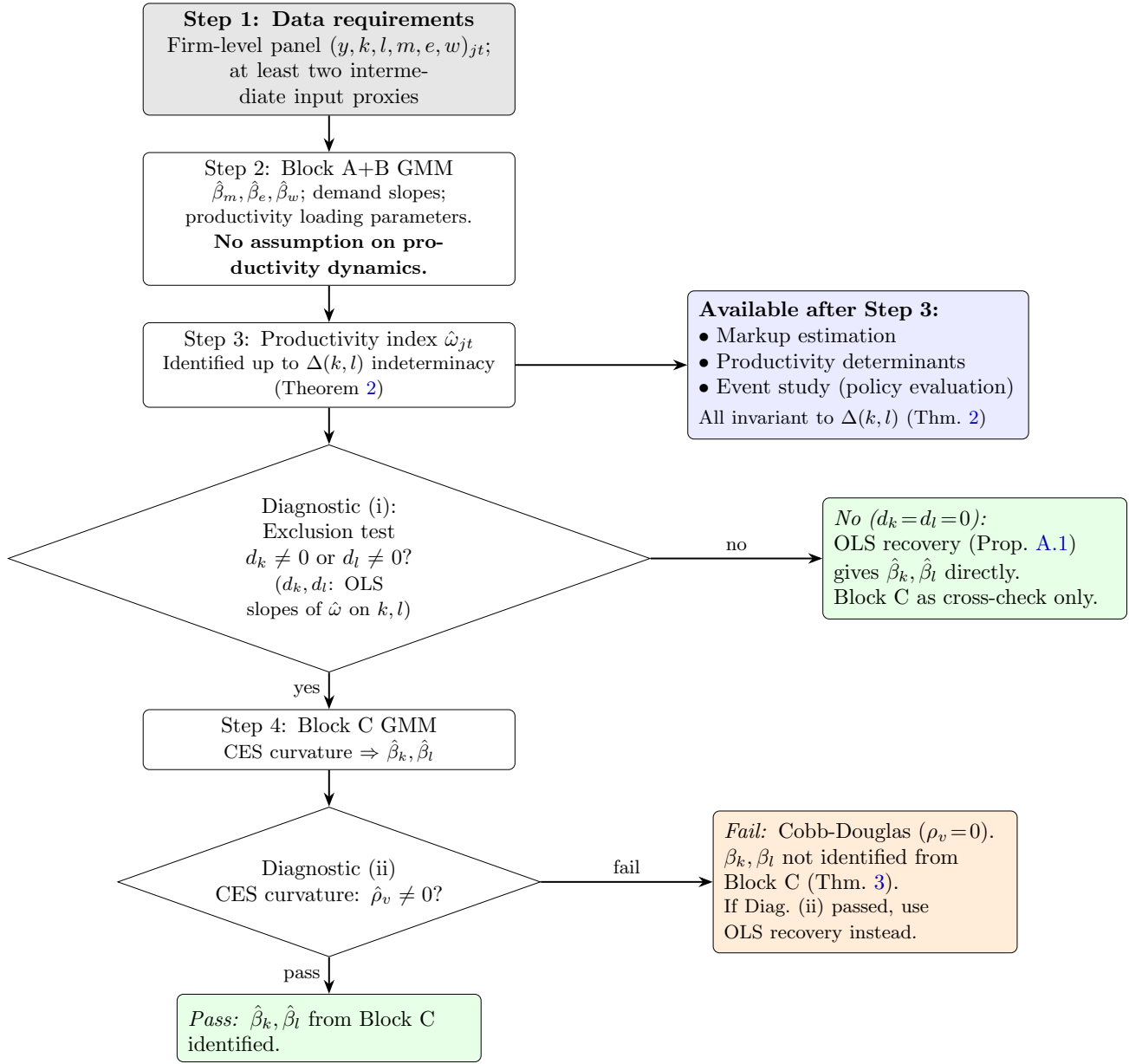

\end{appendices}

\end{document}